\documentclass[
reprint,
superscriptaddress,
amsmath,
amssymb,
aps,
pra
]{revtex4-2}

\PassOptionsToPackage{colorlinks=true,citecolor=blue,filecolor=blue,linkcolor=blue,urlcolor=blue}{hyperref}

\usepackage{orcidlink}
\usepackage{amsmath,amssymb,amsthm}
\newtheorem{proposition}{Proposition}

\usepackage{graphicx}
\usepackage{dcolumn}
\usepackage{bm}
\usepackage{physics}
\usepackage{svg}
\usepackage{todonotes}
\usepackage{csquotes}
\usepackage{hyperref}
\usepackage{dsfont}
\usepackage[capitalize]{cleveref}
\usepackage{placeins}

\newtheorem{theorem}{Theorem}

\newcommand{\aqa}{$\langle aQa ^L\rangle $ Applied Quantum Algorithms, Universiteit Leiden}
\newcommand{\lorentz}{Instituut-Lorentz, Universiteit Leiden, Niels Bohrweg 2, 2333 CA Leiden, Netherlands}

\newcommand{\huji}{Racah Institute of Physics, The Hebrew University of Jerusalem, Givat Ram, Jerusalem 91904, Israel}
\newcommand{\ulm}{Institute for Complex Quantum Systems, Ulm University, 89069 Ulm, Germany}
\newcommand{\iqst}{Center for Integrated Quantum Science and Technology (IQST), Ulm-Stuttgart, Germany}

\begin{document}
\title{Two-Sided Bounds on Ground State Properties in Lattice Gauge Theories}

\author{Elo\"ic Vall\'ee\,\orcidlink{0009-0005-2513-152X}}
\affiliation{\aqa}
\affiliation{\lorentz}

\author{Ariel Kelman\,\orcidlink{0000-0002-7710-6538}}
\affiliation{\huji}

\author{Julius Mildenberger\,\orcidlink{0000-0002-8016-8182}}
\affiliation{\aqa}
\affiliation{\lorentz}

\author{Jordi Tura\,\orcidlink{0000-0002-6123-1422}}
\affiliation{\aqa}
\affiliation{\lorentz}

\author{Patrick Emonts\,\orcidlink{0000-0002-7274-4071}}
\affiliation{\aqa}
\affiliation{\lorentz}
\affiliation{\ulm}
\affiliation{\iqst}

\begin{abstract}
Gauge theories describe fundamental interactions in the standard model of particle physics and build effective theories in condensed matter physics. 
Being notoriously hard to simulate, these theories are commonly regularized as lattice gauge theories which can be  evaluated with high precision through Monte Carlo algorithms or variational methods. 
However, Monte Carlo algorithms are not applicable in all regimes due to the sign problem; while variational methods, quantum and classical alike, only give upper-bounds on the ground-state energy.
Their precision crucially depends on the chosen ansatz.
Here, we present a framework based on a hierarchy of semidefinite programs to yield increasingly good lower-bounds on the ground-state energy and  propagate these bounds to extended observables like Wilson loops and mesonic strings.
We demonstrate the numerical capabilities by applying the algorithm to a (1+1)-dimensional $\mathbb{Z}_2$ theory with dynamic fermionic matter and its pure-gauge version in (2+1)-dimensions.
In both cases, we obtain certified intervals for the ground-state energy. 
For a one-dimensional system with 129 sites, the certified interval has a relative spread of $0.04\%$. 
For a two-dimensional system with 97 sites (corresponding to a $7 \times 8$ lattice), the relative spread is $2.08\%$.
Beyond energy estimates, the proposed method provides rigorous bounds on both short- and long-range observables such as Wilson loops over multiple plaquettes and the magnitude of the mesonic string.
The method can be both used as a computational microscope into many-body physics and a way to certify results of quantum simulators.
\end{abstract}

\date{\today}
\maketitle

\section{Introduction}

Gauge theories are central to diverse fields of physics. 
In high-energy physics, they are fundamental to the Standard model of particle physics~\cite{peskin_introduction_1995,weinberg_quantum_1995,weinberg_making_2004}, they emerge as effective descriptions in condensed matter physics~\cite{wen_quantum_2004,balents_spin_2010,savary_quantum_2016,fradkin_field_2013,affleck_large-n_1988,wen_mean-field_1991,lee_doping_2006}, and they are further related to paradigms in quantum computing~\cite{kitaev_fault-tolerant_2003,spagnoli_fault-tolerant_2026}.
The defining characteristic of gauge theories is the invariance under local transformations parametrized by a group. 
This gives rise to local conservation laws which are intricately linked to a rich phenomenology, with examples including Gauss’s law in quantum electrodynamics or color confinement in quantum chromodynamics~\cite{wilson_confinement_1974,hooft_magnetic_1974,nambu_strings_1974,polyakov_compact_1975,polyakov_quark_1977,gross_50_2023}.

Lattice gauge theories, which involve the spatial discretization of gauge theories enabling regularization of ultraviolet divergences, were originally suggested as a means to investigate non-perturbative phenomena in the strong-coupling regime~\cite{wilson_confinement_1974}.
This allowed for the treatment via classical Monte Carlo methods~\cite{creutz_monte_1983}, which has been very fruitful in, e.g., studying static properties such as the hadronic spectrum or equilibrium properties of the quark-gluon plasma~\cite{gattringer_quantum_2009,aoki_flag_2026}. 
However, the utilization of Euclidean spacetime prohibits real-time evolution and models at finite densities may suffer from the sign-problem~\cite{troyer_computational_2005}.
Working within a Hamiltonian formulation can address these challenges~\cite{kogut_hamiltonian_1975,banuls_simulating_2020}, both for classical~\cite{rico_tensor_2014,banuls_efficient_2017,magnifico_lattice_2021,bender_variational_2023} or quantum simulation~\cite{byrnes_simulating_2006, dalmonte_lattice_2016, zohar_quantum_2015, martinez_real-time_2016,aidelsburger_cold_2021,schweizer_floquet_2019,zohar_quantum_2021,bauer_quantum_2023,klco_standard_2022,barata_quantum_2023,mildenberger_confinement_2025,di_meglio_quantum_2024,farrell_scalable_2024,ciavarella_quantum_2024,cochran_visualizing_2025,gonzalez-cuadra_observation_2025,cobos_real-time_2025}.

Within the Hamiltonian approach, tensor network (TN) methods have emerged as one of the most successful approaches for approximating ground states of lattice gauge theories in recent years~\cite{buyens_matrix_2014,buyens_tensor_2015,banuls_tensor_2018,silvi_tensor_2019,emonts_variational_2020,felser_efficient_2021,cataldi_hamiltonian_2025,xu_string_2025,magnifico_tensor_2025,wu_accurate_2025}.
As a key feature, these states efficiently capture entanglement in low-dimensional systems while respecting, at least approximately, local constraints.
These techniques have been highly effective in practice; however, they are fundamentally variational in nature, i.e. they only yield upper bounds on the ground-state energy and do not provide a rigorous certificate of convergence.

A complementary approach has recently gained attention in many-body physics based on semidefinite programming (SDP) relaxations~\cite{boyd_convex_2004}.
These methods relax the original quantum many-body problem into a convex optimization problem over operator expectation values to  provide provable lower bounds on the ground-state energy.
These methods have been explored in different contexts, ranging from quantum chemistry~\cite{mazziotti_first-order_2004,zhao_reduced_2004,nakata_variational_2001,mazziotti_two-electron_2012,scheer_bootstrapping_2026} to condensed matter systems~\cite{baumgratz_lower_2012,barthel_solving_2012,fawzi_certified_2024}.
Under the name of conformal bootstrapping, similar methods have been developed in the high-energy community~\cite{berenstein_anomalous_2022, berenstein_semidefinite_2023, berenstein_bootstrapping_2022, berenstein_bootstrapping_2021,lawrence_semidefinite_2023,li_bootstrapping_2024}.
The main focus of these investigations have been few-body Hamiltonians or systems in the action formalism.
Recently, structured SDP relaxations based on the Navascues-Pironio-Acin (NPA) hierarchy~\cite{navascues_bounding_2007,navascues_convergent_2008}, a non-commuting version of the Lasserre hierarchy~\cite{lasserre_convergent_2006}, have been applied to the genuine many-body setting.
These works do not only bound the ground-state energy, but propagate the bounds to single-body observables and few-body correlators~\cite{wang_certifying_2024,wang_scalable_2026}.
Additionally, these relaxations can be applied to open quantum systems~\cite{mortimer_certifying_2025}, have been supplemented by measurement data~\cite{mortimer_bounding_2026} and been used to explore phase diagrams~\cite{jansen_mapping_2026}.
  
In this work, we leverage SDP relaxations in conjunction with TN methods to obtain certified results for ground-state energies as well as key observables, ranging from single-body observables to extended Wilson loops in $\mathbb{Z}_2$ lattice gauge theories, both in one and two spatial dimensions.
Such rigorous two-sided bounds on observables can not only enable systematic investigations of quantum-many-body physics, they can also serve as a tool of device certification for quantum simulators.
In particular, we obtain these bounds by considering the NPA hierarchy up to level three with a limited interaction range, exploiting the sparsity of the Hamiltonian.
We obtain physically relevant bounds by restricting the feasible set of the optimization problem to the physical sector of the gauge theory and formulating the problem in a gauge adapted, block-diagonal basis.
Extended Hermitian operators (up to eight-body terms) are bounded by adding specific monomials to enhance expressivity.
Additionally, we provide bounds, up to mild continuity assumptions, for the modulus of non-Hermitian mesonic string operators, central in characterizing confinement.

The remainder of this work is structured as follows:
\Cref{sec:sdp} introduces the necessary methods to find upper variational and lower relaxation-based bounds on the energy.
In the next step, these bounds are propagated to observables. 
In \cref{sec:z2_1d}, we demonstrate these techniques for a $\mathbb{Z}_2$ lattice gauge theory in one spatial dimension coupled to matter. 
Furthermore, we present results for a two-dimensional pure gauge $\mathbb{Z}_2$ theory in \cref{sec:z2_2d}.
Finally, we conclude in \cref{sec:conclusion}.

\section{Bounding Observables}
\label{sec:sdp}

\begin{figure*}
    \centering
    \includegraphics[width=\textwidth]{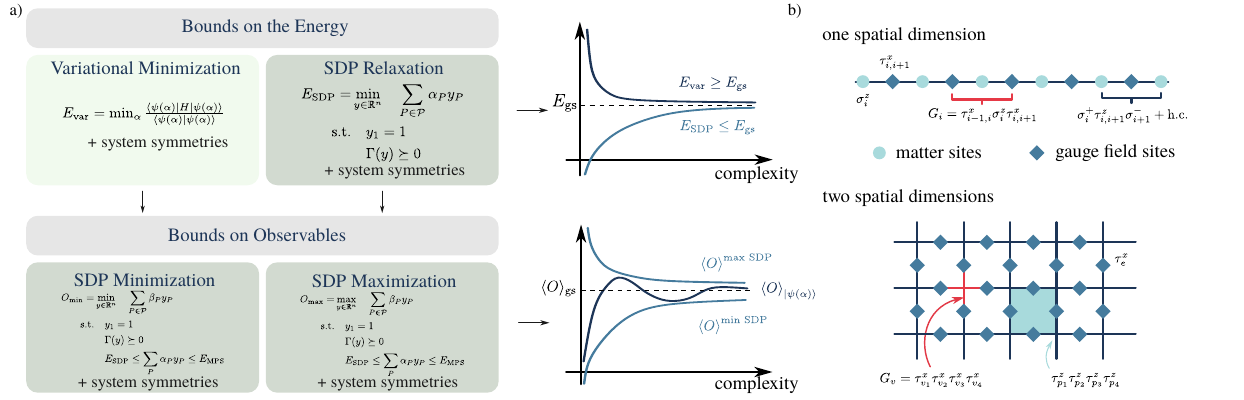}
    \caption{(a) Pipeline of the algorithm. 
    The ground-state energy is bounded from above with a variational ansatz, and from below with an SDP-based relaxation, effectively bracketing the exact value.
    These bounds are used as additional constraints to bound observables with further SDP relaxations.
    The SDP relaxations rely on the Pauli decomposition of the Hamiltonian and the observable: $H = \sum_P \alpha_P P$ and $O = \sum_P \beta_P P$, where $P$ are Pauli strings.
    (b) Illustration of the two systems considered in this work: (top) a $\mathbb{Z}_2$ lattice gauge theory in one spatial dimension with fermionic matter on the vertices and gauge fields on the links. (bottom) the equivalent pure-gauge theory in two spatial dimensions, including a plaquette term.
    In both illustrations, the symmetry generator is marked in red. 
    The different terms of the Hamiltonian are indicated by operators.}
    \label{fig:bounding_groundstate}
\end{figure*}

To find upper and lower bounds on observables, we have to first upper and lower bound the ground-state energy (see~\cref{fig:bounding_groundstate}).
The upper bound is provided by matrix product states (MPS), a variational tensor network approach.
This choice is not fundamental; any other suitable variational model could be employed instead.
To compute the lower bound, we relax the ground state problem via SDP-based methods such as the NPA hierarchy.
Both methods are explained in the following section.
Using further SDPs, we can propagate the energy bounds on the ground-state energy to observables~\cite{wang_certifying_2024}.

\subsection{Optimization problems}

We start by treating the ground state problem as a general non-commutative optimization problem and explore different ways to find either exact or approximate solutions.
In the following, we assume that the Hamiltonian can be expressed as a combination of Pauli strings, i.e. terms acting on two-level systems.
This restriction can easily be lifted, but simplifies the following discussion.

\subsubsection{Ground-state energy problem}
The ground-state energy is obtained by solving 
\begin{equation}
\label{eq:gs_problem}
E_{\mathrm{gs}} = \min_{\rho \in \mathcal{S}} \Tr(H \rho).
\end{equation}
with the Hamiltonian $H := \sum_i \alpha_i P_i$, where $\alpha_i \in \mathbb{R}$ and $P_i$ is a Pauli string of the form $P_i = \sigma_{i_1}^{x} \otimes \sigma_{i_2}^{z} \otimes \sigma_{i_3}^{y}\cdots$.

Such optimization problems can be solved exactly via exact diagonalization when the system size is sufficiently small (typically up to $\sim 20$ qubits). 
For larger systems, variational methods like tensor networks are employed to obtain an approximate solution~\cite{cirac_matrix_2021}.
While these techniques have proven highly successful in practice and have been extensively studied, they generally do not provide rigorous guarantees on the accuracy of the resulting approximation. 
In particular, when used as variational ansätze, MPS methods guarantee only that the estimated energy constitutes an upper bound on the true ground-state energy,
\begin{equation} \label{eq:gs_energy_mps}
E_{\mathrm{gs}} \leq E_{\mathrm{gs}}^{\mathrm{MPS}} := \min_{\rho \in \mathcal{S}_{\mathrm{MPS}}} \mathrm{Tr}(H\rho),
\end{equation}
where $\mathcal{S}_{\mathrm{MPS}}$ denotes the set of states representable as MPS.
This set is a subset of the full state space $\mathcal{S}$, i.e. $\mathcal{S}_{\mathrm{MPS}} \subseteq \mathcal{S}$. 
Increasing the bond dimension enlarges $\mathcal{S}_{\mathrm{MPS}}$, thereby improving the variational approximation and yielding progressively tighter upper bounds on the ground-state energy.
The DMRG algorithm performs this variational optimization efficiently within the MPS manifold by iteratively updating local tensors while controlling the bond dimension~\cite{white_density_1992,schollwock_density-matrix_2011}. 
However, the optimization is not globally convex, and DMRG is therefore not guaranteed to find the minimum over $\mathcal{S}_{\mathrm{MPS}}$.
Consequently, the obtained energy may remain above the optimal MPS variational energy, in addition to the intrinsic error arising from the restriction to a finite bond dimension.

Equivalently, the ground-state energy problem can be formulated as a polynomial optimization problem over the non-commutative Pauli algebra $\mathcal{A}$. 
This formulation treats the Pauli operators as non-commuting variables subject to the defining relations of $\mathcal{A}$, namely $(\sigma^x)^2 = 1$, $(\sigma^x)^{\dagger} = \sigma^x$, $\sigma^x \sigma^y = i\sigma^z$, and so forth. 
The resulting optimization problem is then solved subject to the constraint that the variables satisfy these algebraic relations.
This is a specific instance of non-commutative polynomial optimization (NPO), a widely studied field in mathematics~\cite{pironio_convergent_2010,burgdorf_optimization_2016}.

Specifically, we introduce a variable for each Pauli-string expectation value,
\begin{equation} \label{eq:def_y_variables}
    y_i = \Tr(P_i\rho) .
\end{equation}
Note that because quantum states are normalized, the variable associated with the identity is necessary equal to one, i.e. $y_{\mathbb{I}} := \Tr(\mathbb{I}\rho) = 1$.
We denote $\boldsymbol{y}$ as the vector of expectation values and we define the set of all possible expectation values as 
\begin{equation} \label{eq:algebra_constraint}
    \mathcal{Y} = \{\boldsymbol{y} : y_i = \Tr(P_i \rho) \; , \; \forall \rho \in \mathcal{S}\} .
\end{equation}
The optimization problem in \cref{eq:gs_problem} is equivalent to the NPO problem
\begin{equation} \label{eq:gs_problem_as_npo}
    E_{\mathrm{gs}} = \min_{\boldsymbol{y} \in \mathcal{Y}} \sum_i \alpha_i y_i.
\end{equation}
The constraint resulting from \cref{eq:algebra_constraint} is in general hard to solve, but there exist a formulation, known as SDP relaxation, that yields a \emph{relaxed} version of the problem that is solvable.
The term ``relaxed'' means that the feasible set becomes larger, i.e. $\mathcal{Y}_{\mathrm{SDP}} \supseteq \mathcal{Y}$.
Consequently, the resulting SDP is a convex optimization problem, which enables efficient computation of globally optimal solutions and provides certificates of optimality for the obtained solutions.

The optimal value of the relaxed problem provides a lower bound on the original ground-state energy,
\begin{equation}
E_{\mathrm{gs}} \geq E_{\mathrm{gs}}^{\mathrm{SDP}} := \min_{\boldsymbol{y} \in \mathcal{Y}_{\mathrm{SDP}}} \sum_i \alpha_i y_i,
\end{equation}
since enlarging the feasible set can only decrease (or leave unchanged) the optimal value of a minimization problem.
For the rest of this work, we will denote $E_{\mathrm{gs}}^{\mathrm{MPS}} \rightarrow E_{\max}$ and $E_{\mathrm{gs}}^{\mathrm{SDP}} \rightarrow E_{\min}$.

\subsubsection{Bounding observables} \label{subsubsec:bounding_observables}

A similar approach can be applied to bound observables' expectation values.
Consider the expectation value of an observable $O := \sum_i \beta_i P_i$ within the ground-state manifold, denoted by $\expval{O}_{\mathrm{gs}}$.
When the ground state is non-degenerate, this quantity can be estimated using variational methods such as MPS:
\begin{equation}
    \expval{O}_{\mathrm{gs}}^{\mathrm{MPS}} = \Tr(O \rho_{\mathrm{MPS}})
\end{equation}
where $\rho_{\mathrm{MPS}}$ is the solution of \cref{eq:gs_energy_mps}.
However, in the presence of ground-state degeneracy, such an estimate is generally not well-defined, as different ground states may yield different expectation values.

What remains well-defined in all cases is the range of possible expectation values attained within the ground-state manifold.
These bounds can be expressed as optimization problems.
For example, the lower bound is given by
\begin{equation} \label{eq:obs_problem}
\expval{O}_{\mathrm{gs}}^{\min} :=  \min_{\rho \in \mathcal{S}} \Tr(O \rho) \quad \text{s.t.} \quad \Tr(H \rho) = E_{\mathrm{gs}} ,
\end{equation}
while the upper bound $\expval{O}_{\mathrm{gs}}^{\max}$ is given by a maximization.
In particular, if the ground-state energy is non-degenerate, these bounds coincide, i.e., $\expval{O}_{\mathrm{gs}}^{\min} = \expval{O}_{\mathrm{gs}}^{\max} = \expval{O}_{\mathrm{gs}}$, 
since the ground state is the unique state satisfying the constraint ($\Tr(H\rho)=E_{\mathrm{gs}}$).

In practice, the exact ground-state energy is often unknown. 
However, as discussed previously, lower bounds can be obtained from SDP relaxations, while upper bounds can be computed using variational ansätze. 
In this setting, bounds on the expectation value of the observable can be obtained by solving
\begin{equation} \label{eq:obs_problem2}
\expval{O}_{\sim\mathrm{gs}}^{\min} :=  \min_{\rho \in \mathcal{S}} \Tr(O \rho) \quad \text{s.t.} \quad E_{\mathrm{gs}}^{\mathrm{SDP}} \leq \Tr(H \rho) \leq E_{\mathrm{gs}}^{\mathrm{MPS}}, 
\end{equation}
where $\expval{O}_{\sim\mathrm{gs}}^{\min}$ is the minimal expectation value over all states yielding an energy close to the ground-state energy.
Since the constraints in \cref{eq:obs_problem2} are looser than the ones in \cref{eq:obs_problem}, we have $\expval{O}_{\sim\mathrm{gs}}^{\min} \leq \expval{O}_{\mathrm{gs}}^{\min}$.

Optimization problems such as those in \cref{eq:obs_problem,eq:obs_problem2} also belong to the class of NPO problems. 
In particular, \cref{eq:obs_problem2} can be equivalently formulated as
\begin{equation} \label{eq:obs_problem2_as_npo}
    \expval{O}_{\sim\mathrm{gs}}^{\min} = \min_{\boldsymbol{y} \in \mathcal{Y}} \sum_{i} \beta_i y_i \quad \text{s.t.} \quad E_{\mathrm{gs}}^{\mathrm{SDP}} \leq \sum_i \alpha_i y_i \leq E_{\mathrm{gs}}^{\mathrm{MPS}}.
\end{equation}
As in the ground-state energy case, these NPO problems admit SDP relaxations. 
The resulting relaxed optimization problems provide lower and upper bounds on the optimal values, yielding
\begin{equation}
    \begin{aligned}
        \expval{O}_{\sim\mathrm{gs}}^{\min \; \mathrm{SDP}} \leq \expval{O}_{\sim\mathrm{gs}}^{\min} \leq \expval{O}_{\mathrm{gs}}^{\min} \\
        \expval{O}_{\mathrm{gs}}^{\max} \leq \expval{O}_{\sim\mathrm{gs}}^{\max} \leq \expval{O}_{\sim\mathrm{gs}}^{\max \; \mathrm{SDP}}
    \end{aligned}
\end{equation}
where $\expval{O}_{\sim\mathrm{gs}}^{\min \; \mathrm{SDP}}$ and $\expval{O}_{\sim\mathrm{gs}}^{\max \; \mathrm{SDP}}$ denote the optimal values of the SDP relaxations of \cref{eq:obs_problem2_as_npo}. 
In the rest of this work, we will simplify the notation: $\expval{O}_{\sim\mathrm{gs}}^{\min \; \mathrm{SDP}} \rightarrow \expval{O}_{\min}$ and $\expval{O}_{\sim\mathrm{gs}}^{\max \; \mathrm{SDP}} \rightarrow \expval{O}_{\max}$.

\subsection{SDP relaxation}
In this work, we consider non-commutative polynomial optimization (NPO) problems of the general form
\begin{equation} \label{eq:npo_problem}
\begin{aligned}
    \beta_{\max} = \max_{\boldsymbol{y} \in \mathcal{Y}} \quad & \sum_i b_i y_i \\
    \text{ s.t. } \quad & \sum_i a_{ij} \; y_i \geq 0 , \; \forall_{j = 1 \dots \ell}
\end{aligned}
\end{equation}
where $\mathcal{Y} = \{\boldsymbol{y} : y_i = \Tr(P_i \rho) \; , \; \forall \rho \in \mathcal{S}\}$ as defined in \cref{eq:algebra_constraint}, $b_i$ are the coefficients of the objective function and $a_{ik} \in \mathbb{R}$ are the coefficients of the $k$-th inequality constraints.
Note that equality constraints can be easily included by considering two inequality constraints.

The general NPO problem is, in general, non-convex. 
It can be relaxed to a semidefinite programming (SDP) problem, which is a convex optimization problem. 
The resulting SDP can be expressed as follows:
\begin{equation} \label{eq:npo_problem_sdp_relaxation}
\begin{aligned}
    \max_{\boldsymbol{y} \in \mathbb{R}^n} & \sum_{i} b_i y_i \\
    \text{s.t.} \quad & y_{\mathbb{I}} = 1 \\
    & \sum_{i} a_{ij} y_i \geq 0 , \;  \forall_{j = 1 \dots \ell} \\
    & \Gamma(\boldsymbol{y}) \succeq 0
\end{aligned}
\end{equation}
where $\Gamma(\boldsymbol{y}) = $ is the moment matrix~\cite{lasserre_global_2001,navascues_convergent_2008}, whose construction is described in more detail later.

This problem constitutes a relaxation of the original formulation:  
the solution set of \cref{eq:npo_problem} is contained within the solution set of \cref{eq:npo_problem_sdp_relaxation}.  
The relaxed problem is particularly convenient, as it takes the form of a semidefinite program (SDP).
The \emph{feasible set}, denoted $\mathcal{F}$, is given by the set of $\boldsymbol{y} \in \mathbb{R}^n$ that satisfy the inequality constraints and the positive semidefinite (PSD) constraint.
The function that is maximized is referred as \emph{objective function}.
Such SDP formulations have been extensively studied in the optimization and control communities~\cite{boyd_convex_2004,blekherman_semidefinite_2012}  and a variety of numerical solvers exist, such as \texttt{MOSEK}~\cite{noauthor_mosek_nodate}, which is based on interior-point methods. 
These methods have polynomial-time complexity in the problem dimensions, namely, in the dimension of the PSD matrix and the number of constraints, as well as in the required numerical precision~\cite{nesterov_lectures_2018}.

The moment matrix is a PSD matrix whose entries are linear combinations of variables, i.e. 
\begin{equation}
    \Gamma(\boldsymbol{y}) = \sum_i A^{(i)} y_i \; .
\end{equation}
where the matrices $\{A^{(i)}\}_i$ are $m \times m$ Hermitian matrices.
The moment matrix possesses a highly structured form: it is generated from a sequence $\boldsymbol{s}$ of $m$ operators, with entries defined as 
$\Gamma_{ij} := \Tr\phantom{}(s_i\, s_j^{\dagger} \rho)$.
A detailed example is provided in \cref{app_sdp_relaxation_example}.
By construction, the moment matrix must be positive semidefinite, and this requirement induces a set of linear constraints on products of variables that partially reproduce the relations of the Pauli algebra $\mathcal{A}$. 
The sequence must be chosen such that all the variables appearing in the objective and inequality constraints are also present in the moment matrix, otherwise the problem might end up being unbounded, and thus unsolvable.

The choice of the sequence $\boldsymbol{s}$ plays a crucial role in ensuring that the problem remains computationally manageable, 
as it directly determines the size of the PSD variables. 
Naturally, this choice is not unique. 
A standard strategy is to select a sequence from the NPA hierarchy~\cite{navascues_convergent_2008}, 
which provides a nested family of sequences that progressively converges to the exact solution. 
At hierarchy level~$l$, the sequence contains all Pauli strings with Hamming weight up to~$l$. 
We denote by $\beta^{(l)}_{\max}$ the solution of the corresponding SDP relaxation. 
By construction, these upper-bounds satisfy 
$\beta^{(1)}_{\max} \geq \beta^{(2)}_{\max} \geq \dots \geq \beta_{\max}$. 

Although the NPA hierarchy is guaranteed to converge to the optimal quantum value in the asymptotic limit of infinitely many levels, alternative sequences may provide tighter bounds at comparable levels.
The overarching objective is therefore to identify a sequence that captures as many physically relevant constraints as possible, 
while remaining sufficiently compact to keep the resulting SDP efficiently solvable. 
To achieve this, one can exploit the underlying structure of the physical problem, such as symmetries, the spatial arrangement of the particles, or other problem-specific features~\cite{wang_certifying_2024}.

\subsubsection{SDP relaxation with interaction-range restrictions}
\label{sec:interaction_range}

Reducing the size of the moment matrix provides a way to simplify an SDP relaxation.
In this work, we consider systems on a lattice whose Hamiltonian contains terms involving geometrically local spin operators. 
It is therefore natural to consider an SDP relaxation whose moment matrix is generated by a sequence of Pauli strings restricted to a chosen interaction range $r$. 
As an example, we explicitly give the sequence at level $3$ and interaction range $r$ for $N$ sites:
\begin{equation} \label{eq:NPA_sequence}
\begin{aligned}
    \boldsymbol{s} = \quad & \left\{ \mathbb{I} \right\} \\
    \cup & \left\{\sigma_{i}^{k} \vert i \in \{1 \dots N\} \; , \; k \in \{x,y,z\} \right\} \\
    \cup & \left\{ \sigma_{i}^{k} \sigma_{j}^{l} \vert 0 < \operatorname{dist}(i,j) \leq r \; , \; k,l \in \{x,y,z\} \right\} \\
    \cup & \left\{ \sigma_{i_1}^{k_1} \sigma_{i_2}^{k_2} \sigma_{i_3}^{k_3} \vert 0 < \operatorname{dist}(i_{\alpha},i_{\beta}) \leq r \; , \; k_{\alpha} \in \{x,y,z\} \right\}
\end{aligned}
\end{equation}
where $\operatorname{dist}(i,j)$ denotes the distance between site $i$ and site $j$.
The first (respectively, second, third, and fourth) line of \cref{eq:NPA_sequence} corresponds to the zero- (respectively, one-, two-, and three-)body term, and thus to the zeroth, first, second, and third levels of the hierarchy.
The distance constraint ensures that only monomials within the interaction range are included in the sequence.
We denote by $\beta_{\max}^{(l,r)}$ the solution of the SDP relaxation at hierarchy level $l$ and interaction range $r$. 
By the same argument as before, these bounds satisfy
$\beta_{\max}^{(l,1)} \geq \beta_{\max}^{(l,2)} \geq \dots \geq \beta^{(l)}_{\max}$ for all hierarchy levels~$l$. 

This simplification does not improve the bound, as it further relaxes the constraints, but it substantially reduces the size of the moment matrix $\Gamma$. 
As a result, the optimization problem can be solved for larger system sizes and higher hierarchy levels. 
Thus, this simplification highlights the trade-off between the tightness of the bound and the computational cost of the optimization.

\subsubsection{SDP relaxation with symmetries} \label{subsubsec:sdp_with_symmetries}

When the objective and the inequality constraints satisfy a set of symmetries, the SDP relaxation problem can be further simplified by decomposing the PSD constraint into smaller blocks.
For this we have to introduce the concept of invariant SDP and symmetry adapted basis.

Consider a transformation on the Pauli algebra
\begin{equation}
\begin{aligned}
    g:\; & \mathcal{A} \rightarrow \mathcal{A} \\
         & P \mapsto G P G^{\dagger},
\end{aligned}
\end{equation}
where $P$ is a Pauli string and $G$ is a unitary operator. 
By construction, $g$ satisfies: (i) linearity, $g(aP + bQ) = a g(P) + b g(Q)$ for $a,b \in \mathbb{C}$; (ii) multiplicativity, $g(PQ) = g(P) g(Q)$ for $P,Q \in \mathcal{A}$; and (iii) $*$-preservation, $g(P^{\dagger}) = g(P)^{\dagger}$.
Note that, the third condition is trivially satisfied, since Pauli operators are Hermitian, i.e. $P^{\dagger} = P$.
Such transformations are called $*$-automorphisms, and they naturally describe symmetry transformations in quantum mechanics~\cite{bratteli_operator_1987}.

The action of $g$ on the variables defined in \cref{eq:def_y_variables} is
\begin{equation}
    \mathfrak{g}(y) := \Tr(g(P) \rho).
\end{equation}
Note the font difference to emphasize that $g$ and $\mathfrak{g}$ correspond to the same transformation, but on a different space ($g$ acts on the Pauli algebra $\mathcal{A}$, while $\mathfrak{g}$ acts on variables in $\mathbb{R}$); $\mathfrak{g}$ is also a $*$-automorphism.
For example, the ``$\sigma_1^x$-transformation'' acting on the first qubit of a Pauli string as $g(P) = \sigma_1^x P \sigma_1^x$ induces a transformation on the variable given by $y = \Tr(P\rho)$ as $\mathfrak{g}(y) = \mathrm{Tr}(\sigma_1^x P \sigma_1^x \rho)$.

The SDP in \cref{eq:npo_problem_sdp_relaxation} is said to be \emph{invariant} under $g$ if both its objective function and its feasible set remain unchanged~\cite{gatermann_symmetry_2004}. 
In other words,
\begin{equation} \label{eq:invariant_sdp}
    \begin{aligned}
    & \sum_i b_i\, \mathfrak{g}(y_i) = \sum_i b_i\, y_i, \\
    & \mathfrak{g}(\boldsymbol{y}) \in \mathcal{F} , \quad \forall\, {\boldsymbol{y} \in \mathcal{F}} \; ,
\end{aligned}
\end{equation}
where the transformation on a vector of variables extends naturally to each component separately, i.e. $\mathfrak{g}(\boldsymbol{y}) := (\mathfrak{g}(y_1), \dots, \mathfrak{g}(y_n))^T$. 
More generally, an SDP is said to be invariant (or symmetric) under a group $\mathcal{G}$ if the above conditions hold for all $g \in \mathcal{G}$.
Here and in what follows, we make a slight abuse of notation, directly identifying group elements with their action on the Pauli algebra. It will remain clear from context what is meant.

\begin{theorem} \label{thm:invariant_sdp}
If an SDP is invariant under a group $\mathcal{G}$, we can restrict the feasible set $\mathcal{F}$ to the set of solution that are symmetric
\begin{equation} \label{eq:symmetric_feas_set}
    \tilde{\mathcal{F}} := \mathcal{F} \cap \{ \boldsymbol{y} \vert \mathfrak{g}(\boldsymbol{y}) = \boldsymbol{y} , \,  \mathfrak{g} \in \mathcal{G} \} \quad .
\end{equation}
\end{theorem} 

\begin{proof}
If $\boldsymbol{y}$ is solution, then $\mathfrak{g}(\boldsymbol{y})$ is also solution, because of the SDP invariance. 
And thus $\overline{\boldsymbol{y}} = \frac{1}{|\mathcal{G}|} \sum_{\mathfrak{g} \in \mathcal{G}} \mathfrak{g}(\boldsymbol{y})$ is also solution.
Since $\mathfrak{g}(\overline{\boldsymbol{y}}) = \overline{\boldsymbol{y}}$, we can say without loss of generality that there exists a solution satisfying the constraints in \cref{eq:symmetric_feas_set}.
\end{proof}

In the specific case of SDP relaxations, it is possible to show that the PSD constraint, i.e. the positivity of the moment matrix, is preserved under the action of a group $\mathcal{G}$, if there exists a unitary representation $R:\mathcal{G} \rightarrow \mathbb{C}^{m \times m}$ such that 
\begin{equation} \label{eq:sequnce_transformation}
    g(\boldsymbol{s}) = R(g) \boldsymbol{s} .
\end{equation}
Again, the transformation on a vector of operators is naturally extended to each component individually, i.e. $g(\boldsymbol{s}) = (g(s_1) , \dots , g(s_m))^T$.
This corresponds to the transformation of the generating sequence $\boldsymbol{s}$ under $g \in \mathcal{G}$.
From representation theory, there exists a unitary matrix $P \in \mathbb{C}^{m \times m}$ such that $P^{\dagger} R(g) P$ is block-diagonal for all $g \in \mathcal{G}$, with block sizes $m_k$ satisfying $\sum_k m_k = m$ \cite{serre_linear_1977}. 
The matrix $P$ is known as the symmetry-adapted basis transformation.

It follows from \Cref{thm:invariant_sdp} and \cref{eq:sequnce_transformation} that, in the symmetry-adapted basis, the SDP relaxation is equivalent to
\begin{equation}
    \begin{aligned}
    \max_{\boldsymbol{y} \in \mathbb{R}^n} \quad & \sum_i b_i y_i \\
    \text{s.t.} \quad & y_{\mathbb{I}} = 1 \\
    & \sum_i a_{ij} y_i \geq 0 \quad \forall j \\
    & \sum_i \tilde{A}^{(i)}_k y_i \succeq 0 \quad \forall k
    \end{aligned}
\end{equation}
where $\tilde{A}^{(i)}_k$ denotes the $k$-th block of $\tilde{A}^{(i)} := P A^{(i)} P^{\dagger}$. See \cref{app:block-diag_sdp} for more details. 

This reformulation reduces the PSD constraints to smaller blocks of size $m_k \leq m$, thereby significantly reducing the computational complexity of the semidefinite program.

\section{The $\mathbb{Z}_2$ Kogut-Susskind model in (1+1)d}
\label{sec:z2_1d}

\subsection{Hamiltonian and symmetries}
In this work, we consider the Kogut-Susskind Hamiltonian~\cite{kogut_hamiltonian_1975} of a $\mathbb{Z}_2$ lattice gauge theory including dynamical fermionic matter.
For ease of description, we formulate the whole theory after a Jordan-Wigner transformation in terms of spin-$1/2$ particles~\cite{jordan_uber_1928}.
The Hamiltonian $H=H_J+H_f+H_m$ contains the three distinct terms
\begin{equation}\label{eq:hamiltonian}
\begin{split}
    H_J &= -J\sum_{i=1}^{N-1} (\sigma_i^+ \tau_{i,i+1}^z \sigma_{i+1}^- + \text{h.c.}),\\
    H_f &= -f\sum_{i=1}^{N-1} \tau_{i,i+1}^x,\\
    H_m &= \frac{\mu}{2}\sum_{i=1}^N (-1)^i\sigma_i^z.
\end{split}
\end{equation} 
The term $H_J$ is the gauge-matter interaction, $H_f$ is a background field and $H_m$ is the fermionic mass term.
Since the system is considered in $(1+1)$d, there is no self-interaction of the gauge field, i.e. the plaquette operator term is missing.
Unless stated otherwise, the numerical examples in this work use $J=f=\mu=1$.

The operators $\tau^{\alpha}_{i, i+1}$ and $\sigma^{\alpha}_i$ are both Pauli matrices. 
The $\sigma$'s act on the matter, while the $\tau$'s act on the gauge fields. 
The system size is specified by the number of matter sites $N$; thus the total number of spins is $2N-1$, where the $-1$ corresponds to the missing gauge field spin due to open boundary conditions.

The theory described by the Hamiltonian in \cref{eq:hamiltonian} exhibits a local $\mathbb{Z}_2$ gauge symmetry that is generated by the symmetry operators
\begin{equation} \begin{aligned}
\label{eq:gauge-1d-bulk}
    G_i = (-1)^{i+1}\tau_{i-1, i}^x \sigma_i^z \tau_{i, i+1}^x 
        \quad \forall i \in \{ 2\dots N-1 \}
\end{aligned} \end{equation}
where $N$ is the number of sites.
For the boundary sites, we have 
\begin{equation} \begin{aligned}
\label{eq:gauge-1d-boundary}
    G_1 &= \sigma_1^z \tau_{1, 2}^x \\
    G_N &= (-1)^{N+1}\tau_{N-1, N}^x \sigma_N^z \\
\end{aligned} \end{equation}
Unless stated otherwise, we focus on the physical sector where $G_i\ket{\psi}=+1\ket{\psi}\ \forall\ i = 1 \dots N$.
We enforce the resulting constraints in the variational computation and the semi-definite program.
By leaving out these constraints, we could also consider the corresponding unconstrained gauge theory~\cite{inacio_unconstrained_2026}.

\subsection{Generating sequence} \label{subsec:generating_sequence_1D}

In this work, we select the generating sequence up to level 3 with an interaction range of 1.
The sequence is explicitly given in \cref{eq:NPA_sequence}. 
This choice is motivated by the fact that it maintains a reasonable computational cost while providing robust results. 
Since the field operators act on the links, we treat them as located halfway between neighboring matter sites. 
Consequently, an interaction range of 1 effectively corresponds to a range of 2 when using integer indexing convention.

To simplify the SDP, we compute the symmetry-adapted basis $P$ as explained in \cref{subsubsec:sdp_with_symmetries}, using the group generated by the gauge operators, i.e. $\mathcal{G} = \langle G_1, \dots , G_N \rangle$. 
Additional symmetries, such as translation invariance or mirror symmetry, could also be exploited; we leave this extension to future work.

After computations of several examples with dedicated software such as \texttt{RepLAB}~\cite{rosset_replab_2021} and \texttt{RepnDecomp}~\cite{hymabaccus_repndecomp_2020}, we observe that the symmetry-adapted basis $P$ is a permutation matrix. 
Consequently, arranging the sequence in the correct order ensures that the moment matrix $\Gamma$ becomes block-diagonal. 
We refer to this specific arrangement as the \emph{block-order} of the sequence. 
To construct the block-ordered sequence, we must partition the monomials into disjoint subsequences $t_1, t_2, \dots$, such that each subsequence $t_i$ generates a single block. 
Notably, we observe that if $g(s_i s_j) = s_i s_j$ for all $g \in \mathcal{G}$, then the elements $s_i$ and $s_j$ must belong to the same subsequence.

The underlying intuition is as follows: since both the generating sequence and the gauge operators consist of Pauli strings, the action of a symmetry yields an eigenvalue of either $+1$ or $-1$ (i.e., $g(s_i s_j) = \pm s_i s_j$). 
However, because the SDP is invariant and we restrict our analysis to the symmetric subspace, we must satisfy $\mathfrak{g}(y) = y$. 
This implies $\pm \Tr(s_i s_j \rho) = \Tr(s_i s_j \rho)$. 
Therefore, whenever $g(s_i s_j) = -s_i s_j$, the associated variable $y$ must vanish ($y=0$). 
In such cases, $s_i$ and $s_j$ can be assigned to different subsequences, as their product---which would appear in an off-block-diagonal position---is zero.

In summary, the criterion can be reformulated as follows: we begin by assuming that all Pauli strings lie in distinct subsequences, and then merge subsequences containing $s_i$ and $s_j$ whenever 
\begin{equation} \label{eq:ordering_criterion}
    g(s_i s_j) = s_i s_j.
\end{equation}

Computations reveal that the blocks within the block-ordered sequence fall into five distinct classes. 
These blocks exhibit a high degree of structure: three types of blocks appear for each matter site, while one type is associated to each field site. 
The fifth and final class comprises all remaining elements not accounted for by the preceding four classes. 
A detailed definition of these classes is provided in~\cref{app:blocks_1D_model}.

\subsection{Bounds on the energy}\label{subsec:energy_bounds_1D}

For small system sizes, we compare ground-state energy estimates obtained using three methods: exact diagonalization (ED), matrix product states (MPS), and semidefinite programming (SDP) relaxations. 
These methods address different computational objectives. 
Exact diagonalization computes the exact ground-state energy and therefore serves as a reference whenever it is computationally feasible. 
The availability of the exact solution for small systems allows us to assess the quality of both bounds and to verify the correctness of the proposed SDP formulation before extending the comparison to larger system sizes, where exact diagonalization is no longer tractable.

MPS is a variational method that provides an upper bound on the ground-state energy, which we can use with the SDP relaxation  lower bound. 
In our case, we use \texttt{tenpy}~\cite{hauschild_tensor_2024} to build an MPS at varying bond dimensions, using DMRG to find the ground state. 
The gauge symmetry is enforced by adding a penalty term to the Hamiltonian, as specified in appendix~\ref{app:mps_convergence}.
Unless otherwise stated, all MPS simulations in (1+1)d were performed with a bond dimension of $D=80$. 
As demonstrated in \cref{app:mps_convergence}, this choice is sufficient to ensure convergence.

To obtain an SDP-based estimate of the ground-state energy, we consider a relaxation of the noncommutative polynomial optimization problem introduced in \cref{eq:gs_problem_as_npo}. 
The sequence used to construct the moment matrix is described in \cref{subsec:generating_sequence_1D}. 
To restrict the optimization to the physical gauge sector, we impose the additional constraints $\expval{G_i} = 1, \forall i=1, \dots, N$.
The resulting SDP relaxation is
\begin{equation} \label{eq:sdp_relaxation_energy_1D}
\begin{aligned}
    E_{\min} = \min_{\boldsymbol{y} \in \mathbb{R}^n} \quad & \expval{H} \\
    \text{s.t.} \quad & y_{\mathbb{I}} = 1 \\
    & \expval{G_i} = 1 \quad \forall i = 1 \dots N\\
    & \Gamma(\boldsymbol{y}) \succeq 0, 
\end{aligned}
\end{equation}
where $\expval{H}$ and $\expval{G_i}$ denote the expectation values of the Hamiltonian and the $i$-th gauge symmetry operator after decomposition in the Pauli basis. 
Each Pauli expectation value is replaced by its corresponding moment variable $y$, making both quantities linear functions of $\boldsymbol{y}$.
The equality constraints on $\expval{G_i}$ are a choice: they determine the super-selection sector that we are working in.

\begin{figure*}[t]
    \centering
    \includegraphics[width=0.47\linewidth]{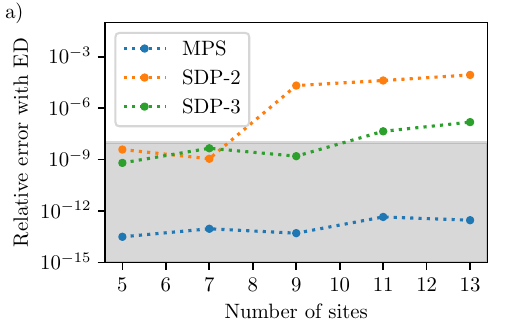}
    \hfill
    \includegraphics[width=0.47\linewidth]{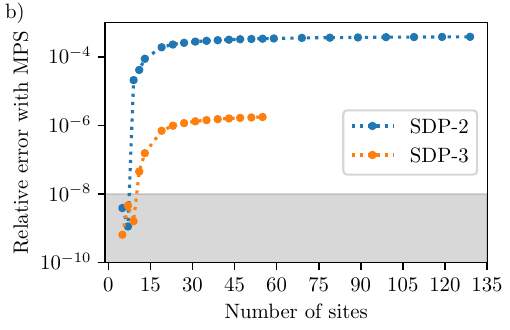} \\
    \caption{Relative error in the ground-state energy estimate of the Kogut-Susskind Hamiltonian with parameters $J=f=\mu=1$. 
    (a) For small system sizes, the reference ground-state energy is obtained by exact diagonalization. 
    (b) For larger system sizes, the reference is provided by MPS.
    The grey shaded region indicates errors below the guaranteed numerical precision of the SDP solver.}
    \label{fig:energies_1D}
\end{figure*}
The relative errors of the MPS approximation and of the level-2 and level-3 SDP relaxations with respect to the exact ground-state energy are shown in \cref{fig:energies_1D}~(a). 
The figure is limited to small system sizes, where exact diagonalization is still feasible. 
For larger systems, we instead report the relative error of the level-2 and level-3 SDP relaxations with respect to the MPS energy, as shown in \cref{fig:energies_1D}~(b).

As expected, the MPS energies are always above the exact values, while the SDP relaxations provide lower bounds (this distinction is not visible in \cref{fig:energies_1D}, since absolute values of the errors are plotted on a logarithmic scale).

Increasing either the relaxation level or the interaction range systematically improves the quality of the SDP lower bound.
Notably, increasing the level generally implies a larger generating sequence, whereas increasing the interaction range does not, and therefore the former typically leads to better bounds.
Furthermore, incorporating the extensive amount of gauge-sector constraints leads to a tighter relaxation than the unconstrained formulation, as the additional linear constraints reduce the feasible set of the SDP and therefore yield a more accurate approximation of the original optimization problem.

\subsection{Bounding few-body observables}

The main advantage of the SDP approach in comparison to variational methods is its ability to certify expectation values of observables. 
Bounds on such observables are obtained by formulating the corresponding NPO problem in \cref{eq:obs_problem2_as_npo}, which is then relaxed to an SDP.
As in the energy case, the sequence defining the moment matrix is constructed as described in \cref{subsec:generating_sequence_1D}. 
The block-diagonalization procedure applies only when the objective function is invariant under the relevant symmetries, i.e., when \cref{eq:invariant_sdp} is satisfied. 
Here, this condition holds for all gauge-invariant operators, which yield the non-vanishing operators in physical states~\cite{elitzur_impossibility_1975}.
In particular, \cref{eq:invariant_sdp} is fulfilled by the one-body observables $\sigma^z_i$, representing matter occupation on site $i$. 
Furthermore, we again impose sector constraints in order to restrict the optimization to the physically relevant subspace.
The resulting SDP for the lower bound takes the form:
\begin{equation} \label{eq:omin}
    \begin{aligned}
        \expval{\sigma^z_i}_{\min} = \min_{\boldsymbol{y} \in \mathbb{R}^n} \quad & \expval{\sigma^z_i} \\
        \text{s.t.} \quad & y_{\mathbb{I}} = 1 \\
        & \expval{G_k} = 1 \quad \forall k = 1 \dots N\\
        & E_{\max} - \expval{H} \geq 0 \\
        & \expval{H} - E_{\min} \geq 0 \\
        & \Gamma(\boldsymbol{y}) \succeq 0 
    \end{aligned}
\end{equation}
where $\expval{H}$ and $\expval{G_k}$ depend on $\boldsymbol{y}$.
Following the idea proposed in \cref{subsubsec:bounding_observables}, the two inequality constraints are introduced to ensure that the system is in a state close to the ground-state.
The upper-bound $\expval{\sigma^z_i}_{\max}$ is obtained by maximizing the same objective function over the same feasible set.

\begin{figure*}[t]
    \centering
    \includegraphics[width=0.47\linewidth]{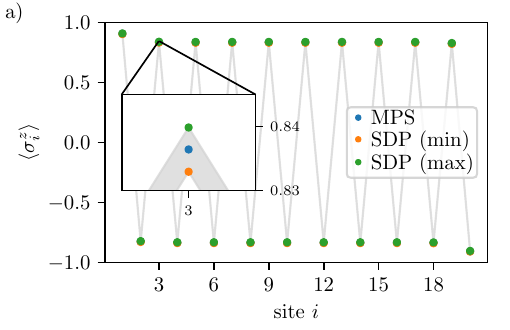}
    \hfill
    \includegraphics[width=0.47\linewidth]{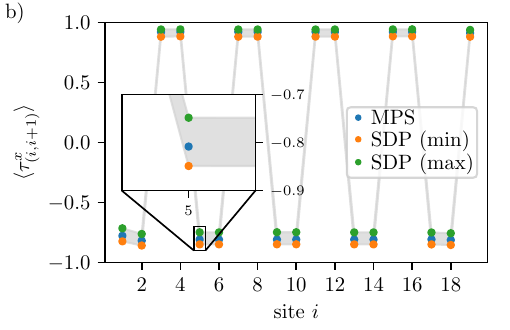} \\
    \caption{Bounds on one-body observables at the ground state energy of the Kogut-Susskind Hamiltonian in a 39-spin system with parameters $J=f=\mu=1$. 
    The observable is computed for each of the $N=20$ matter sites.  
    The expectation values are computed using MPS, while SDP relaxations are used to certify the admissible range (gray area). 
    (a) Estimates of the spatial distribution of the mass field $\langle \sigma_i^z \rangle$, with parameters $J=\mu=f=1$.
    (b) Estimates of the spatial distribution of the electric field terms $\langle \tau_{(i,i+1)}^z \rangle$, in the sector $\langle G_i\rangle=(-1)^i$ with parameters $J=1$, $\mu=1.5$ and $f=0.5$.}
    \label{fig:energies_1D_spatial_distrib}
\end{figure*}
\cref{fig:energies_1D_spatial_distrib}~(a) displays the estimation of the matter occupation $\expval{\sigma^z_i}$ for each matter site in a 39-spin system, hosting 20 matter sites.
The expectation values are obtained via MPS simulations, while the admissible area is provided by SDP relaxations at level 3. 
\cref{fig:energies_1D_spatial_distrib}~(a) reveals the distinctive alternating pattern in $\langle\sigma_i^z\rangle$ that is induced in the chosen gauge-invariant sector by the mass and field term of \cref{eq:hamiltonian}. 
Translation symmetry is explicitly broken by one lattice spacing and the gauge-invariant local density is biased toward two alternating sublattice matter occupations.
In the strong-$|\mu|$ (or strong-$f$) limit, the system approaches a maximally staggered pattern with $\lvert\expval{\sigma_i^z}\rvert\rightarrow1$. 
For finite hopping, field, and mass terms ($0\neq J;f,\mu\ll\infty$), as used here, quantum fluctuations reduce the density oscillations to the observed pattern $\expval{\sigma_i^z}\simeq(-1)^{i+1}Z_0$ with $0<Z_0<1$, which is well-resolved by the two-sided bounds.

Similarly, \cref{fig:energies_1D_spatial_distrib}~(b) shows the estimation of the electric field $\expval{\tau^x_{i,i+1}}$ for each of the 19 links in a 39-spin system, now considering the sector $\langle G_i\rangle=(-1)^i$. 
In the strong-$\mu$ (or small-f) limit, the electric field approaches $\lvert\tau^x_{i,i+1}\rvert\rightarrow1$ while exhibiting a staggered pattern of 2-link domains with a characteristic period of four sites. 
The sign of the electric field changes upon crossing each antiparticle site, $\tau^x_{i-1,i}\tau^x_{i,i+1}=(-1)^i$. 
For the finite hopping, field, and mass parameters considered here, this takes the form $\expval{\tau^x_{i,i+1}}\simeq(-1)^{\lfloor(i+1)/2\rfloor}X_{\lfloor(i+1)/2\rfloor}$ with $0<X_1<X_0<1$, where the asymmetry in amplitudes $X_1,X_0$ is a consequence of the non-vanishing field term $f$ favoring a positive electric field. 
This pattern is well described by the two-sided bounds.

\subsection{Bounding many-body observables}

Interestingly, the SDP relaxation approach can be extended to certify expectation values of many-body observables whose support extends beyond the interaction range of the generating sequence.
In the (1+1)-dimensional KS model, an important class of long-range operators are \emph{mesonic strings}~\cite{zohar_local_2020,borla_confined_2020,kebric_confinement_2021,davoudi_scattering_2024,chai_scalable_2026}.
A mesonic string operator connecting matter sites $i$ and $j$ is defined as
\begin{equation} \label{eq:mesonic_string}
    S(i,j) = \sigma^{+}_{i} \tau_{(i,i+1)}^z \prod_{k=i+1}^{j-1} \left( -\sigma^{z}_{k} \tau^{z}_{(k,k+1)}  \right) \sigma^{-}_{j}.
\end{equation}
For example, $S(1,4) = \sigma^{+}_{1} \tau^{z}_{(1,2)} \sigma^{z}_{2} \tau^{z}_{(2,3)} \sigma^{z}_{3} \tau^{z}_{(3,4)} \sigma^{-}_{4}$.
The length of the mesonic string is defined as $r = j-i$.
In what follows, we will fix the starting point of the mesonic string to the second left-most matter site $i=2$---to avoid boundary effects---and write $S(r) \equiv S(2,2+r)$.

Since $S(r)$ is not Hermitian, its expectation value cannot be certified directly via the SDP relaxation. 
Moreover, the quantity of physical interest is typically not $\expval{S(r)}$ itself but its modulus $|\expval{S(r)}|$, whose decay as a function of the length $r$ encodes information about confinement~\cite{borla_confined_2020}.
After decomposition in the Pauli basis, the expectation value $\expval{S(r)}$ is linear in $\boldsymbol{y}$ (see \cref{eq:def_y_variables}), however the modulus,
\begin{equation}
    |\expval{S}| = \sqrt{\left(\Re \expval{S} \right)^2 + \left( \Im \expval{S} \right)^2}\,,
\end{equation}
is not and, thus, it cannot be bounded directly by the SDP relaxation.

A first strategy is to bound the real and imaginary parts individually, and combine the resulting bounds via error propagation. 
However, this approach yields loose bounds, since the real and imaginary parts are optimized independently and the overall bound is effectively determined by the worst case of each.

To obtain tighter bounds, we instead express the modulus as a function of the phase. 
For a complex expression $z$, fix a phase $\theta$ and define
\begin{equation} \label{eq:modulus_def}
    |z|_{\theta} =
    \begin{cases}
        \dfrac{\Re z }{\cos{\theta}} & \text{if } \theta \in \left[-\frac{\pi}{4}, \frac{\pi}{4}\right) \cup \left[\frac{3\pi}{4}, \frac{5\pi}{4}\right), \\[2mm]
        \dfrac{\Im z}{\sin{\theta}} & \text{otherwise},
    \end{cases}
\end{equation}
where the real and imaginary parts satisfy the constraint
\begin{equation} \label{eq:angle_constraint}
    \Re z \sin \theta = \Im z \cos \theta.
\end{equation}
Both \cref{eq:modulus_def} and \cref{eq:angle_constraint} are linear in $z$, and hence in the variables $\boldsymbol{y}$ once $z$ is expressed as a linear function of $\boldsymbol{y}$ in the SDP relaxation. 
The problem
\begin{equation}
    \begin{split}
        \max_{\boldsymbol{y} \in \mathbb{R}^n} \; & |z|_{\theta} \\
        \text{s.t.} \; & y_{\mathbb{I}} = 1, \\
        & \Gamma(\boldsymbol{y}) \succeq 0, \\
        & \Re z \sin \theta = \Im z \cos \theta,
    \end{split}
\end{equation}
is therefore a valid SDP for any fixed $\theta$. 
An upper bound on $|z|$ is then obtained by an outer optimization over $\theta$: for each phase we solve the SDP above, and the modulus bound is the maximum over all phases.

Applying this scheme to the mesonic string, the upper bound on $|\expval{S(r)}|$ is obtained from
\begin{equation} \label{eq:mesonic_string_optimization}
    \begin{split}
        |\expval{S(r)}|_{\max} = \max_{\theta} & \max_{\boldsymbol{y} \in \mathbb{R}^n} \; |\expval{S(r)}|_{\theta} \\
        \text{s.t.} \; & y_{\mathbb{I}} = 1, \\
        & \expval{G_k} = 1 \quad \forall k = 1, \dots, N, \\
        & E_{\max} - \expval{H} \geq 0, \\
        & \expval{H} - E_{\min} \geq 0, \\
        & \Gamma(\boldsymbol{y}) \succeq 0, \\
        & \Re\expval{S(r)} \sin \theta = \Im \expval{S(r)} \cos \theta,
    \end{split}
\end{equation}
where $\expval{S(r)}$ is a linear (complex) expression in the variables $\boldsymbol{y}$. 
The lower bound is obtained analogously, by replacing the double maximization with a double minimization.

In practice, the outer maximization over $\theta$ is performed by choosing a finite number of points around the trigonometric circle, $\theta \in [0, 2\pi)$, and solving the corresponding SDP at each sampled phase. 
This procedure does not constitute a proof of optimality as does the certification of single-body observables or the energy, as the true optimal phase could in principle fall between two sampled points. 
However, by inspecting the feasible set under the linear map given by $\expval{S}$ and assuming it varies smoothly with $\theta$, we can reliably extrapolate the optimum from the sampled values.
More details are given in \Cref{app:feas_set}.

Since the support of the mesonic string can extend beyond the interaction range, the generating sequence underpinning the SDP relaxation must be augmented with additional monomials so that the desired observable appears naturally in the moment matrix.
In this case, we choose to include the following set of Pauli strings to the generating sequence:
\begin{equation}
    \left\{ \sigma^\alpha_i \tau_{(i,i+1)}^{\gamma} \prod_{k=i+1}^{j-1} \left( \sigma^{\gamma}_{k} \tau^{\gamma}_{(k,k+1)}  \right) \sigma^{\beta}_j : \alpha,\beta,\gamma \in \{x,y,z\} \right\}
\end{equation}
This choice is purely empirical, based on our observation that it performs well in practice. 
There may exist better choices that introduce fewer additional operators while imposing tighter constraints on the feasible set.

To preserve the block ordering of the generating sequence, we apply the criterion in \cref{eq:ordering_criterion}. 
This ensures that each newly introduced monomial is assigned to the appropriate subsequence.
Note that each of these extra monomials only add a single row and a single column to the moment matrix, leading to only minimal impact on the runtime.

\begin{figure}
    \centering
    \includegraphics[width=\linewidth]{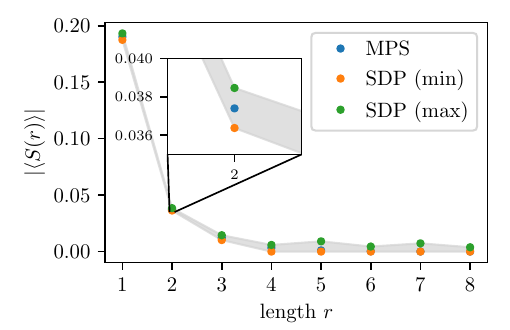}
    \caption{Bounds on the expectation value of the modulus of the mesonic string operator $|\expval{S(r)}|$ as a function of the string length $r$. 
    Simulations are performed on a $39$-spin system, with matter sites labeled from $1$ to $20$. 
    The Hamiltonian parameter are fixed to $J = f = \mu = 1$.
    The mesonic string operator extends from matter site $2$ to matter site $r+2$, with $r = 1, \ldots, 8$; strings starting or ending too close to the boundary are excluded to minimize finite-size effects. 
    Expectation values are obtained from matrix product state simulations (solid line/markers), while the SDP relaxation, combined with an optimization over a discrete set of $20$ phases $\theta \in [0,2\pi)$, provides upper and lower bounds, shown as the grey shaded region.}
    \label{fig:mesonic_strings}
\end{figure}

Figure~\ref{fig:mesonic_strings} shows the value of $|\expval{S(r)}|$ as a function of the string length $r$ on a 39-spin system. 
For each admissible endpoint, the MPS result is compared with the interval certified by the SDP relaxation. 
Boundary endpoints are excluded to suppress finite-size effects. 
The agreement between the MPS values and the certified intervals demonstrates that the relaxation can provide non-trivial bounds on the expectation value of non-local string observables over a broad range of string lengths. 
The certified intervals can be tightened by strengthening the relaxation, for example by increasing the level of the hierarchy, enlarging the interaction range, or including additional monomials in the moment sequence.

\section{Pure Gauge Model in (2+1) dimensions}
\label{sec:z2_2d}

The Kogut-Susskind model can also be formulated in two spatial dimensions. 
As before, we consider the model with open boundary conditions.
In contrast to the one-dimensional case, we limit our computations here to the pure gauge case, i.e. we do not consider fermions.
Consequently, the degrees of freedom lie only on the edges of the 2D lattice.
For a $m \times n$ lattice, this corresponds effectively to a $2mn-m-n$ spin system.
In analogy with the one-dimensional case, we describe the model in terms of Pauli operators.

The Hamiltonian is the sum of the two terms $ H = H_\mathrm{E} + H_\mathrm{B}$ with
\begin{equation}
\begin{aligned}
    H_\mathrm{E} & = -f\sum_{e \in \mathcal{E}} \tau^{x}_{e} \\
    H_\mathrm{B} & = -h\sum_{p \in \mathcal{P}} \tau^z_{p_1} \tau^z_{p_2} \tau^z_{p_3} \tau^z_{p_4}
\end{aligned}
\end{equation}
where $f$ and $h$ are electric and magnetic couplings and $\mathcal{E}$ denotes the set of all edges and $\mathcal{P}$ the set of all plaquettes. 
With a slight abuse of notation, we represent each plaquette by the set of its four boundary edges, $p = \{p_1, p_2, p_3, p_4\}$.
In analogy with the $U(1)$ theory, the four-body term is called the magnetic term $H_B$ and the terms acting on the edges are referred to as the electric term $H_E$.

The model is invariant with respect to the local gauge operators 
\begin{equation} \label{eq:gauge_operators_2D}
    G_{v} = \tau^x_{v_1} \tau^x_{v_2} \tau^x_{v_3} \tau^x_{v_4} \quad \forall v \in \mathcal{V}
\end{equation}
where $\mathcal{V}$ denotes the set of vertices. 
Again, with a slight abuse of notation, we represent each vertex $v \in \mathcal{V}$ by the set of its four incident edges, $v = \{v_1, v_2, v_3, v_4\}$.
For vertices at the boundaries, the $\tau^x$ whose support is not on the lattice is dropped, for example, the gauge operator associated to a corner vertex has only two terms: $G_{\text{corner}} = \tau^x_{c_1} \tau^x_{c_2}$.

\subsection{Generating sequence}

Similarly to the one-dimensional case, we consider the NPA hierarchy up to level~3 with interaction range~1. 
\cref{fig:interaction_range_1_on_2D_lattice} shows all field sites whose Manhattan (1-norm) distance from the reference site is at most one.
\begin{figure}
    \centering
    \includegraphics[width=\linewidth]{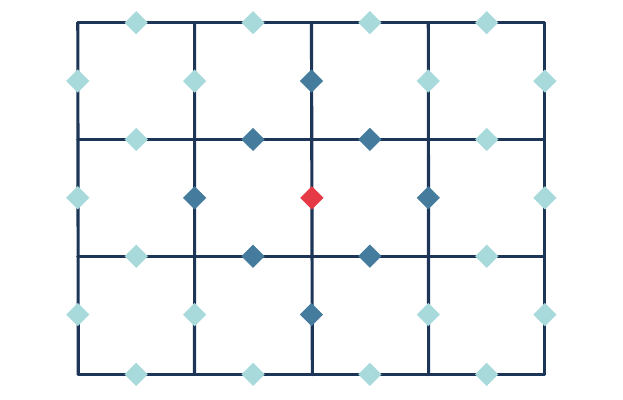}
    \caption{Illustration of the interaction range on a two-dimensional lattice. 
    Sites highlighted in dark blue are those whose Manhattan (1-norm) distance from the reference site (red) is equal to one.}
    \label{fig:interaction_range_1_on_2D_lattice}
\end{figure}
The same observations and reasoning as those presented in \cref{subsec:generating_sequence_1D} carry over to this setting. 
In particular, the symmetry-adapted basis matrix~$P$ is a permutation matrix, which implies that there exists an ordering of the monomial sequence for which the moment matrix has a block-diagonal form.

The underlying intuition is analogous to the one-dimensional case and the block-ordered sequence is obtained by applying the same criterion as in \cref{eq:ordering_criterion}. 

As in the one-dimensional case, the subsequences generating the blocks can be grouped into a finite number of classes. 
Each class consists of subsequences sharing the same geometric pattern, as illustrated in \cref{fig:classes_in_2D_sequence}. 
These patterns are obtained by representing the monomials appearing in the subsequences geometrically: each Pauli operator is associated with a lattice site, while monomials involving two Pauli operators are represented by edges connecting the corresponding sites. 
By construction, the resulting patterns encode the geometric structure of the monomials defining each class.
\begin{figure}
    \centering
    \includegraphics[width=\linewidth]{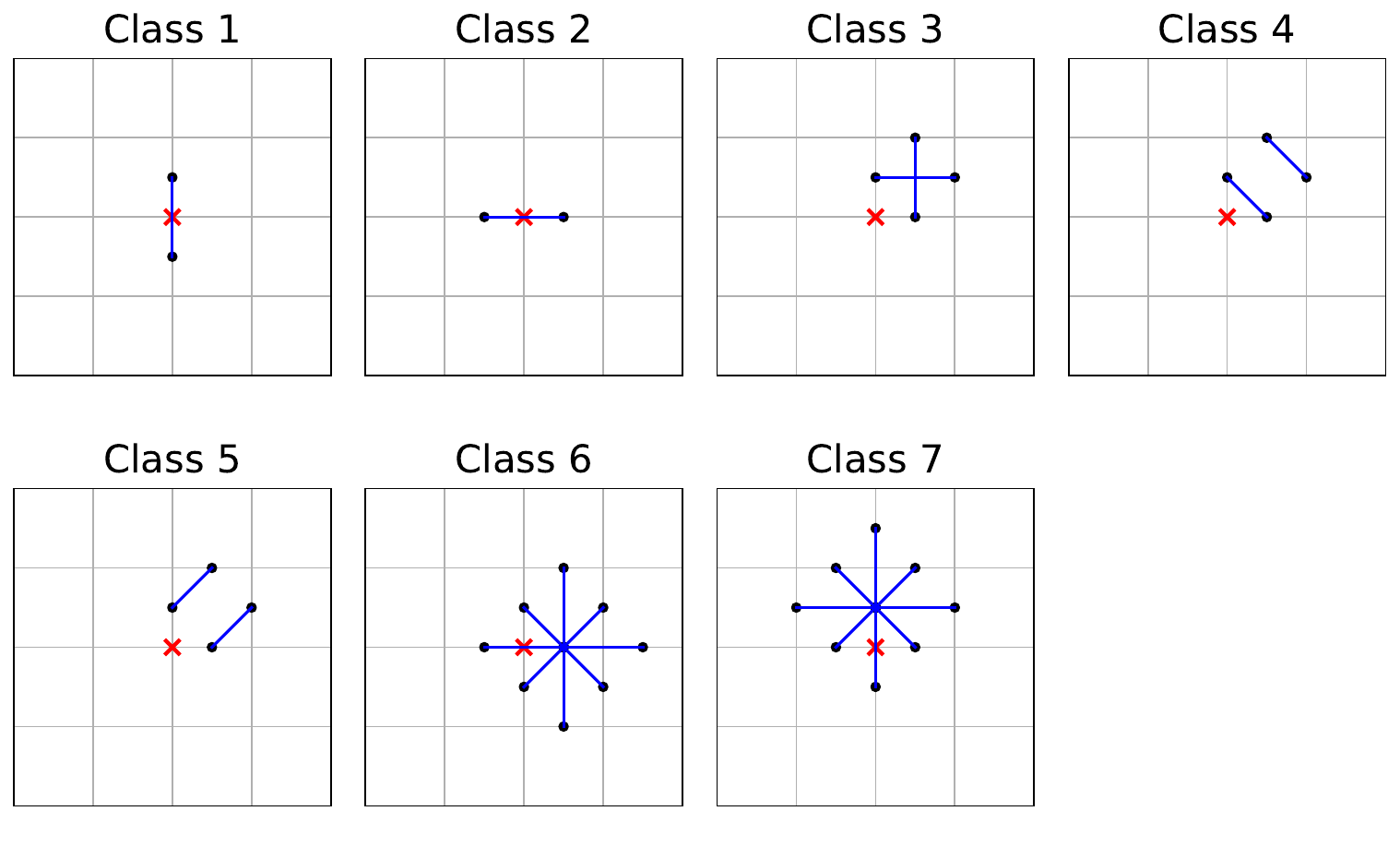}
    \caption{Classification of the subsequences contributing to the level~2, interaction-range~1 moment matrix. 
    Each lattice vertex $(i,j)$ (indicated by a red cross) is associated with seven distinct subsequences, each belonging to a different class. 
    The remaining subsequence, not shown in the figure, collects all monomials not included in the first seven classes. 
    It consists of all $\sigma^x$ Pauli strings up to level~2 and interaction range~1 that are not assigned to the other classes and therefore has size scaling linearly with the system size.}
    \label{fig:classes_in_2D_sequence}
\end{figure}

At level~2 with interaction range~1, the construction of the moment matrix can be systematically decomposed by partitioning all contributing subsequences into this finite set of classes. 
Every subsequence, except for a single \emph{remaining subsequence}, can be uniquely assigned to one of the classes shown in \cref{fig:classes_in_2D_sequence} according to its geometric structure. 
In particular, for every vertex $(i,j)$, represented by a red cross in the diagram, there are seven distinct subsequences associated with that vertex, one from each class. 
Together, these classes account for all admissible local monomial configurations.
The remaining subsequence contains all monomials that are present in any of the seven classes.
More precisely, it contains all $\sigma^x$ Pauli strings up to level~2 and interaction range~1 that are not included in the first seven classes. 
Its size scales linearly with the system size.

A detailed definition of the corresponding classes at level~3 is provided in \cref{app:blocks_2D_model}.

\subsection{Bounds on the energy}

We now extend the SDP approach described in \cref{subsec:energy_bounds_1D} to the two-dimensional Kogut-Susskind lattice gauge theory in the pure gauge setting. 
The construction of the relaxation follows the same procedure as in the one-dimensional case. 
In particular, the moment matrix, positivity constraints, and objective function are defined analogously, with the only modification being that the one-dimensional Hamiltonian, gauge operators and generating sequence are replaced by their two-dimensional counterparts. 
No additional conceptual changes to the SDP formulation presented in \cref{eq:sdp_relaxation_energy_1D} are required.

To assess the performance of the relaxation, we estimate the ground-state energy of the two-dimensional Kogut-Susskind Hamiltonian and compare the resulting certified lower bounds with reference values obtained from MPS calculations.
The MPS results are computed using a 1D chain in a snake pattern; it is expected that for large system sizes the results deteriorate due to non-local interactions, even for large bond dimension.
For such systems, better results could be found by using tensor network approaches tailored for higher dimensions, such as tree tensor networks~\cite{magnifico_lattice_2021, magnifico_tensor_2025}.
In this work, we employ a bond dimension of $D = 500$, which provides a suitable compromise between numerical convergence and computational efficiency, see \cref{app:mps_convergence} for more details.
Here, we do not focus on specific variational approaches, but on the relaxation bounds via SDP computations.

Figure~\ref{fig:energies_2D} reports the relative error of the certified lower bound for two relaxation levels, namely levels $2$ and $3$. 
As expected, increasing the relaxation level systematically improves the estimate of the ground-state energy, yielding tighter lower bounds. 
The grey shaded region indicates errors below the guaranteed numerical precision of the SDP solver.

Interestingly, both relaxation levels (SDP-2 and SDP-3) exhibit spikes at the same system sizes and with nearly identical amplitudes. 
More precisely, the data alternate between two well-defined branches: a higher-error branch corresponding to relatively \emph{wide} lattices of size $m \times n$ with $m,n>3$ (highlighted by the star markers in the figure), and a lower-error branch associated with \emph{narrow} geometries, i.e. quasi-1D systems, of the form $3 \times m$ (dot markers).
These features do not originate from the SDP relaxation but rather from the MPS reference energies used in the comparison. 
The energy given by the SDP relaxation is smooth over the whole range of lattices.
This behavior rather reflects the varying accuracy of the MPS approximation across different lattice geometries. 
For narrow systems, the MPS ansatz provides a more accurate approximation of the ground-state energy, resulting in a smaller discrepancy between the MPS energy and the SDP lower bound. 
In contrast, for wider lattices the MPS approximation deteriorates, leading to a larger apparent gap between the two methods.
This dependence on the lattice geometry is consistent with the structure of the MPS ansatz. 
The SDP relaxation, by contrast, is formulated directly in terms of local operator constraints and does not rely on such a mapping. 
Consequently, its performance is significantly less sensitive to the aspect ratio of the lattice and remains comparatively robust as the system becomes more genuinely two-dimensional.

These results demonstrate that the SDP hierarchy can be transferred directly to two-dimensional lattice gauge theories while retaining its ability to produce certified energy bounds. 
Moreover, they highlight an intrinsic advantage of the SDP approach over strongly geometry dependent methods for wider two-dimensional lattices.

\begin{figure}
    \centering
    \includegraphics[width=\linewidth]{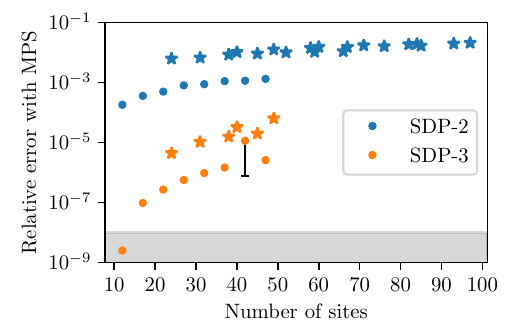}
    \caption{Relative error of the certified ground-state energy lower bound for the two-dimensional Kogut--Susskind Hamiltonian in the pure-gauge regime. 
    Results are shown for SDP relaxations at levels $2$ and $3$, using matrix product state (MPS) calculations with bond dimension $500$ as the reference. 
    The grey shaded region indicates errors below the guaranteed numerical precision of the SDP solver. 
    Dot markers denote narrow lattices of the form $3 \times m$, whereas star markers correspond to wider lattices of the form $m \times n$, with $m,n>3$.
    For the $3 \times 9$ lattice ($42$ sites), the level-$3$ relaxation did not fully converge. 
    The SDP solver nevertheless returned feasible primal and dual solutions with a relative duality gap below $2 \times 10^{-5}$, so the optimal value of the relaxation lies within the interval shown by the black error bar. 
    The orange marker gives the worst-case bound over this interval.}
    \label{fig:energies_2D}
\end{figure}

\subsection{Bounding few-body observables}

As in the the one-dimensional case, we consider observables invariant under the relevant symmetries, in particular the gauge symmetry operators in \cref{eq:gauge_operators_2D}.
One important class of such operators are Wilson loops~\cite{wilson_confinement_1974,kogut_introduction_1979}, which, for the $\mathbb{Z}_2$ theory considered here, are given by
\begin{equation}
    W_\mathcal{C}=\prod_{e\in\mathcal{C}}\tau_e^z
\end{equation}
where $\mathcal{C}$ is a closed loop of edges over the lattice.

In the limit of $h\rightarrow 0$, i.e. $H=H_\mathrm{E}$, the ground state takes the shape of the simple product state $\ket{\psi_\mathrm{E}}=\bigotimes_{e\in\mathcal{E}}\ket{+}_e$, corresponding to an absence of excitations of the electric field, and yielding a vanishing expectation value for all Wilson loops.

For the limiting case of $h\rightarrow\infty$, where $H_\mathrm{B}$ dominates, the ground state is the normalized simultaneous $+1$ eigenstate of all plaquette operators, corresponding to an absence of magnetic flux. Accordingly, it is $T_p\ket{\psi_{\mathrm{B}^+}}=\ket{\psi_{\mathrm{B}^+}}\,\forall p\in\mathcal{P}$ with the plaquette operator
\begin{equation}
    T_p=\tau^z_{p_1} \tau^z_{p_2} \tau^z_{p_3} \tau^z_{p_4}\,,
    \label{equ:plaquette_operator}
\end{equation} 
where $p_1$, $p_2$, $p_3$ and $p_4$ denote the edges surrounding the plaquette $p$.
For non-periodic boundary conditions, as considered here, the ground state uniquely reads
\begin{equation}
    \ket{\psi_{\mathrm{B}^+}}=\prod_{p\in\mathcal{P}}\left(\frac{1+T_p}{\sqrt{2}}\right)\ket{\psi_\mathrm{E}}=2^{-\lvert\mathcal{P}\rvert/2}\sum_\mathcal{C}W_\mathcal{C}\ket{\psi_\mathrm{E}}\,,
\end{equation}
where $\lvert\mathcal{P}\rvert$ denotes the total number of plaquettes in the system. Here, one utilizes that all Wilson loops can be expressed as a product of plaquette operators.
In the opposite limit $h\rightarrow-\infty$,
\begin{equation}
    \ket{\psi_{\mathrm{B}^-}}=\prod_{p\in\mathcal{P}}\left(\frac{1-T_p}{\sqrt{2}}\right)\ket{\psi_\mathrm{E}}=2^{-\lvert\mathcal{P}\rvert/2}\sum_\mathcal{C}(-1)^{A_\mathcal{C}}W_\mathcal{C}\ket{\psi_\mathrm{E}}\,,
\end{equation}
where the area $A_\mathcal{C}$ of a path $\mathcal{C}$ is given by the number of plaquettes it encloses.

Consequently, in the limiting cases $h\rightarrow0,\pm\infty$, one observes for all closed paths $\mathcal{C}$:
\begin{equation}
\begin{aligned}
    \bra{\psi_\mathrm{E}}W_\mathcal{C}\ket{\psi_\mathrm{E}}&=0 &&h\rightarrow0\\
    \bra{\psi_\mathrm{B^+}}W_\mathcal{C}\ket{\psi_\mathrm{B^+}}&=1 &&h\rightarrow+\infty\\
    \bra{\psi_\mathrm{B^-}}W_\mathcal{C}\ket{\psi_\mathrm{B^-}}&=(-1)^{A_\mathcal{C}} \hspace{2ex} &&h\rightarrow-\infty
\end{aligned}
\end{equation}
For finite $h$, the system encounters a topological phase transition in the thermodynamic limit at a critical magnetic coupling $\pm h_\mathrm{c}$, and the Wilson loop expectation values decay with an area law (of the area enclosed by $\mathcal{C}$) for $\lvert h\rvert<h_\mathrm{c}$, where large loops are suppressed, signaling confinement, and a perimeter law (of the perimeter described by $\mathcal{C}$) in a deconfined phase $\lvert h\rvert>h_\mathrm{c}$~\cite{wilson_confinement_1974,kogut_introduction_1979,fradkin_order_1978,horn_hamiltonian_1979,zohar_wilson_2021,lumia_two-dimensional_2022}.

The minimal Wilson loop operator is the one forming the perimeter around a single plaquette $p\in\mathcal{P}$, i.e.
$W_1 = T_p = \tau^z_{p_1}\tau^z_{p_2}\tau^z_{p_3}\tau^z_{p_4}$.
The corresponding SDP relaxation used to obtain certified bounds is analogous to \cref{eq:omin}, with the objective function, Hamiltonian, and gauge operators replaced by their two-dimensional counterparts.
\Cref{fig:observables_2D} shows the expectation value of the Wilson loop around the central plaquette of a $4 \times 4$ two-dimensional lattice, corresponding to a $24$-qubit system, as a function of the magnetic-field parameter $h$. 
The gray shaded region represents the interval certified by the level-3 SDP relaxation, which align closely with the Wilson loop behavior described above. The MPS estimates lie entirely within these certified bounds.

\begin{figure}[t]
    \centering
    \includegraphics[width=\linewidth]{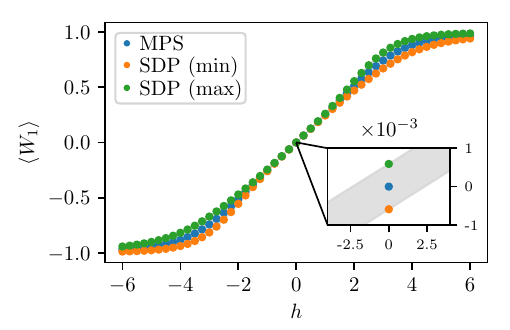}
    \caption{Expectation value of the Wilson loop operator around the central plaquette of a $4 \times 4$ two-dimensional lattice (24-qubit system) as a function of the magnetic-field parameter $h$. 
    The gray shaded region denotes the interval certified by the level-3 SDP relaxation, while the MPS estimates (blue) lie entirely within the certified bounds.}
    \label{fig:observables_2D}
\end{figure}

\subsection{Bounding many-body observables}

Analogously to the one-dimensional case, where we demonstrated that the SDP relaxation can be used to certify expectation values of non-local many-body observables such as mesonic string operators, the same approach can be extended to two-dimensional systems.

As an illustrative example, we consider the Wilson loop enclosing the four central plaquettes of a $5 \times 5$ lattice (i.e. $16$ plaquettes in total). 
We label the four central plaquettes by $p,q,r,s$ and their corresponding plaquette operators are given by \cref{equ:plaquette_operator}.
The plaquette operators $T^x_p,T^y_p$, etc. are defined analogously.
The Wilson operator corresponding to our Wilson loop is given by the product of the four plaquette operators
\begin{equation}
    W_4 = T^{z}_{p} T^{z}_{q} T^{z}_{r} T^{z}_{s}
\end{equation}

The Wilson loop operator is not contained in the moment matrix generated by the standard level-3 relaxation with interaction range $1$, since it involves the product of Pauli operators acting on eight distinct sites. 
To certify its expectation value, we therefore augment the operator sequence used to construct the moment matrix with all products of plaquette operators.
Specifically, we choose to include the operators
\begin{equation} \label{eq:wilson_loop_extra_operators}
    \left\{\prod_{p \in P} T^{\alpha}_{p} : \forall P \subset \{p,q,r,s\} \text{ and } \alpha=x,y,z\right\},
\end{equation}
where we consider all possible subsets of plaquettes $P$. 
Note that the choice of operator is purely empirical, based on our observation that it performs well in practice. 
There may exist alternative choices that introduce fewer additional operators while imposing tighter constraints on the feasible set.

\Cref{eq:wilson_loop_extra_operators} provides a list of $48$ operators that are added to the generating sequence.
These additional operators guarantee that the Wilson loop operator is represented in the moment matrix, allowing its expectation value to be bounded within the SDP relaxation.
As in the one-dimensional case, each newly introduced monomial is assigned to the appropriate subsequence according to the criterion given in \cref{eq:ordering_criterion}.
This guarantees that the generating sequence is block-ordered, and hence that the moment matrix is block-diagonal.

\begin{figure}
    \centering
    \includegraphics[width=\linewidth]{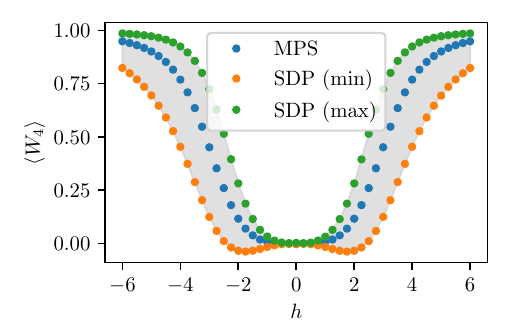}
    \caption{Expectation value of the Wilson loop operator enclosing the four central plaquettes of a $5 \times 5$ lattice as a function of the magnetic-field parameter $h$. 
    The grey shaded region denotes the interval certified by the level-3 SDP relaxation with the augmented operator sequence, while the tensor-network estimates (blue) lie entirely within the certified bounds.}
    \label{fig:wilson_loops}
\end{figure}

Figure~\ref{fig:wilson_loops} shows the certified expectation value of the Wilson loop enclosing the four central plaquettes of the $5 \times 5$ lattice as a function of the magnetic-field parameter $h$. 
The grey shaded region represents the interval certified by the level-3 SDP relaxation with the augmented operator sequence, while the tensor-network estimates lie entirely within the certified bounds.

\section{Conclusion}
\label{sec:conclusion}

Semi-definite programming relaxations are a powerful complement to variational methods.
Variational ansätze such as tensor networks find upper bounds to the ground-state energy and estimates for the observables.
Their quality is often difficult to certify: a low variational energy does not rule out the existence of a lower-lying state outside the ansatz manifold or a failure to find the global minimum inside of it.
An SDP relaxation, by contrast, produces a certificate --- an explicit witness that no quantum state consistent with the imposed operator constraints can have energy below the reported bound. 
This certificate can be propagated to other observables, is checkable independently of the solver, and constitutes a machine-verifiable proof of the bound. 

In this work, we have presented two-sided bounds on the ground-state energy and on observables of an Abelian lattice gauge theories by combining variational methods with a hierarchy of semidefinite programming relaxations. 
Applied to a $(1+1)$-dimensional $\mathbb{Z}_2$ gauge theory with dynamical fermionic matter and to a pure $\mathbb{Z}_2$ gauge theory in $(2+1)$ dimensions, the method yields rigorously certified bounds that tighten monotonically as the SDP level increases.
For a one-dimensional system with $129$ sites, the certified interval has a relative spread of $0.04\%$. 
For a two-dimensional system with $97$ sites (corresponding to a $7 \times 8$ lattice), the relative spread is $2.08\%$.
Beyond energy estimates, the proposed method provides rigorous bounds on both short- and long-range observables. 
As a proof of concept, we demonstrate non-trivial bounds for physically relevant quantities, including the mesonic string magnitude and Wilson loops spanning multiple plaquettes.

The same certification logic extends naturally to quantum simulation experiments. 
Given a measured expectation value of the Hamiltonian on a prepared state, the two-sided bounds on energy and observables guarantee that the system has been correctly simulated, independently of the experimental platform and without requiring classical simulation of the full system.

A distinctive feature of the approach is that it accommodates the gauge symmetry and any additional global symmetries of the theory in a way that both sharpens the bounds and dramatically reduces the computational cost. 
Since the number of independent gauge symmetries is an extensive quantity, the size of the symmetry group grows with system size. 
In a symmetry-adapted basis, the moment matrix is block diagonal, and larger symmetry groups lead to a finer block decomposition. 
Together, these properties yield a favorable scaling of the block sizes and, consequently, of the computational cost.

Several natural extensions can be considered:
First, the construction of the block-ordered operator sequence underlying the moment matrix currently proceeds by enumerating all pairs of operators $s_i$, $s_j$ and then grouping them in the same block if they satisfy the condition $g(s_i s_j^{\dagger})  = s_i s_j^{\dagger}$.
A more efficient method to apply the gauge symmetry and additional, global symmetries would enable even larger computations.

Second, only a subset of the available symmetries has been exploited here.
Incorporating additional symmetries, such as translational or rotation invariance, would further block-diagonalize the moment matrix and allow the hierarchy to reach higher levels and larger interaction ranges at fixed computational cost.

Third, the present formulation treats the operator algebra in a bosonic framework. 
Extending the hierarchy to handle fermionic algebras natively --- rather than via a Jordan-Wigner mapping, which obscures locality in more than one spatial dimension --- would open the method to a broader class of fermionic lattice models and would be a natural next step toward $(2+1)$-dimensional theories with dynamical matter.
Additionally, investigating extensions of the symmetries to compact Lie or non-Abelian groups would be of particular interest.

Together, these directions suggest that SDP-based certification is a practical and scalable tool for the study of strongly correlated lattice gauge theories, with applications ranging from classical benchmarking to the verification of near-term quantum hardware.

\begin{acknowledgments}
The authors acknowledge fruitful discussions with Erez Zohar.
P.E.~acknowledges the support received from the Dutch National Growth Fund (NGF) as part of the QDNL program in the NWO-Quantum Technology program (Grant No.~NGF.1623.23.006).
P.E.~also acknowledges funding from the Carl-Zeiss-Stiftung (CZS Center QPhoton) and the financial support programmes for early career researcher by the Graduate and Professional Training Center of Ulm University. 
J.T.~and E.V.~acknowledge the support received from the European Union's Horizon Europe research and innovation programme through the ERC StG FINE-TEA-SQUAD (Grant No.~101040729).
A.K.~acknowledges support from the European Union Research Council (ERC, OverSign, 101122583).
J.T and J.M~acknowledge support from the Dutch National Growth Fund (NGF), as part of the QDNL Programme.
\end{acknowledgments}

\section*{Data availability}
This work was performed using the ALICE high-performance computing resources provided by Leiden University. 
The data as well as the code used to generate the plots presented in this work is publicly available in the accompanying Zenodo repository~\cite{vallee_data_2026}. 
Unless stated otherwise, all computations were carried out using 16 CPU cores per task, with a maximum wall-clock time of 6 hours and a memory limit of 150~GB allocated to each data point. 

\onecolumngrid
\FloatBarrier
\appendix

\section{SDP relaxation example} \label{app_sdp_relaxation_example}

In this section, we illustrate the application of SDP relaxation to a non-commutative polynomial optimization problem over a single-qubit system.
Consider the optimization problem:
\begin{equation}
\begin{aligned}
    \max_{\rho \in \mathcal{S}} \expval{\sigma^x} + \expval{\sigma^y} \\
    \text{s.t.} \expval{\sigma^z} = 0,
\end{aligned}
\end{equation}
where we omit the explicit reference to the state $\rho$ in the expectation value, i.e. $\expval{ \cdot } = \Tr( \; \cdot \; \rho)$.

To construct the SDP relaxation of this problem, we begin by associating a variable with the expectation value of every word in the algebra (in this case the algebra is the Pauli algebra of a single qubit). 
It satisfies the Pauli multiplication rules: $\sigma^i \sigma^j = \delta_{ij} + \mathrm{i} \epsilon_{ijk} \sigma^k$.
We define the variables $y_k$ as the expectation values of the elements:
\begin{equation}
y_1 = \langle\mathbb{I}\rangle, \quad y_2 = \langle\sigma^x\rangle, \quad y_3 = \langle\sigma^y\rangle, \quad y_4 = \langle\sigma^z\rangle.
\end{equation}
Note that due to state normalization, we inherently have the constraint $y_1 = 1$. 

Next, we build the moment matrix. 
We choose the generating sequence $\boldsymbol{s}$ as:
\begin{equation}
\boldsymbol{s} = 
\begin{pmatrix}
\mathbb{I} \\
\sigma^x \\
\sigma^y \\
\sigma^z
\end{pmatrix}.
\end{equation}
By definition, the entries of the corresponding moment matrix $\Gamma$ are given by
\begin{equation}
\Gamma_{ij} = \Tr(s_i s_j^{\dagger} \rho),
\end{equation}
where $\rho$ is the density matrix of the qubit. 
Using the Pauli algebra rules, the explicit form of the moment matrix is
\begin{equation}
\Gamma =
\begin{pmatrix}
\expval{\mathbb{I}} & \langle\sigma^x\rangle & \langle\sigma^y\rangle & \langle\sigma^z\rangle \\
\langle\sigma^x\rangle & \expval{\mathbb{I}} & \mathrm{i}\langle\sigma^z\rangle & -\mathrm{i}\langle\sigma^y\rangle \\
\langle\sigma^y\rangle & -\mathrm{i}\langle\sigma^z\rangle & \expval{\mathbb{I}} & \mathrm{i}\langle\sigma^x\rangle \\
\langle\sigma^z\rangle & \mathrm{i}\langle\sigma^y\rangle & -\mathrm{i}\langle\sigma^x\rangle & \expval{\mathbb{I}}
\end{pmatrix}.
\end{equation}
The moment matrix can then be written as a linear combination of Hermitian matrices:
\begin{equation}
\Gamma(\boldsymbol{y}) = \sum_{k=1}^4 y_k A^{(k)} = y_1 A^{(1)} + y_2 A^{(2)} + y_3 A^{(3)} + y_4 A^{(4)},
\end{equation}
where the matrices $A^{(k)}$ are defined as:
\begin{align}
A^{(1)} = 
\begin{pmatrix}
1 & 0 & 0 & 0 \\
0 & 1 & 0 & 0 \\
0 & 0 & 1 & 0 \\
0 & 0 & 0 & 1
\end{pmatrix}, 
\quad
A^{(2)} = 
\begin{pmatrix}
0 & 1 & 0 & 0 \\
1 & 0 & 0 & 0 \\
0 & 0 & 0 & \mathrm{i} \\
0 & 0 & -\mathrm{i} & 0
\end{pmatrix},
\quad
A^{(3)} = 
\begin{pmatrix}
0 & 0 & 1 & 0 \\
0 & 0 & 0 & -\mathrm{i} \\
1 & 0 & 0 & 0 \\
0 & \mathrm{i} & 0 & 0
\end{pmatrix}, 
\quad
A^{(4)} = 
\begin{pmatrix}
0 & 0 & 0 & 1 \\
0 & 0 & \mathrm{i} & 0 \\
0 & -\mathrm{i} & 0 & 0 \\
1 & 0 & 0 & 0
\end{pmatrix}.
\end{align}
They serve as weighted mask matrices, with each entry specifying the contribution (i.e., coefficient) of the $i$-th variable to the corresponding entry of $\Gamma(\boldsymbol{y})$.

Finally the SDP relaxation takes the form
\begin{equation}
    \begin{aligned}
        \max_{\boldsymbol{y} \in \mathbb{R}^4} \quad & y_2+y_3 \\
        \text{s.t.} \quad & y_1 = 1 \quad \text{(normalization constraint)}\\
        & y_4 = 0  \quad \text{(original problem constraint)}\\
        & \Gamma(\boldsymbol{y}) \succeq 0 \quad \text{(moment matrix constraint)} .
    \end{aligned}
\end{equation}

Note that, in the more general case, this procedure can be carried out for any free quotient algebra $\mathcal{A}$. 
Consequently, all expressions must be systematically reduced to a standard representation, commonly referred to as the \emph{normal form}.
For example, in the single-qubit Pauli algebra, the product $\sigma^x \sigma^y$ is always reduced to $\mathrm{i} \sigma^z$, etc.

\section{Invariant SDP relaxation} \label{app:block-diag_sdp}

In this section, we show that, under the assumptions detailed in \cref{subsubsec:sdp_with_symmetries} and summarized below, the SDP relaxation admits a block-diagonalization. 
This structure allows the original semidefinite program to be decomposed into smaller, independent blocks, which can substantially simplify both its analysis and its numerical solution.
 
Consider the SDP relaxation given in \cref{eq:npo_problem_sdp_relaxation}. The positive-semidefinite (PSD) constraint corresponds to an $m \times m$ moment matrix associated with a sequence $\boldsymbol{s}$, namely
\begin{equation}
    \Gamma_{ij}(\boldsymbol{y}) := \Tr\!\left(s_i s_j^{\dagger} \rho\right) .
\end{equation}
Without loss of generality, the moment matrix can be written as $\Gamma(\boldsymbol{y}) = \sum_k A^{(k)} y_k$, where the $A^{(k)}$ are Hermitian matrices.
 
As explained in \cref{subsubsec:sdp_with_symmetries}, each symmetry transformation — denoted $g$ or $\mathfrak{g}$, depending on the space on which it acts — is a $*$-automorphism. 
We further assume that there exists a unitary matrix $R(g) \in \mathbb{C}^{m \times m}$ such that
\begin{equation}
    g(\boldsymbol{s}) = R(g)\, \boldsymbol{s} .
\end{equation}
This condition is crucial, as it guarantees that the transformed generating sequence can always be expressed as a linear combination of the elements of the original sequence $\boldsymbol{s}$.
Taking the adjoint, we have $g(\boldsymbol{s})^{\dagger} = \boldsymbol{s}^{\dagger} R(g)^{\dagger}$. Using this, each entry of the moment matrix transforms as
\begin{equation}
\label{eq:transformation_law}
\begin{split}
    \mathfrak{g}\bigl(\Gamma_{ij}(\boldsymbol{y})\bigr)
    &= \sum_k A_{ij}^{(k)} \, \mathfrak{g}(y_k)
     = \Tr\!\left(g(s_i)\, g(s_j)^{\dagger} \rho\right) = \sum_{i',j'} R_{ii'} R_{jj'}^{*} \, \Tr\!\left(s_{i'} s_{j'}^{\dagger} \rho\right) \\
    &= \sum_{i',j',k} R_{ii'} R_{jj'}^{*} A_{i'j'}^{(k)} \, y_k = \sum_{i',j'} R_{ii'} R_{jj'}^{*} \, \Gamma_{i'j'}(\boldsymbol{y}).
\end{split}
\end{equation}
In matrix form, this reads
\begin{equation}
\label{eq:transformation_matrix_form}
    \mathfrak{g}\bigl(\Gamma(\boldsymbol{y})\bigr) = R(g) \Gamma(\boldsymbol{y}) R(g)^{\dagger},
\end{equation}
where the action of $\mathfrak{g}$ on a matrix of variables extends naturally to each of its components individually. We are now in a position to state the main result of this section.
 
\begin{theorem}
Suppose the objective and all the (in)equality constraints are invariant under a group $\mathcal{G}$, and suppose there exists a unitary representation $R : \mathcal{G} \to \mathbb{C}^{m \times m}$ such that
\begin{equation}
\label{eq:group_action_on_s}
    g(\boldsymbol{s}) = R(g)\, \boldsymbol{s} .
\end{equation}
Then:
\begin{enumerate}
    \item[(1)] the SDP relaxation is invariant under $\mathcal{G}$;
    \item[(2)] the matrix $P$ that block-diagonalizes $R$ also block-diagonalizes the moment matrix. We refer to $P$ as the symmetry-adapted basis.
\end{enumerate}
\end{theorem}
 
\begin{proof}[Proof (part 1)]
By hypothesis, the objective and the (in)equality constraints are invariant. Invariance of the PSD constraint follows from the identity
\begin{equation}
\label{eq:psd_transform}
    \mathfrak{g}\bigl(\Gamma(\boldsymbol{y})\bigr) = R(g)\, \Gamma(\boldsymbol{y})\, R(g)^{\dagger} ,
\end{equation}
together with the fact that, since $R(g)$ is unitary,
\begin{equation}
\label{eq:psd_preserved}
    R(g)\, \Gamma(\boldsymbol{y})\, R(g)^{\dagger} \succeq 0 \iff \Gamma(\boldsymbol{y}) \succeq 0 .
\end{equation}
\end{proof}
 
\begin{proof}[Proof (part 2)]
Since the SDP is invariant under $\mathcal{G}$, the feasible set may be restricted to the symmetric subspace $\tilde{\mathcal{F}}$, as described in \cref{eq:symmetric_feas_set}, and any solution therefore satisfies $\mathfrak{g}(\boldsymbol{y}) = \boldsymbol{y}$. Consequently, the moment matrix itself must be invariant under $\mathfrak{g}$, i.e.\ $\Gamma(\boldsymbol{y}) = \mathfrak{g}\bigl(\Gamma(\boldsymbol{y})\bigr)$, which yields
\begin{equation}
\label{eq:invariance_condition}
    \sum_i A^{(i)} y_i = \sum_i R(g)\, A^{(i)}\, R(g)^{\dagger} \, y_i .
\end{equation}
Since this identity holds for every $\boldsymbol{y} \in \tilde{\mathcal{F}}$, it follows that
\begin{equation}
\label{eq:commutant_condition}
    A^{(i)} = R(g)\, A^{(i)}\, R(g)^{\dagger}
\end{equation}
for all $g \in \mathcal{G}$ and all $i$. The matrices $A^{(i)}$ therefore lie in the commutant of the representation $R$.
 
By a standard result of representation theory, if a matrix $A$ lies in the commutant of a representation $R$ — that is, $A R(g) = R(g) A$ for all $g \in \mathcal{G}$ — then there exists a unitary change of basis $P$ under which the representation decomposes into irreducible components and both operators become simultaneously block diagonal. In this basis,
\begin{equation}
\label{eq:block_diagonalization}
\begin{aligned}
    P^\dagger R(g) P &= \bigoplus_k \left(\mathbb{I}_{\mu_k} \otimes \tilde{R}_k(g)\right), \\
    P^\dagger A P &= \bigoplus_k \left(\tilde{A}_k \otimes \mathbb{I}\right),
\end{aligned}
\end{equation}
where the $\tilde{R}_k$ denote the inequivalent irreducible representations occurring in $R$, and $\mu_k$ their respective multiplicities. In particular, the decomposition is indexed by $k$, and the action of $A$ reduces to independent blocks of size $m_k := \dim(\tilde{R}_k)\, \mu_k$, one for each irreducible sector, with $\sum_k m_k = m$ \cite{serre_linear_1977}.
 
Finally, positive semidefiniteness is preserved under unitary transformations, i.e.
\begin{equation}
\label{eq:psd_unitary_invariance}
    \Gamma(\boldsymbol{y}) \succeq 0 \iff P^{\dagger} \Gamma(\boldsymbol{y}) P \succeq 0 .
\end{equation}
Hence the single PSD constraint in \cref{eq:npo_problem_sdp_relaxation} may be replaced by the collection of smaller PSD constraints
\begin{equation}
\label{eq:reduced_psd_constraints}
    \sum_i \tilde{A}_k^{(i)} \, y_i \succeq 0 \qquad \forall k ,
\end{equation}
where $\tilde{A}_k^{(i)}$ denotes the $k$-th irreducible block of $A^{(i)}$ in the basis $P$, i.e.\ $P^\dagger A^{(i)} P = \bigoplus_k \bigl(\tilde{A}_k^{(i)} \otimes \mathbb{I}\bigr)$. This establishes the desired block-diagonal reformulation of the SDP announced at the beginning of this section.
\end{proof}

\section{MPS convergence} \label{app:mps_convergence}

To verify the numerical accuracy of the tensor network calculations, we performed a systematic convergence analysis with respect to the bond dimension $D$. 
The bond dimension determines the amount of entanglement that can be represented by the matrix product state (MPS) ansatz; consequently, increasing $D$ systematically improves the variational approximation while also increasing the computational cost.
To assess convergence, we computed the ground-state energy for a sequence of increasing bond dimensions and monitored its dependence on $D$.

\begin{figure}
    \centering
    \includegraphics[width=0.45\linewidth]{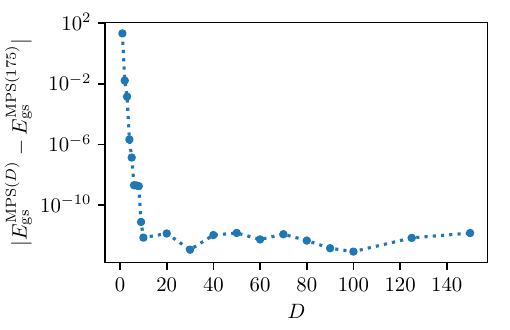}
    \hfill
    \includegraphics[width=0.45\linewidth]{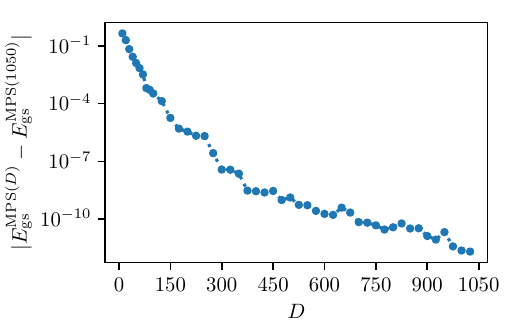} \\
    (a) \hspace{26em} (b)
    \caption{Difference between the ground-state energy and a reference value as a function of the bond dimension $D$. 
    (a) Convergence of the ground-state energy for the one-dimensional Kogut-Susskind Hamiltonian on a 159-spin system. 
    The reference value is the ground-state energy obtained from an MPS simulation with bond dimension $D=150$. 
    (b) Convergence of the ground-state energy for the two-dimensional pure-gauge Kogut-Susskind Hamiltonian on an 84-spin system. 
    The reference value is the ground-state energy obtained from an MPS simulation with bond dimension $D=1000$.}
    \label{fig:mps_convergence}
\end{figure}

Figure~\ref{fig:mps_convergence} shows the convergence of the ground-state energy with increasing MPS bond dimension for the one- and two-dimensional Kogut-Susskind Hamiltonians. 
In both cases, the energy decreases as $D$ increases, consistent with the variational nature of the MPS ansatz.

For the one-dimensional model, shown in \cref{fig:mps_convergence}~(a), the ground-state energy converges rapidly and is essentially converged already at $D=30$. 
Beyond this point, increasing the bond dimension leads to only negligible changes in the energy, essentially due to numerical errors, indicating that the MPS ansatz has reached numerical convergence. 
Based on this observation, we use a bond dimension of $D=80$ for all one-dimensional simulations, providing a comfortable margin while keeping the computational cost low.

The two-dimensional pure-gauge model exhibits a significantly slower convergence with bond dimension. 
As shown in \cref{fig:mps_convergence}~(b), the energy continues to decrease even at the largest bond dimension considered, $D=1000$, indicating that full convergence has not yet been reached. 
Nevertheless, the energy differences between successive bond dimensions become progressively smaller, suggesting that the approximation systematically improves with increasing $D$. 
This incomplete convergence does not affect the validity of our study, since the MPS calculations are used solely to obtain variational upper bounds on the ground-state energy. 
An insufficiently converged MPS can only produce a less tight upper bound on the ground-state energy and therefore cannot invalidate our methodology or conclusions.
Considering both the computational cost and the achieved accuracy, we choose $D=500$ as the default bond dimension for all two-dimensional simulations.

For both the one- and two-dimensional systems, the gauge symmetry was enforced by adding a penalty term to the Hamiltonian. 
The term was 
\begin{equation}
    \sum_i (G_i - 1)^2
\end{equation}
where $G_i$ is defined in ~\cref{eq:gauge-1d-bulk,eq:gauge-1d-boundary,eq:gauge_operators_2D}. 
The choice of $1$ is simply the choice of sector for the gauge constraint.
The coupling for this term was set as $100\max\{\abs{J}, \abs{\mu}, \abs{f}\}$ for the 1D case and $100\max\{\abs{f}, \abs{h}\}$ for the 2D case.
\Cref{fig:gauss_law} shows the expectation value of each gauge operator evaluated on the MPS state for both the 1D and 2D systems. 
As expected, all operators satisfy the physical sector constraint $\expval{G_i} = 1$ in each case.

\begin{figure}
    \centering
    \includegraphics[width=0.45\linewidth]{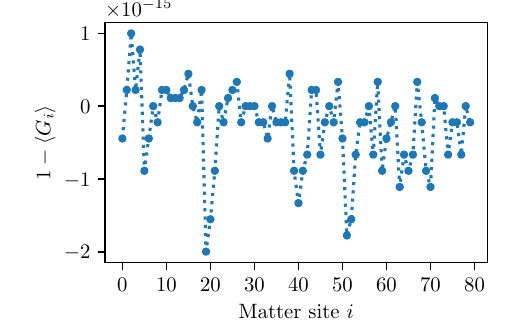}
    \hfill
    \includegraphics[width=0.45\linewidth]{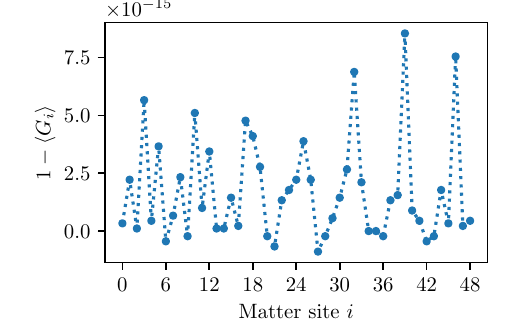} \\
    (a) \hspace{30em} (b)
    \caption{Deviation from the physical sector constraint, $1 - \expval{G_i}$, for each gauge operator $G_i$.
    (a) Results for the one-dimensional Kogut-Susskind Hamiltonian on a 159-spin system, obtained from an MPS simulation with bond dimension $D=80$.
    (b) Results for the two-dimensional pure-gauge Kogut-Susskind Hamiltonian on an 84-spin system, obtained from an MPS simulation with bond dimension $D=500$.}
    \label{fig:gauss_law}
\end{figure}

\section{Blocks of the SDP relaxation for the Kogut-Susskind model in (1+1)d}
\label{app:blocks_1D_model}

As explained in the main text, the block-ordered sequence at level 3 with interaction range 1 results into a highly-structured form.
We observe five types of blocks:

\hspace{1em}

\textbf{Class 1} (for each matter site $i$ not at the boundary): \\
\noindent
$\tau_{i-1,i}^z \sigma_{i}^z \tau_{i,i+1}^z$, $\tau_{i-1,i}^y \sigma_{i}^z \tau_{i,i+1}^z$, $\tau_{i-1,i}^z \tau_{i,i+1}^z$, $\tau_{i-1,i}^z \sigma_{i}^z \tau_{i,i+1}^y$, $\tau_{i-1,i}^y \sigma_{i}^z \tau_{i,i+1}^y$, $\tau_{i-1,i}^z \tau_{i,i+1}^y$, $\tau_{i-1,i}^y \tau_{i,i+1}^z$, $\tau_{i-1,i}^y \tau_{i,i+1}^y$.

\hspace{1em}

\textbf{Class 2} (for each matter site $i$ not at the boundary): \\
\noindent
$\tau_{i-1,i}^y \sigma_{i}^y \tau_{i,i+1}^y$, $\tau_{i-1,i}^z \sigma_{i}^y \tau_{i,i+1}^y$, $\tau_{i-1,i}^z \sigma_{i}^y \tau_{i,i+1}^z$, $\tau_{i-1,i}^z \sigma_{i}^x \tau_{i,i+1}^y$, $\tau_{i-1,i}^y \sigma_{i}^x \tau_{i,i+1}^y$, $\tau_{i-1,i}^z \sigma_{i}^x \tau_{i,i+1}^z$, $\tau_{i-1,i}^y \sigma_{i}^y \tau_{i,i+1}^z$, $\tau_{i-1,i}^y \sigma_{i}^x \tau_{i,i+1}^z$.

\hspace{1em}

\textbf{Class 3} (for each matter site $i$): \\
\noindent
$\sigma_{i}^y $, $\sigma_{i}^x $, $\tau_{i,i+1}^z \sigma_{i+1}^x $, $\tau_{i-1,i}^x \sigma_{i}^x $, $\sigma_{i-1}^y \tau_{i-1,i}^y $, $\sigma_{i}^y \tau_{i,i+1}^x $, $\tau_{i,i+1}^z \sigma_{i+1}^y $, $\sigma_{i}^y \sigma_{i+1}^z $, $\tau_{i,i+1}^y \sigma_{i+1}^x $, $\sigma_{i-1}^z \sigma_{i}^y $, $\sigma_{i}^x \tau_{i,i+1}^x $, $\sigma_{i-1}^x \tau_{i-1,i}^z $, $\sigma_{i}^x \sigma_{i+1}^z $, $\tau_{i,i+1}^y \sigma_{i+1}^y $, $\tau_{i-1,i}^x \sigma_{i}^y $, $\sigma_{i-1}^y \tau_{i-1,i}^z $, $\sigma_{i-1}^z \sigma_{i}^x $, $\sigma_{i-1}^x \tau_{i-1,i}^y $, $\tau_{i-1,i}^x \sigma_{i}^y \tau_{i,i+1}^x $, $\sigma_{i}^z \tau_{i,i+1}^y \sigma_{i+1}^x $, $\sigma_{i-1}^y \tau_{i-1,i}^y \sigma_{i}^z $, $\tau_{i-2,i-1}^x \sigma_{i-1}^x \tau_{i-1,i}^z $, $\tau_{i-1,i}^x \sigma_{i}^x \tau_{i,i+1}^x $, $\tau_{i,i+1}^z \sigma_{i+1}^y \tau_{i+1,i+2}^x $, $\sigma_{i}^x \tau_{i,i+1}^x \sigma_{i+1}^z $, $\sigma_{i}^z \tau_{i,i+1}^y \sigma_{i+1}^y $, $\tau_{i,i+1}^y \sigma_{i+1}^x \tau_{i+1,i+2}^x $, $\sigma_{i-1}^y \tau_{i-1,i}^z \sigma_{i}^z $, $\tau_{i-2,i-1}^x \sigma_{i-1}^y \tau_{i-1,i}^z $, $\tau_{i-2,i-1}^x \sigma_{i-1}^x \tau_{i-1,i}^y $, $\sigma_{i}^z \tau_{i,i+1}^z \sigma_{i+1}^x $, $\sigma_{i-1}^z \tau_{i-1,i}^x \sigma_{i}^y $, $\sigma_{i-1}^x \tau_{i-1,i}^y \sigma_{i}^z $, $\tau_{i-2,i-1}^x \sigma_{i-1}^y \tau_{i-1,i}^y $, $\sigma_{i-1}^x \tau_{i-1,i}^z \sigma_{i}^z $, $\tau_{i,i+1}^y \sigma_{i+1}^y \tau_{i+1,i+2}^x $, $\tau_{i,i+1}^z \sigma_{i+1}^x \tau_{i+1,i+2}^x $, $\sigma_{i}^z \tau_{i,i+1}^z \sigma_{i+1}^y $, $\sigma_{i-1}^z \tau_{i-1,i}^x \sigma_{i}^x $, $\sigma_{i}^y \tau_{i,i+1}^x \sigma_{i+1}^z $.

\hspace{1em}

\textbf{Class 4} (for each field site $(i,i+1)$): \\
\noindent
$\tau_{i,i+1}^y $, $\tau_{i,i+1}^z $, $\sigma_{i}^y \sigma_{i+1}^x $, $\sigma_{i}^y \sigma_{i+1}^y $, $\sigma_{i}^z \tau_{i,i+1}^y $, $\tau_{i-1,i}^x \tau_{i,i+1}^y $, $\tau_{i-1,i}^x \tau_{i,i+1}^z $, $\sigma_{i}^z \tau_{i,i+1}^z $, $\tau_{i,i+1}^y \sigma_{i+1}^z $, $\tau_{i,i+1}^y \tau_{i+1,i+2}^x $, $\tau_{i,i+1}^z \sigma_{i+1}^z $, $\sigma_{i}^x \sigma_{i+1}^x $, $\tau_{i,i+1}^z \tau_{i+1,i+2}^x $, $\sigma_{i}^x \sigma_{i+1}^y $, $\tau_{i,i+1}^y \sigma_{i+1}^z \tau_{i+1,i+2}^x $, $\sigma_{i}^z \tau_{i,i+1}^y \sigma_{i+1}^z $, $\sigma_{i}^x \tau_{i,i+1}^x \sigma_{i+1}^x $, $\sigma_{i}^x \tau_{i,i+1}^x \sigma_{i+1}^y $, $\sigma_{i}^y \tau_{i,i+1}^x \sigma_{i+1}^x $, $\tau_{i-1,i}^x \sigma_{i}^z \tau_{i,i+1}^y $, $\sigma_{i}^z \tau_{i,i+1}^z \sigma_{i+1}^z $, $\tau_{i-1,i}^x \sigma_{i}^z \tau_{i,i+1}^z $, $\sigma_{i}^y \tau_{i,i+1}^x \sigma_{i+1}^y $, $\tau_{i,i+1}^z \sigma_{i+1}^z \tau_{i+1,i+2}^x $.

\hspace{1em}

\textbf{Class 5:} all the remaining operators that are in the sequence (level 3 and interaction range 1) but do not appear in the previous blocks. This last class contributes to a single block.

\hspace{1em}

Counting the total number of blocks is straightforward: the first two classes contribute each $N-2$ blocks (one for each matter site, except the ones at the boundary).
The third class contributes $N$ blocks (one for each matter site) and the fourth class contributes $N-1$ blocks (one for each field site).
The final class contributes a single block, yielding a total of  $2(N-2) + N + (N-1) + 1 = 4(N-1)$ blocks.

However, enumerating the size of these blocks is more complex for two reasons. 
First, the explicit size of the blocks generated by the fifth class is not immediately known. 
In practice, this could be calculated by enumerating all possible Pauli strings in the sequence at level 3 with range 1 and subtracting the sizes of the blocks from the preceding classes. 
Second, open boundary conditions imply that blocks associated with matter sites at the boundaries contain fewer terms; specifically, any term in the third class involving a Pauli operator with support extending outside the 1D chain is discarded. 
Consequently, while estimating the block sizes remains feasible in practice, the presence of boundary effects complicates the derivation.

\section{Blocks of the SDP relaxation for the Kogut-Susskind model in (2+1)d } \label{app:blocks_2D_model}
For interaction range~1, the level-3 block-ordered sequence exhibits a highly structured organization. 
In this case, we identify twelve distinct classes of subsequences. 
As in the lower levels, each subsequence can be uniquely assigned to one of these classes when applied to a vertex $(i,j)$. 

Each class admits a geometric representation of the monomials contained in the corresponding subsequence. 
As in the previous cases, Pauli strings are associated with nodes on the lattice, while monomials involving two Pauli operators are represented by edges connecting the corresponding nodes. 
In addition, at level~3, monomials involving three Pauli operators are represented by hyper-edges.
On \Cref{fig:geom_rep} these are depicted as yellow triangles connecting the three corresponding lattice sites. 
For example, the Pauli string $\tau_{(i-1,j),(i,j)}^y \, \tau_{(i,j-1),(i,j)}^z$ is represented by an edge connecting the lattice edges $((i-1,j),(i,j))$ and $((i,j-1),(i,j))$. 
Similarly, three-body monomials are represented by triangular hyper-edges connecting the three involved lattice points.

The twelve block types are listed explicitly below together with their corresponding geometric representations. 
Due to open boundary conditions, some of the monomials appearing in the definitions of these blocks may have support extending beyond the lattice. 
Such monomials are omitted when constructing the block-ordered sequence.

\begin{flushleft}
\textbf{Class 1:} \\
\noindent
$\tau_{(i-1,j),(i,j)}^y \tau_{(i,j-1),(i,j)}^z \tau_{(i,j),(i,j+1)}^z $, $\tau_{(i-1,j),(i,j)}^y \tau_{(i,j-1),(i,j)}^z \tau_{(i,j),(i,j+1)}^y $, $\tau_{(i-1,j),(i,j)}^z \tau_{(i,j-1),(i,j)}^y \tau_{(i,j),(i,j+1)}^z $, $\tau_{(i-1,j),(i,j)}^z \tau_{(i,j-1),(i,j)}^y \tau_{(i,j),(i,j+1)}^y $, $\tau_{(i-1,j),(i,j)}^z \tau_{(i,j-1),(i,j)}^z \tau_{(i,j),(i,j+1)}^y $, $\tau_{(i-1,j),(i,j)}^z \tau_{(i,j-1),(i,j)}^z \tau_{(i,j),(i,j+1)}^z $, $\tau_{(i-1,j),(i,j)}^y \tau_{(i,j-1),(i,j)}^y \tau_{(i,j),(i,j+1)}^y $, $\tau_{(i-1,j),(i,j)}^y\tau_{(i,j-1),(i,j)}^y \tau_{(i,j),(i,j+1)}^z$
\end{flushleft}

\hspace{1em}

\begin{flushleft}
\textbf{Class 2:} \\
\noindent
$\tau_{(i-1,j),(i,j)}^z \tau_{(i,j),(i,j+1)}^y \tau_{(i,j),(i+1,j)}^y $, $\tau_{(i-1,j),(i,j)}^z \tau_{(i,j),(i,j+1)}^z \tau_{(i,j),(i+1,j)}^y $, $\tau_{(i-1,j),(i,j)}^z \tau_{(i,j),(i,j+1)}^z \tau_{(i,j),(i+1,j)}^z $, $\tau_{(i-1,j),(i,j)}^y \tau_{(i,j),(i,j+1)}^z \tau_{(i,j),(i+1,j)}^z $, $\tau_{(i-1,j),(i,j)}^z \tau_{(i,j),(i,j+1)}^y \tau_{(i,j),(i+1,j)}^z $, $\tau_{(i-1,j),(i,j)}^y \tau_{(i,j),(i,j+1)}^y \tau_{(i,j),(i+1,j)}^z $, $\tau_{(i-1,j),(i,j)}^y \tau_{(i,j),(i,j+1)}^y \tau_{(i,j),(i+1,j)}^y $, $\tau_{(i-1,j),(i,j)}^y\tau_{(i,j),(i,j+1)}^z \tau_{(i,j),(i+1,j)}^y $
\end{flushleft}

\hspace{1em}

\begin{flushleft}
\textbf{Class 3:} \\
\noindent
$\tau_{(i,j-1),(i,j)}^y \tau_{(i,j),(i,j+1)}^z \tau_{(i,j),(i+1,j)}^z $, $\tau_{(i,j-1),(i,j)}^z \tau_{(i,j),(i,j+1)}^z \tau_{(i,j),(i+1,j)}^z $, $\tau_{(i,j-1),(i,j)}^y \tau_{(i,j),(i,j+1)}^y \tau_{(i,j),(i+1,j)}^y $, $\tau_{(i,j-1),(i,j)}^z \tau_{(i,j),(i,j+1)}^y \tau_{(i,j),(i+1,j)}^y $, $\tau_{(i,j-1),(i,j)}^z \tau_{(i,j),(i,j+1)}^y \tau_{(i,j),(i+1,j)}^z $, $\tau_{(i,j-1),(i,j)}^y \tau_{(i,j),(i,j+1)}^y \tau_{(i,j),(i+1,j)}^z $, $\tau_{(i,j-1),(i,j)}^z \tau_{(i,j),(i,j+1)}^z \tau_{(i,j),(i+1,j)}^y $, $\tau_{(i,j-1),(i,j)}^y\tau_{(i,j),(i,j+1)}^z \tau_{(i,j),(i+1,j)}^y $
\end{flushleft}

\hspace{1em}

\begin{flushleft}
\textbf{Class 4:} \\
\noindent
$\tau_{(i-1,j),(i,j)}^y \tau_{(i,j-1),(i,j)}^y \tau_{(i,j),(i+1,j)}^z $, $\tau_{(i-1,j),(i,j)}^z \tau_{(i,j-1),(i,j)}^z \tau_{(i,j),(i+1,j)}^y $, $\tau_{(i-1,j),(i,j)}^z \tau_{(i,j-1),(i,j)}^y \tau_{(i,j),(i+1,j)}^z $, $\tau_{(i-1,j),(i,j)}^y \tau_{(i,j-1),(i,j)}^z \tau_{(i,j),(i+1,j)}^y $, $\tau_{(i-1,j),(i,j)}^z \tau_{(i,j-1),(i,j)}^z \tau_{(i,j),(i+1,j)}^z $, $\tau_{(i-1,j),(i,j)}^y \tau_{(i,j-1),(i,j)}^y \tau_{(i,j),(i+1,j)}^y $, $\tau_{(i-1,j),(i,j)}^y \tau_{(i,j-1),(i,j)}^z \tau_{(i,j),(i+1,j)}^z $, $\tau_{(i-1,j),(i,j)}^z\tau_{(i,j-1),(i,j)}^y \tau_{(i,j),(i+1,j)}^y$
\end{flushleft}

\hspace{1em}

\begin{flushleft}
\textbf{Class 5:} \\
\noindent
$\tau_{(i,j-1),(i,j)}^y \tau_{(i,j),(i,j+1)}^y $, $\tau_{(i-1,j),(i,j)}^x \tau_{(i,j-1),(i,j)}^z \tau_{(i,j),(i,j+1)}^y $, $\tau_{(i,j-1),(i,j)}^z \tau_{(i,j),(i,j+1)}^y \tau_{(i,j),(i+1,j)}^x $, $\tau_{(i,j-1),(i,j)}^z \tau_{(i,j),(i,j+1)}^z \tau_{(i,j),(i+1,j)}^x $, $\tau_{(i-1,j),(i,j)}^x \tau_{(i,j-1),(i,j)}^z \tau_{(i,j),(i,j+1)}^z $, $\tau_{(i,j-1),(i,j)}^y \tau_{(i,j),(i,j+1)}^z $, $\tau_{(i,j-1),(i,j)}^z \tau_{(i,j),(i,j+1)}^z $, $\tau_{(i-1,j),(i,j)}^x \tau_{(i,j-1),(i,j)}^y \tau_{(i,j),(i,j+1)}^z $, $\tau_{(i,j-1),(i,j)}^y \tau_{(i,j),(i,j+1)}^z \tau_{(i,j),(i+1,j)}^x $, $\tau_{(i,j-1),(i,j)}^y \tau_{(i,j),(i,j+1)}^y \tau_{(i,j),(i+1,j)}^x $, $\tau_{(i-1,j),(i,j)}^x \tau_{(i,j-1),(i,j)}^y \tau_{(i,j),(i,j+1)}^y $, $\tau_{(i,j-1),(i,j)}^z \tau_{(i,j),(i,j+1)}^y$
\end{flushleft}

\hspace{1em}

\begin{flushleft}
\textbf{Class 6:} \\
\noindent
$\tau_{(i-1,j),(i,j)}^y \tau_{(i,j-1),(i,j)}^x \tau_{(i,j),(i+1,j)}^z $, $\tau_{(i-1,j),(i,j)}^z \tau_{(i,j-1),(i,j)}^x \tau_{(i,j),(i+1,j)}^y $, $\tau_{(i-1,j),(i,j)}^z \tau_{(i,j),(i+1,j)}^z $, $\tau_{(i-1,j),(i,j)}^y \tau_{(i,j),(i,j+1)}^x \tau_{(i,j),(i+1,j)}^z $, $\tau_{(i-1,j),(i,j)}^y \tau_{(i,j),(i+1,j)}^y $, $\tau_{(i-1,j),(i,j)}^z \tau_{(i,j),(i+1,j)}^y $, $\tau_{(i-1,j),(i,j)}^y \tau_{(i,j),(i,j+1)}^x \tau_{(i,j),(i+1,j)}^y $, $\tau_{(i-1,j),(i,j)}^z \tau_{(i,j),(i,j+1)}^x \tau_{(i,j),(i+1,j)}^y $, $\tau_{(i-1,j),(i,j)}^y \tau_{(i,j),(i+1,j)}^z $, $\tau_{(i-1,j),(i,j)}^z \tau_{(i,j),(i,j+1)}^x \tau_{(i,j),(i+1,j)}^z $, $\tau_{(i-1,j),(i,j)}^y \tau_{(i,j-1),(i,j)}^x \tau_{(i,j),(i+1,j)}^y $, $\tau_{(i-1,j),(i,j)}^z \tau_{(i,j-1),(i,j)}^x \tau_{(i,j),(i+1,j)}^z$
\end{flushleft}

\hspace{1em}

\begin{flushleft}
\textbf{Class 7:} \\
\noindent
$\tau_{(i,j),(i,j+1)}^y \tau_{(i+1,j),(i+1,j+1)}^y $, $\tau_{(i,j),(i,j+1)}^z \tau_{(i+1,j),(i+1,j+1)}^y $, $\tau_{(i,j),(i,j+1)}^x \tau_{(i,j),(i+1,j)}^z \tau_{(i,j+1),(i+1,j+1)}^z $, $\tau_{(i,j),(i,j+1)}^y \tau_{(i,j),(i+1,j)}^x \tau_{(i+1,j),(i+1,j+1)}^y $, $\tau_{(i,j),(i+1,j)}^z \tau_{(i,j+1),(i+1,j+1)}^y \tau_{(i+1,j),(i+1,j+1)}^x $, $\tau_{(i,j),(i+1,j)}^y \tau_{(i,j+1),(i+1,j+1)}^y \tau_{(i+1,j),(i+1,j+1)}^x $, $\tau_{(i,j),(i,j+1)}^y \tau_{(i,j),(i+1,j)}^x \tau_{(i+1,j),(i+1,j+1)}^z $, $\tau_{(i,j),(i+1,j)}^z \tau_{(i,j+1),(i+1,j+1)}^z $, $\tau_{(i,j),(i+1,j)}^z \tau_{(i,j+1),(i+1,j+1)}^y $, $\tau_{(i,j),(i,j+1)}^y \tau_{(i,j+1),(i+1,j+1)}^x \tau_{(i+1,j),(i+1,j+1)}^y $, $\tau_{(i,j),(i,j+1)}^x \tau_{(i,j),(i+1,j)}^y \tau_{(i,j+1),(i+1,j+1)}^y $, $\tau_{(i,j),(i,j+1)}^y \tau_{(i,j+1),(i+1,j+1)}^x \tau_{(i+1,j),(i+1,j+1)}^z $, $\tau_{(i,j),(i,j+1)}^x \tau_{(i,j),(i+1,j)}^z \tau_{(i,j+1),(i+1,j+1)}^y $, $\tau_{(i,j),(i+1,j)}^z \tau_{(i,j+1),(i+1,j+1)}^z \tau_{(i+1,j),(i+1,j+1)}^x $, $\tau_{(i,j),(i,j+1)}^z \tau_{(i+1,j),(i+1,j+1)}^z $, $\tau_{(i,j),(i+1,j)}^y \tau_{(i,j+1),(i+1,j+1)}^y $, $\tau_{(i,j),(i,j+1)}^y \tau_{(i+1,j),(i+1,j+1)}^z $, $\tau_{(i,j),(i+1,j)}^y \tau_{(i,j+1),(i+1,j+1)}^z $, $\tau_{(i,j),(i,j+1)}^z \tau_{(i,j+1),(i+1,j+1)}^x \tau_{(i+1,j),(i+1,j+1)}^y $, $\tau_{(i,j),(i,j+1)}^x \tau_{(i,j),(i+1,j)}^y \tau_{(i,j+1),(i+1,j+1)}^z $, $\tau_{(i,j),(i,j+1)}^z \tau_{(i,j+1),(i+1,j+1)}^x \tau_{(i+1,j),(i+1,j+1)}^z $, $\tau_{(i,j),(i,j+1)}^z \tau_{(i,j),(i+1,j)}^x \tau_{(i+1,j),(i+1,j+1)}^y $, $\tau_{(i,j),(i+1,j)}^y \tau_{(i,j+1),(i+1,j+1)}^z \tau_{(i+1,j),(i+1,j+1)}^x $, $\tau_{(i,j),(i,j+1)}^z \tau_{(i,j),(i+1,j)}^x \tau_{(i+1,j),(i+1,j+1)}^z$
\end{flushleft}

\hspace{1em}

\begin{flushleft}
\textbf{Class 8:} \\
\noindent
$\tau_{(i,j+1),(i+1,j+1)}^y \tau_{(i+1,j),(i+1,j+1)}^y \tau_{(i+1,j+1),(i+1,j+2)}^x $, $\tau_{(i,j+1),(i+1,j+1)}^y \tau_{(i+1,j),(i+1,j+1)}^z \tau_{(i+1,j+1),(i+1,j+2)}^x $, $\tau_{(i,j+1),(i+1,j+1)}^z \tau_{(i+1,j),(i+1,j+1)}^y \tau_{(i+1,j+1),(i+2,j+1)}^x $, $\tau_{(i,j-1),(i,j)}^x \tau_{(i,j),(i,j+1)}^y \tau_{(i,j),(i+1,j)}^y $, $\tau_{(i,j),(i,j+1)}^y \tau_{(i,j),(i+1,j)}^y \tau_{(i+1,j),(i+1,j+1)}^x $, $\tau_{(i,j),(i,j+1)}^y \tau_{(i,j),(i+1,j)}^z \tau_{(i+1,j),(i+1,j+1)}^x $, $\tau_{(i,j+1),(i+1,j+1)}^y \tau_{(i+1,j),(i+1,j+1)}^z $, $\tau_{(i,j),(i,j+1)}^z \tau_{(i,j),(i+1,j)}^z $, $\tau_{(i,j),(i,j+1)}^y \tau_{(i,j),(i+1,j)}^z $, $\tau_{(i,j+1),(i+1,j+1)}^y \tau_{(i+1,j),(i+1,j+1)}^z \tau_{(i+1,j+1),(i+2,j+1)}^x $, $\tau_{(i,j),(i,j+1)}^y \tau_{(i,j),(i+1,j)}^z \tau_{(i,j+1),(i+1,j+1)}^x $, $\tau_{(i,j+1),(i+1,j+1)}^z \tau_{(i+1,j),(i+1,j+1)}^z \tau_{(i+1,j+1),(i+1,j+2)}^x $, $\tau_{(i,j-1),(i,j)}^x \tau_{(i,j),(i,j+1)}^z \tau_{(i,j),(i+1,j)}^y $, $\tau_{(i,j+1),(i+1,j+1)}^z \tau_{(i+1,j),(i+1,j+1)}^y $, $\tau_{(i,j+1),(i+1,j+1)}^y \tau_{(i+1,j),(i+1,j+1)}^y $, $\tau_{(i-1,j),(i,j)}^x \tau_{(i,j),(i,j+1)}^y \tau_{(i,j),(i+1,j)}^y $, $\tau_{(i,j),(i+1,j)}^x \tau_{(i,j+1),(i+1,j+1)}^y \tau_{(i+1,j),(i+1,j+1)}^y $, $\tau_{(i,j),(i,j+1)}^z \tau_{(i,j),(i+1,j)}^z \tau_{(i,j+1),(i+1,j+1)}^x $, $\tau_{(i,j),(i,j+1)}^x \tau_{(i,j+1),(i+1,j+1)}^z \tau_{(i+1,j),(i+1,j+1)}^z $, $\tau_{(i,j+1),(i+1,j+1)}^z \tau_{(i+1,j),(i+1,j+1)}^y \tau_{(i+1,j+1),(i+1,j+2)}^x $, $\tau_{(i-1,j),(i,j)}^x \tau_{(i,j),(i,j+1)}^z \tau_{(i,j),(i+1,j)}^z $, $\tau_{(i,j),(i,j+1)}^z \tau_{(i,j),(i+1,j)}^y \tau_{(i+1,j),(i+1,j+1)}^x $, $\tau_{(i,j-1),(i,j)}^x \tau_{(i,j),(i,j+1)}^y \tau_{(i,j),(i+1,j)}^z $, $\tau_{(i,j-1),(i,j)}^x \tau_{(i,j),(i,j+1)}^z \tau_{(i,j),(i+1,j)}^z $,$\tau_{(i,j),(i+1,j)}^x \tau_{(i,j+1),(i+1,j+1)}^y \tau_{(i+1,j),(i+1,j+1)}^z $, $\tau_{(i,j),(i,j+1)}^x \tau_{(i,j+1),(i+1,j+1)}^y \tau_{(i+1,j),(i+1,j+1)}^z $, $\tau_{(i,j),(i,j+1)}^y \tau_{(i,j),(i+1,j)}^y \tau_{(i,j+1),(i+1,j+1)}^x $, $\tau_{(i,j),(i,j+1)}^z \tau_{(i,j),(i+1,j)}^y \tau_{(i,j+1),(i+1,j+1)}^x $, $\tau_{(i,j+1),(i+1,j+1)}^z \tau_{(i+1,j),(i+1,j+1)}^z \tau_{(i+1,j+1),(i+2,j+1)}^x $, $\tau_{(i,j),(i,j+1)}^z \tau_{(i,j),(i+1,j)}^y $, $\tau_{(i,j),(i,j+1)}^z \tau_{(i,j),(i+1,j)}^z \tau_{(i+1,j),(i+1,j+1)}^x $, $\tau_{(i-1,j),(i,j)}^x \tau_{(i,j),(i,j+1)}^z \tau_{(i,j),(i+1,j)}^y $, $\tau_{(i,j+1),(i+1,j+1)}^y \tau_{(i+1,j),(i+1,j+1)}^y \tau_{(i+1,j+1),(i+2,j+1)}^x $, $\tau_{(i,j),(i,j+1)}^x \tau_{(i,j+1),(i+1,j+1)}^z \tau_{(i+1,j),(i+1,j+1)}^y $, $\tau_{(i,j),(i,j+1)}^x \tau_{(i,j+1),(i+1,j+1)}^y \tau_{(i+1,j),(i+1,j+1)}^y $, $\tau_{(i,j),(i+1,j)}^x \tau_{(i,j+1),(i+1,j+1)}^z \tau_{(i+1,j),(i+1,j+1)}^y $, $\tau_{(i,j),(i+1,j)}^x \tau_{(i,j+1),(i+1,j+1)}^z \tau_{(i+1,j),(i+1,j+1)}^z $, $\tau_{(i,j),(i,j+1)}^y \tau_{(i,j),(i+1,j)}^y $, $\tau_{(i-1,j),(i,j)}^x \tau_{(i,j),(i,j+1)}^y \tau_{(i,j),(i+1,j)}^z $, $\tau_{(i,j+1),(i+1,j+1)}^z \tau_{(i+1,j),(i+1,j+1)}^z$
\end{flushleft}

\hspace{1em}

\begin{flushleft}
\textbf{Class 9:} \\
\noindent
$\tau_{(i-1,j+1),(i,j+1)}^x \tau_{(i,j),(i,j+1)}^z \tau_{(i,j+1),(i+1,j+1)}^y $, $\tau_{(i,j),(i,j+1)}^y \tau_{(i,j+1),(i+1,j+1)}^z \tau_{(i+1,j),(i+1,j+1)}^x $, $\tau_{(i,j),(i,j+1)}^z \tau_{(i,j+1),(i+1,j+1)}^y $, $\tau_{(i,j),(i,j+1)}^x \tau_{(i,j),(i+1,j)}^y \tau_{(i+1,j),(i+1,j+1)}^y $, $\tau_{(i,j),(i,j+1)}^y \tau_{(i,j+1),(i+1,j+1)}^y \tau_{(i+1,j),(i+1,j+1)}^x $, $\tau_{(i,j),(i,j+1)}^x \tau_{(i,j),(i+1,j)}^z \tau_{(i+1,j),(i+1,j+1)}^y $, $\tau_{(i,j),(i,j+1)}^y \tau_{(i,j+1),(i,j+2)}^x \tau_{(i,j+1),(i+1,j+1)}^y $, $\tau_{(i,j),(i+1,j)}^y \tau_{(i+1,j-1),(i+1,j)}^x \tau_{(i+1,j),(i+1,j+1)}^y $, $\tau_{(i,j),(i,j+1)}^z \tau_{(i,j+1),(i,j+2)}^x \tau_{(i,j+1),(i+1,j+1)}^y $, $\tau_{(i,j),(i,j+1)}^z \tau_{(i,j+1),(i+1,j+1)}^z $, $\tau_{(i,j),(i+1,j)}^y \tau_{(i+1,j),(i+1,j+1)}^z \tau_{(i+1,j),(i+2,j)}^x $, $\tau_{(i,j),(i+1,j)}^z \tau_{(i+1,j-1),(i+1,j)}^x \tau_{(i+1,j),(i+1,j+1)}^y $, $\tau_{(i,j),(i,j+1)}^x \tau_{(i,j),(i+1,j)}^z \tau_{(i+1,j),(i+1,j+1)}^z $, $\tau_{(i,j),(i+1,j)}^y \tau_{(i+1,j),(i+1,j+1)}^y $, $\tau_{(i,j),(i,j+1)}^y \tau_{(i,j+1),(i+1,j+1)}^y $, $\tau_{(i,j),(i+1,j)}^z \tau_{(i+1,j),(i+1,j+1)}^z \tau_{(i+1,j),(i+2,j)}^x $, $\tau_{(i,j),(i+1,j)}^y \tau_{(i+1,j),(i+1,j+1)}^z $, $\tau_{(i,j),(i,j+1)}^x \tau_{(i,j),(i+1,j)}^y \tau_{(i+1,j),(i+1,j+1)}^z $, $\tau_{(i,j),(i,j+1)}^z \tau_{(i,j+1),(i+1,j+1)}^z \tau_{(i+1,j),(i+1,j+1)}^x $, $\tau_{(i,j),(i,j+1)}^y \tau_{(i,j+1),(i+1,j+1)}^z $, $\tau_{(i,j),(i+1,j)}^z \tau_{(i+1,j-1),(i+1,j)}^x \tau_{(i+1,j),(i+1,j+1)}^z $, $\tau_{(i-1,j+1),(i,j+1)}^x \tau_{(i,j),(i,j+1)}^y \tau_{(i,j+1),(i+1,j+1)}^z $, $\tau_{(i,j),(i+1,j)}^y \tau_{(i+1,j-1),(i+1,j)}^x \tau_{(i+1,j),(i+1,j+1)}^z $, $\tau_{(i,j),(i,j+1)}^z \tau_{(i,j),(i+1,j)}^x \tau_{(i,j+1),(i+1,j+1)}^z $, $\tau_{(i-1,j+1),(i,j+1)}^x \tau_{(i,j),(i,j+1)}^y \tau_{(i,j+1),(i+1,j+1)}^y $, $\tau_{(i,j),(i,j+1)}^z \tau_{(i,j+1),(i+1,j+1)}^y \tau_{(i+1,j),(i+1,j+1)}^x $, $\tau_{(i,j),(i,j+1)}^z \tau_{(i,j),(i+1,j)}^x \tau_{(i,j+1),(i+1,j+1)}^y $, $\tau_{(i,j),(i+1,j)}^y \tau_{(i,j+1),(i+1,j+1)}^x \tau_{(i+1,j),(i+1,j+1)}^z $, $\tau_{(i,j),(i,j+1)}^y \tau_{(i,j),(i+1,j)}^x \tau_{(i,j+1),(i+1,j+1)}^z $, $\tau_{(i,j),(i+1,j)}^z \tau_{(i+1,j),(i+1,j+1)}^y $, $\tau_{(i,j),(i,j+1)}^y \tau_{(i,j+1),(i,j+2)}^x \tau_{(i,j+1),(i+1,j+1)}^z $, $\tau_{(i,j),(i+1,j)}^z \tau_{(i,j+1),(i+1,j+1)}^x \tau_{(i+1,j),(i+1,j+1)}^z $, $\tau_{(i-1,j+1),(i,j+1)}^x \tau_{(i,j),(i,j+1)}^z \tau_{(i,j+1),(i+1,j+1)}^z $, $\tau_{(i,j),(i,j+1)}^y \tau_{(i,j),(i+1,j)}^x\tau_{(i,j+1),(i+1,j+1)}^y $, $\tau_{(i,j),(i+1,j)}^z \tau_{(i+1,j),(i+1,j+1)}^y \tau_{(i+1,j),(i+2,j)}^x $, $\tau_{(i,j),(i+1,j)}^z \tau_{(i,j+1),(i+1,j+1)}^x \tau_{(i+1,j),(i+1,j+1)}^y $, $\tau_{(i,j),(i,j+1)}^z \tau_{(i,j+1),(i,j+2)}^x \tau_{(i,j+1),(i+1,j+1)}^z $, $\tau_{(i,j),(i+1,j)}^z \tau_{(i+1,j),(i+1,j+1)}^z $, $\tau_{(i,j),(i+1,j)}^y \tau_{(i+1,j),(i+1,j+1)}^y \tau_{(i+1,j),(i+2,j)}^x $, $\tau_{(i,j),(i+1,j)}^y \tau_{(i,j+1),(i+1,j+1)}^x \tau_{(i+1,j),(i+1,j+1)}^y$
\end{flushleft}

\hspace{1em}

\begin{flushleft}
\textbf{Class 10:} \\
\noindent
$\tau_{(i,j),(i,j+1)}^z \tau_{(i,j+1),(i+1,j+1)}^z \tau_{(i+1,j),(i+1,j+1)}^z $, $\tau_{(i,j),(i,j+1)}^z \tau_{(i,j+1),(i+1,j+1)}^z \tau_{(i+1,j),(i+1,j+1)}^y $, $\tau_{(i,j),(i+1,j)}^z \tau_{(i,j+1),(i+1,j+1)}^x \tau_{(i+1,j),(i+1,j+1)}^x $, $\tau_{(i,j),(i+1,j)}^y \tau_{(i+1,j),(i+2,j)}^x $, $\tau_{(i,j),(i,j+1)}^x \tau_{(i,j),(i+1,j)}^z \tau_{(i,j+1),(i+1,j+1)}^x $, $\tau_{(i,j),(i+1,j)}^y \tau_{(i+1,j),(i+1,j+1)}^x $, $\tau_{(i,j-1),(i,j)}^x \tau_{(i,j),(i+1,j)}^z \tau_{(i+1,j-1),(i+1,j)}^x $, $\tau_{(i,j),(i,j+1)}^y \tau_{(i,j+1),(i+1,j+1)}^z \tau_{(i+1,j),(i+1,j+1)}^y $, $\tau_{(i-1,j),(i,j)}^x \tau_{(i,j),(i,j+1)}^x \tau_{(i,j),(i+1,j)}^z $, $\tau_{(i,j-1),(i,j)}^z \tau_{(i,j-1),(i+1,j-1)}^z \tau_{(i+1,j-1),(i+1,j)}^y $, $\tau_{(i,j),(i+1,j)}^z \tau_{(i+1,j),(i+1,j+1)}^x \tau_{(i+1,j),(i+2,j)}^x $, $\tau_{(i,j-1),(i,j)}^x \tau_{(i,j-1),(i+1,j-1)}^x \tau_{(i,j),(i+1,j)}^y$, $\tau_{(i,j-1),(i,j)}^z \tau_{(i,j-1),(i+1,j-1)}^y \tau_{(i+1,j-1),(i+1,j)}^z $, $\tau_{(i,j),(i+1,j)}^y $, $\tau_{(i,j),(i+1,j)}^z \tau_{(i+1,j-1),(i+1,j)}^x \tau_{(i+1,j),(i+1,j+1)}^x $, $\tau_{(i,j-1),(i,j)}^y \tau_{(i,j-1),(i+1,j-1)}^z \tau_{(i+1,j-1),(i+1,j)}^y $, $\tau_{(i,j-1),(i,j)}^x \tau_{(i,j),(i+1,j)}^z $, $\tau_{(i,j-1),(i,j)}^y \tau_{(i,j-1),(i+1,j-1)}^y \tau_{(i+1,j-1),(i+1,j)}^z $, $\tau_{(i,j),(i+1,j)}^y \tau_{(i+1,j-1),(i+1,j)}^x $, $\tau_{(i,j-1),(i+1,j-1)}^x \tau_{(i,j),(i+1,j)}^y $, $\tau_{(i,j-1),(i,j)}^x \tau_{(i,j),(i+1,j)}^y \tau_{(i+1,j-1),(i+1,j)}^x $, $\tau_{(i,j),(i+1,j)}^z \tau_{(i+1,j-1),(i+1,j)}^x \tau_{(i+1,j),(i+2,j)}^x $, $\tau_{(i,j),(i+1,j)}^y \tau_{(i+1,j),(i+1,j+1)}^x \tau_{(i+1,j),(i+2,j)}^x $, $\tau_{(i,j),(i,j+1)}^y \tau_{(i,j+1),(i+1,j+1)}^y \tau_{(i+1,j),(i+1,j+1)}^z $, $\tau_{(i-1,j),(i,j)}^x \tau_{(i,j),(i+1,j)}^y $, $\tau_{(i,j),(i,j+1)}^z \tau_{(i,j+1),(i+1,j+1)}^y \tau_{(i+1,j),(i+1,j+1)}^z $, $\tau_{(i,j),(i,j+1)}^z \tau_{(i,j+1),(i+1,j+1)}^y \tau_{(i+1,j),(i+1,j+1)}^y $, $\tau_{(i,j-1),(i+1,j-1)}^x\tau_{(i,j),(i+1,j)}^y \tau_{(i+1,j-1),(i+1,j)}^x $, $\tau_{(i,j),(i+1,j)}^z \tau_{(i+1,j),(i+2,j)}^x $, $\tau_{(i,j),(i,j+1)}^x \tau_{(i,j),(i+1,j)}^y \tau_{(i,j+1),(i+1,j+1)}^x $, $\tau_{(i,j-1),(i,j)}^x \tau_{(i,j),(i,j+1)}^x \tau_{(i,j),(i+1,j)}^y $, $\tau_{(i,j),(i,j+1)}^x \tau_{(i,j),(i+1,j)}^y $, $\tau_{(i,j),(i,j+1)}^y \tau_{(i,j+1),(i+1,j+1)}^y\tau_{(i+1,j),(i+1,j+1)}^y $, $\tau_{(i,j),(i+1,j)}^y \tau_{(i+1,j-1),(i+1,j)}^x \tau_{(i+1,j),(i+2,j)}^x $, $\tau_{(i,j),(i+1,j)}^z $, $\tau_{(i,j-1),(i,j)}^y \tau_{(i,j-1),(i+1,j-1)}^z \tau_{(i+1,j-1),(i+1,j)}^z $, $\tau_{(i-1,j),(i,j)}^x \tau_{(i,j-1),(i,j)}^x \tau_{(i,j),(i+1,j)}^y $, $\tau_{(i,j),(i,j+1)}^x \tau_{(i,j),(i+1,j)}^z \tau_{(i+1,j),(i+1,j+1)}^x $, $\tau_{(i,j-1),(i,j)}^y \tau_{(i,j-1),(i+1,j-1)}^y \tau_{(i+1,j-1),(i+1,j)}^y $, $\tau_{(i,j),(i+1,j)}^y \tau_{(i,j+1),(i+1,j+1)}^x \tau_{(i+1,j),(i+1,j+1)}^x $, $\tau_{(i,j-1),(i,j)}^x \tau_{(i,j),(i,j+1)}^x \tau_{(i,j),(i+1,j)}^z $, $\tau_{(i,j),(i+1,j)}^z \tau_{(i+1,j),(i+1,j+1)}^x $, $\tau_{(i,j),(i+1,j)}^z \tau_{(i,j+1),(i+1,j+1)}^x $, $\tau_{(i,j-1),(i+1,j-1)}^x \tau_{(i,j),(i+1,j)}^z $, $\tau_{(i,j-1),(i,j)}^x \tau_{(i,j-1),(i+1,j-1)}^x \tau_{(i,j),(i+1,j)}^z $, $\tau_{(i-1,j),(i,j)}^x \tau_{(i,j),(i+1,j)}^z $, $\tau_{(i,j),(i,j+1)}^x \tau_{(i,j),(i+1,j)}^y \tau_{(i+1,j),(i+1,j+1)}^x $, $\tau_{(i-1,j),(i,j)}^x \tau_{(i,j-1),(i,j)}^x \tau_{(i,j),(i+1,j)}^z $, $\tau_{(i,j-1),(i,j)}^x \tau_{(i,j),(i+1,j)}^y $, $\tau_{(i,j-1),(i+1,j-1)}^x \tau_{(i,j),(i+1,j)}^z \tau_{(i+1,j-1),(i+1,j)}^x $, $\tau_{(i,j),(i+1,j)}^z \tau_{(i+1,j-1),(i+1,j)}^x $, $\tau_{(i,j),(i,j+1)}^x \tau_{(i,j),(i+1,j)}^z $, $\tau_{(i-1,j),(i,j)}^x \tau_{(i,j),(i,j+1)}^x \tau_{(i,j),(i+1,j)}^y $, $\tau_{(i,j-1),(i,j)}^z \tau_{(i,j-1),(i+1,j-1)}^z \tau_{(i+1,j-1),(i+1,j)}^z $, $\tau_{(i,j),(i+1,j)}^y \tau_{(i,j+1),(i+1,j+1)}^x $, $\tau_{(i,j-1),(i,j)}^z \tau_{(i,j-1),(i+1,j-1)}^y \tau_{(i+1,j-1),(i+1,j)}^y $, $\tau_{(i,j),(i,j+1)}^y \tau_{(i,j+1),(i+1,j+1)}^z \tau_{(i+1,j),(i+1,j+1)}^z $, $\tau_{(i,j),(i+1,j)}^y \tau_{(i+1,j-1),(i+1,j)}^x \tau_{(i+1,j),(i+1,j+1)}^x$
\end{flushleft}

\hspace{1em}

\begin{flushleft}
\textbf{Class 11:} \\
\noindent
$\tau_{(i,j),(i,j+1)}^y \tau_{(i,j),(i+1,j)}^x \tau_{(i,j+1),(i+1,j+1)}^x $, $\tau_{(i,j),(i,j+1)}^z \tau_{(i,j),(i+1,j)}^x \tau_{(i,j+1),(i+1,j+1)}^x $, $\tau_{(i-1,j),(i-1,j+1)}^y \tau_{(i-1,j),(i,j)}^y \tau_{(i-1,j+1),(i,j+1)}^z $, $\tau_{(i,j),(i+1,j)}^z \tau_{(i,j+1),(i+1,j+1)}^y \tau_{(i+1,j),(i+1,j+1)}^y $, $\tau_{(i-1,j),(i,j)}^x \tau_{(i,j),(i,j+1)}^y $, $\tau_{(i,j),(i+1,j)}^y \tau_{(i,j+1),(i+1,j+1)}^z \tau_{(i+1,j),(i+1,j+1)}^y $, $\tau_{(i,j),(i+1,j)}^y \tau_{(i,j+1),(i+1,j+1)}^y \tau_{(i+1,j),(i+1,j+1)}^y $, $\tau_{(i,j),(i+1,j)}^z \tau_{(i,j+1),(i+1,j+1)}^z \tau_{(i+1,j),(i+1,j+1)}^y $, $\tau_{(i-1,j),(i,j)}^x \tau_{(i,j),(i,j+1)}^z \tau_{(i,j),(i+1,j)}^x $, $\tau_{(i,j-1),(i,j)}^x \tau_{(i,j),(i,j+1)}^y \tau_{(i,j),(i+1,j)}^x $, $\tau_{(i,j),(i,j+1)}^z \tau_{(i,j+1),(i+1,j+1)}^x $, $\tau_{(i-1,j),(i-1,j+1)}^z \tau_{(i-1,j),(i,j)}^y \tau_{(i-1,j+1),(i,j+1)}^z $, $\tau_{(i,j),(i+1,j)}^z \tau_{(i,j+1),(i+1,j+1)}^y \tau_{(i+1,j),(i+1,j+1)}^z $, $\tau_{(i,j),(i+1,j)}^y \tau_{(i,j+1),(i+1,j+1)}^z \tau_{(i+1,j),(i+1,j+1)}^z $, $\tau_{(i-1,j+1),(i,j+1)}^x \tau_{(i,j),(i,j+1)}^y $, $\tau_{(i-1,j+1),(i,j+1)}^x \tau_{(i,j),(i,j+1)}^z $, $\tau_{(i-1,j+1),(i,j+1)}^x \tau_{(i,j),(i,j+1)}^y \tau_{(i,j+1),(i+1,j+1)}^x $, $\tau_{(i,j),(i,j+1)}^z \tau_{(i,j),(i+1,j)}^x \tau_{(i+1,j),(i+1,j+1)}^x $, $\tau_{(i,j),(i,j+1)}^y \tau_{(i,j),(i+1,j)}^x \tau_{(i+1,j),(i+1,j+1)}^x $, $\tau_{(i-1,j),(i-1,j+1)}^z \tau_{(i-1,j),(i,j)}^z \tau_{(i-1,j+1),(i,j+1)}^z $, $\tau_{(i-1,j),(i-1,j+1)}^x \tau_{(i-1,j+1),(i,j+1)}^x \tau_{(i,j),(i,j+1)}^y $, $\tau_{(i-1,j),(i-1,j+1)}^x \tau_{(i-1,j),(i,j)}^x \tau_{(i,j),(i,j+1)}^z $, $\tau_{(i,j),(i,j+1)}^y \tau_{(i,j+1),(i+1,j+1)}^x $, $\tau_{(i-1,j),(i,j)}^x \tau_{(i-1,j+1),(i,j+1)}^x \tau_{(i,j),(i,j+1)}^y $, $\tau_{(i-1,j),(i,j)}^x \tau_{(i-1,j+1),(i,j+1)}^x \tau_{(i,j),(i,j+1)}^z $, $\tau_{(i-1,j),(i-1,j+1)}^y \tau_{(i-1,j),(i,j)}^z \tau_{(i-1,j+1),(i,j+1)}^y $, $\tau_{(i-1,j),(i,j)}^x \tau_{(i,j-1),(i,j)}^x \tau_{(i,j),(i,j+1)}^y $, $\tau_{(i-1,j),(i,j)}^x \tau_{(i,j-1),(i,j)}^x \tau_{(i,j),(i,j+1)}^z $, $\tau_{(i-1,j),(i,j)}^x \tau_{(i,j),(i,j+1)}^z $, $\tau_{(i,j),(i,j+1)}^z \tau_{(i,j+1),(i,j+2)}^x \tau_{(i,j+1),(i+1,j+1)}^x $, $\tau_{(i,j),(i,j+1)}^y \tau_{(i+1,j),(i+1,j+1)}^x $, $\tau_{(i-1,j+1),(i,j+1)}^x \tau_{(i,j),(i,j+1)}^z \tau_{(i,j+1),(i+1,j+1)}^x $, $\tau_{(i,j),(i,j+1)}^y \tau_{(i,j+1),(i,j+2)}^x $, $\tau_{(i,j),(i,j+1)}^z \tau_{(i,j+1),(i+1,j+1)}^x \tau_{(i+1,j),(i+1,j+1)}^x $, $\tau_{(i-1,j),(i-1,j+1)}^x \tau_{(i-1,j),(i,j)}^x \tau_{(i,j),(i,j+1)}^y $, $\tau_{(i-1,j),(i-1,j+1)}^x \tau_{(i-1,j+1),(i,j+1)}^x \tau_{(i,j),(i,j+1)}^z $, $\tau_{(i,j),(i,j+1)}^y \tau_{(i,j),(i+1,j)}^x$, $\tau_{(i-1,j+1),(i,j+1)}^x \tau_{(i,j),(i,j+1)}^y \tau_{(i,j+1),(i,j+2)}^x $, $\tau_{(i-1,j),(i-1,j+1)}^y \tau_{(i-1,j),(i,j)}^z \tau_{(i-1,j+1),(i,j+1)}^z $, $\tau_{(i-1,j+1),(i,j+1)}^x \tau_{(i,j),(i,j+1)}^z \tau_{(i,j+1),(i,j+2)}^x $, $\tau_{(i-1,j),(i-1,j+1)}^z \tau_{(i-1,j),(i,j)}^z \tau_{(i-1,j+1),(i,j+1)}^y $, $\tau_{(i,j),(i,j+1)}^z \tau_{(i+1,j),(i+1,j+1)}^x $, $\tau_{(i,j),(i,j+1)}^y \tau_{(i,j+1),(i,j+2)}^x \tau_{(i,j+1),(i+1,j+1)}^x $, $\tau_{(i-1,j),(i-1,j+1)}^z \tau_{(i-1,j),(i,j)}^y \tau_{(i-1,j+1),(i,j+1)}^y$, $\tau_{(i-1,j),(i-1,j+1)}^x \tau_{(i,j),(i,j+1)}^z $, $\tau_{(i,j),(i,j+1)}^z \tau_{(i,j),(i+1,j)}^x $, $\tau_{(i,j),(i+1,j)}^y \tau_{(i,j+1),(i+1,j+1)}^y \tau_{(i+1,j),(i+1,j+1)}^z $, $\tau_{(i,j),(i+1,j)}^z \tau_{(i,j+1),(i+1,j+1)}^z \tau_{(i+1,j),(i+1,j+1)}^z $, $\tau_{(i,j),(i,j+1)}^z $, $\tau_{(i,j),(i,j+1)}^y $, $\tau_{(i,j-1),(i,j)}^x \tau_{(i,j),(i,j+1)}^y $, $\tau_{(i,j-1),(i,j)}^x \tau_{(i,j),(i,j+1)}^z $, $\tau_{(i,j),(i,j+1)}^z \tau_{(i,j+1),(i,j+2)}^x $, $\tau_{(i,j),(i,j+1)}^y \tau_{(i,j+1),(i+1,j+1)}^x \tau_{(i+1,j),(i+1,j+1)}^x $, $\tau_{(i-1,j),(i-1,j+1)}^x \tau_{(i,j),(i,j+1)}^y $, $\tau_{(i-1,j),(i,j)}^x \tau_{(i,j),(i,j+1)}^y \tau_{(i,j),(i+1,j)}^x $, $\tau_{(i-1,j),(i-1,j+1)}^y \tau_{(i-1,j),(i,j)}^y \tau_{(i-1,j+1),(i,j+1)}^y $, $\tau_{(i,j-1),(i,j)}^x \tau_{(i,j),(i,j+1)}^z \tau_{(i,j),(i+1,j)}^x$
\end{flushleft}

\textbf{Class 12:} all the remaining operators that are in the sequence (level 3 and interaction range 1) but do not appear in the previous blocks. 
This last class contributes to a single block.

\begin{figure}
    \centering
    \includegraphics[width=0.7\linewidth]{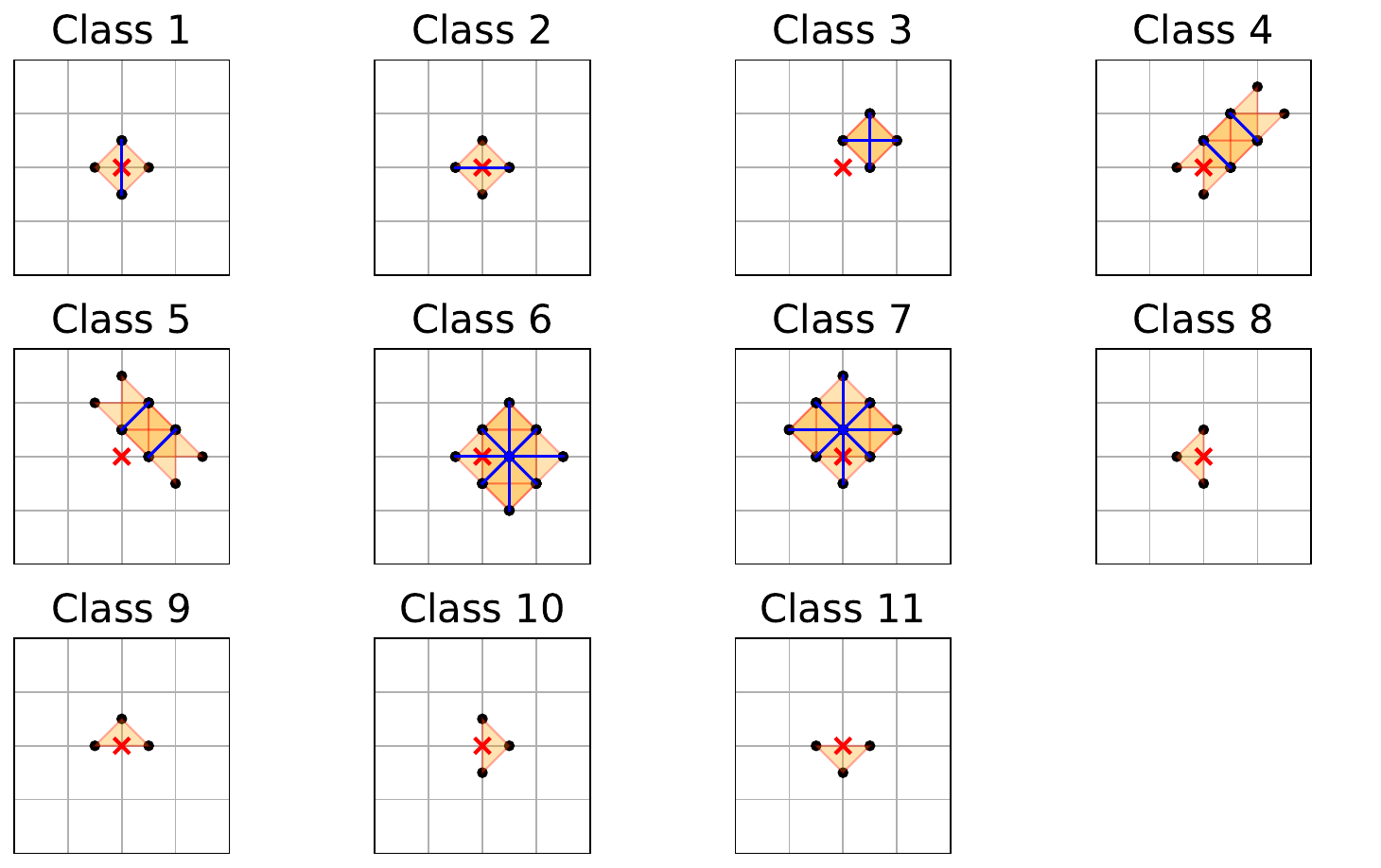}
    \caption{Geometrical representation of classes in the block-ordered sequence.
    The red cross indicates the reference vertex $(i,j)$.
    Each class generates a set of subsequences contributing to a distinct block.
    Nodes correspond to Pauli operators; blue edges represent Pauli strings of Hamming weight~2, and yellow triangles represent Pauli strings of Hamming weight~3. 
    The remaining subsequence, not shown in the figure, consists of all remaining Pauli strings not included in the other classes.}
    \label{fig:geom_rep}
\end{figure}

\section{Projection of the feasible set onto the complex plane} \label{app:feas_set}

\begin{figure}
    \centering
    \includegraphics[width=\linewidth]{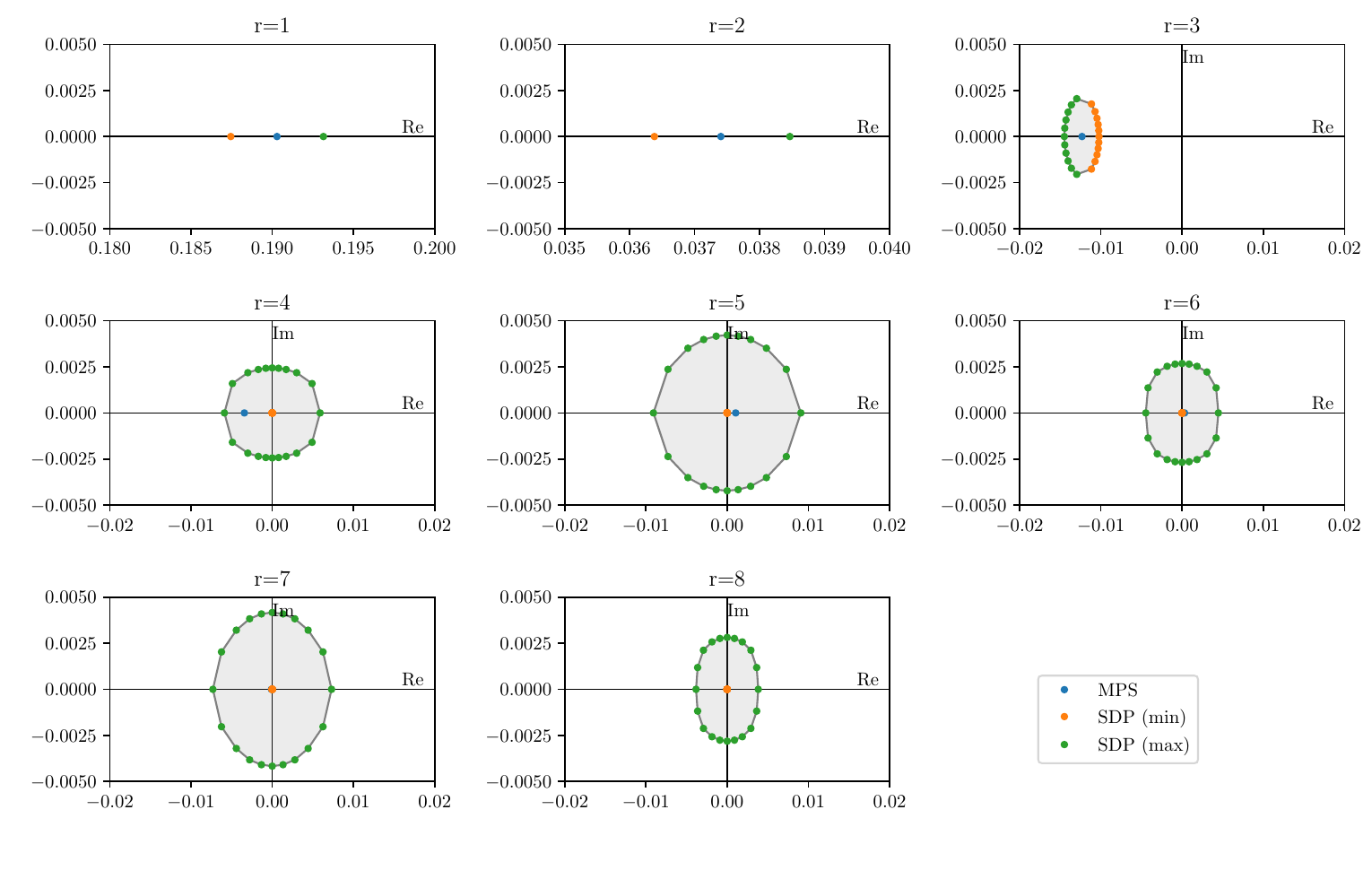}
    \caption{Approximation of the projection of the feasible set onto the complex plane for mesonic strings of lengths $r=1,\ldots,8$. 
    Each subplot corresponds to a fixed string length $r$. 
    The projected feasible set is approximated by discretizing the angular variable $\theta$ and solving the corresponding radial optimization problems at the selected values of $\theta$. 
    The approximation therefore results from the finite discretization of $\theta\in[0,2\pi)$.}
    \label{fig:feas_set}
\end{figure}

In \Cref{fig:feas_set}, we plot the projected feasible set of the mesonic string expectation value $\expval{S(r)}$ in the complex plane, for a fixed length $r$. Concretely, this is the region
\begin{equation}
\mathcal{K}_{S(r)}
=
\left\{
\expval{S(r)} : \boldsymbol{y} \in Y
\right\}
\subset \mathbb{C},
\end{equation}
where $Y$ is the feasible set of the SDP relaxation (the moment matrix PSD condition together with the energy and sector constraints), and $\expval{S(r)}$ is regarded as a linear function of $\boldsymbol{y}$. 
In other words, $\mathcal{K}_{S(r)}$ is the set of all complex values that $\expval{S(r)}$ can take for any state/process consistent with the observed data. 
The certified bounds $|\expval{S(r)}|_{\min}$ and $|\expval{S(r)}|_{\max}$ reported in the main text are simply the closest and farthest distances from the origin to this region.

We obtain this region numerically by choosing a finite set of angles $\theta_1,\dots,\theta_m \in [0,2\pi)$ and, at each $\theta_k$, solving a pair of SDPs for the minimal and maximal radial coordinate of $\expval{S(r)}$ along the direction $\theta_k$. 

For the cases $r=1$ and $r=2$, we observe that the projected feasible sets appear to be one-dimensional. 
This observation has been verified numerically by sampling angles exponentially close to $\theta=0$ and observing that the corresponding SDPs are either infeasible, unbounded or numerically unstable whenever $\theta\neq 0$.

Plotting the resulting points and taking their convex hull traces out a convex polygon, which we take as an approximation of $\mathcal{K}_{S(r)}$. 
Below, we justify why this radial-sampling procedure is mathematically guaranteed to recover $\mathcal{K}_{S(r)}$ (in the limit of dense sampling), for any compact convex feasible set and any linear complex-valued function.

\subsection*{Estimating the feasible set through the radial-sampling procedure}

Let $Y\subset\mathbb{R}^n$ denote the feasible set of the problem. In particular, we assume that $Y$ is a compact spectrahedron, and hence convex.
Let $f(\boldsymbol{y})=f_{\mathrm{R}}(\boldsymbol{y})+\mathrm{i}f_{\mathrm{I}}(\boldsymbol{y})$ be a complex-valued linear function of $\boldsymbol{y}\in Y$, where $f_{\mathrm{R}}$ and $f_{\mathrm{I}}$ are real-valued linear functions. 
In the application above, $f(\boldsymbol{y}) = \expval{S(r)}$, where $S(r)$ is the mesonic string defined in \cref{eq:mesonic_string} and is a non-hermitian operator.

We are interested in the image of the feasible set in the complex plane induced by $f$.
We refer to this set as the projected feasible set and denote it by
\begin{equation}
\mathcal{K}_f := f(Y) \subset \mathbb{C} .
\end{equation}
We will occasionally identify this set with a subset of $\mathbb{R}^2$ via the canonical correspondence between $\mathbb{C}$ and $\mathbb{R}^2$:
\begin{equation}
    \mathcal{K}_f \simeq
    \left\{
        \begin{pmatrix}
            f_{\mathrm{R}}(\boldsymbol{y}) \\
            f_{\mathrm{I}}(\boldsymbol{y})
        \end{pmatrix}
        : \boldsymbol{y} \in Y
    \right\}
    \subset \mathbb{R}^2.
\end{equation}
It will be clear from the context whether $\mathcal{K}_f$ is regarded as a subset of $\mathbb{C}$ or of $\mathbb{R}^2$.
Since $Y$ is compact and convex and $f$ is linear, $\mathcal{K}_f$ is a compact convex subset of the complex plane.

For a given angle $\theta\in[0,2\pi)$, we define the unit vector
\begin{equation}
\boldsymbol{u}_\theta
:=
\begin{pmatrix}
\cos\theta\\
\sin\theta
\end{pmatrix}.
\end{equation}
The corresponding ray from the origin is
\begin{equation}
R_\theta
:=
\left\{
r \boldsymbol{u}_{\theta}:r\geq0
\right\}.
\end{equation}
A point $f(\boldsymbol{y}) \in \mathcal{K}_f$ belongs to this ray if and only if there exists $r\geq0$ such that
\begin{equation}
f_{\mathrm{R}}(\boldsymbol{y})=r\cos\theta,
\qquad
f_{\mathrm{I}}(\boldsymbol{y})=r\sin\theta.
\end{equation}
Equivalently,
\begin{equation}
f_{\mathrm{R}}(\boldsymbol{y})\sin\theta
=
f_{\mathrm{I}}(\boldsymbol{y})\cos\theta,
\end{equation}
with the additional requirement that the corresponding radial coordinate is nonnegative.

Thus, for each $\theta$, we define the minimum and maximum radial coordinates by
\begin{equation}
\begin{aligned}
    r_{\min}(\theta) := \quad & \min_{\boldsymbol{y} \in Y} r \\
    \text{s.t.} \quad & r \geq 0 \\
    & f_{\mathrm{R}}(\boldsymbol{y})=r\cos\theta\\
    & f_{\mathrm{I}}(\boldsymbol{y})=r\sin\theta
\end{aligned}
\end{equation}
and
\begin{equation}
\begin{aligned}
    r_{\max}(\theta) := \quad & \max_{\boldsymbol{y} \in Y} r \\
    \text{s.t.} \quad & r \geq 0 \\
    & f_{\mathrm{R}}(\boldsymbol{y})=r\cos\theta\\
    & f_{\mathrm{I}}(\boldsymbol{y})=r\sin\theta,
\end{aligned}
\end{equation}
whenever the corresponding feasible set is nonempty.
The points obtained from these radial optimizations are
\begin{equation}
\boldsymbol{p}_{\min}(\theta)
=
r_{\min}(\theta)\boldsymbol{u}_{\theta},
\qquad
\boldsymbol{p}_{\max}(\theta)
=
r_{\max}(\theta)\boldsymbol{u}_{\theta}.
\end{equation}

\begin{proposition}[Convex hull recovery of the projected feasible set]
Let $Y\subset\mathbb{R}^n$ be a compact convex set and let $f:Y\to\mathbb{C}$ be linear. Define $\mathcal{K}_f=f(Y)$ and let $r_{\min}(\theta)$ and $r_{\max}(\theta)$ be defined as above. Let
\begin{equation}
\mathcal{E}
:=
\left\{
\boldsymbol{p}_{\min}(\theta),\boldsymbol{p}_{\max}(\theta):
\theta\in[0,2\pi),\
\mathcal{K}_f\cap R_{\theta}\neq\varnothing
\right\}.
\end{equation}
Then
\begin{equation}
\mathcal{K}_f
=
\operatorname{conv}(\mathcal{E}).
\end{equation}
\end{proposition}

\begin{proof}
Since $Y$ is convex and $f$ is linear, $\mathcal{K}_f=f(Y)$ is convex. 
Since $Y$ is compact, $\mathcal{K}_f$ is also compact, so the radial extrema are attained.

Fix $\theta$ such that $\mathcal{K}_f\cap R_{\theta}\neq\varnothing$. Since $\mathcal{K}_f$ is convex, its intersection with the ray $R_{\theta}$ is a compact interval. 
Hence
\begin{equation}
\mathcal{K}_f\cap R_{\theta}
=
\left\{
r \boldsymbol{u}_{\theta}:
r_{\min}(\theta)\leq r\leq r_{\max}(\theta)
\right\}.
\end{equation}
In particular, $\boldsymbol{p}_{\min}(\theta)$ and $\boldsymbol{p}_{\max}(\theta)$ belong to $\mathcal{K}_f$.

Now let $x\in\mathcal{K}_f$ with $x\neq0$. There exists $\theta\in[0,2\pi)$ and $r\geq0$ such that
\begin{equation}
x=r \boldsymbol{u}_{\theta}.
\end{equation}
Since $x\in\mathcal{K}_f\cap R_{\theta}$,
\begin{equation}
r_{\min}(\theta)
\leq r
\leq r_{\max}(\theta).
\end{equation}
Consequently, $x$ lies on the line segment joining $\boldsymbol{p}_{\min}(\theta)$ and $\boldsymbol{p}_{\max}(\theta)$, and therefore
\begin{equation}
x\in\operatorname{conv}(\mathcal{E}).
\end{equation}
Thus
\begin{equation}
\mathcal{K}_f\subseteq\operatorname{conv}(\mathcal{E}).
\end{equation}
Conversely, $\mathcal{E}\subseteq\mathcal{K}_f$ and $\mathcal{K}_f$ is convex. Hence
\begin{equation}
\operatorname{conv}(\mathcal{E})
\subseteq
\mathcal{K}_f.
\end{equation}
The two inclusions establish
\begin{equation}
\mathcal{K}_f
=
\operatorname{conv}(\mathcal{E}).
\end{equation}
\end{proof}
The proposition provides a direct procedure for computing the projected feasible set without explicitly constructing the image $f(\boldsymbol{y})$. 
For each $\theta\in[0,2\pi)$, one solves two optimization problems over the original spectrahedron $Y$, corresponding to the minimum and maximum radial coordinates along the ray $R_{\theta}$. 
The resulting points are mapped to the complex plane and their convex hull is taken:
\begin{equation}
\mathcal{K}_f = \operatorname{conv}
\left\{
r_{\min}(\theta)e^{\mathrm{i}\theta},
r_{\max}(\theta)e^{\mathrm{i}\theta}
:
\theta\in[0,2\pi)
\right\}.
\end{equation}
In terms of the original complex quantity, the radial constraint can equivalently be written as
\begin{equation}
f(\boldsymbol{y})=r e^{\mathrm{i}\theta},
\qquad r\geq0.
\end{equation}
When $\cos\theta\neq0$ or $\sin\theta\neq0$, the radial coordinate is equivalently given by
\begin{equation}
r = \frac{\operatorname{Re}f(\boldsymbol{y})}{\cos\theta} = \frac{\operatorname{Im}f(\boldsymbol{y})}{\sin\theta},
\end{equation}
whenever the corresponding quotient is defined. This formulation makes clear that the optimization computes the radial intersections of the projected feasible set with the rays emanating from the origin.

In practice, since only a finite set of angles $\theta_1,\dots,\theta_m$ is evaluated, the resulting polygon $\operatorname{conv}\{r_{\min}(\theta_k)e^{\mathrm{i}\theta_k}, r_{\max}(\theta_k)e^{\mathrm{i}\theta_k}\}_{k=1}^m$ is, by the proposition, always an \emph{inner} approximation of $\mathcal{K}_f$, with equality recovered only in the limit of dense sampling. 
Assuming, as discussed in the main text, that $|\expval{S(r)}|_\theta$ varies smoothly with $\theta$, this inscribed polygon converges quickly and provides a reliable stand-in for the true projected feasible set $\mathcal{K}_{S(r)}$ shown in the figures.

\FloatBarrier
\twocolumngrid
\bibliography{references}

@article{affleck_large-n_1988,
  title = {Large-n Limit of the {{Heisenberg-Hubbard}} Model: {{Implications}} for High-${T}_{c}$ Superconductors},
  shorttitle = {Large-n Limit of the {{Heisenberg-Hubbard}} Model},
  author = {Affleck, Ian and Marston, J. Brad},
  year = 1988,
  month = mar,
  journal = {Physical Review B},
  volume = {37},
  number = {7},
  pages = {3774--3777},
  publisher = {American Physical Society},
  doi = {10.1103/PhysRevB.37.3774},
  urldate = {2026-06-30}
}

@article{aidelsburger_cold_2021,
  title = {Cold Atoms Meet Lattice Gauge Theory},
  author = {Aidelsburger, Monika and Barbiero, Luca and Bermudez, Alejandro and Chanda, Titas and Dauphin, Alexandre and {Gonz{\'a}lez-Cuadra}, Daniel and Grzybowski, Przemys{\l}aw R. and Hands, Simon and Jendrzejewski, Fred and J{\"u}nemann, Johannes and Juzeli{\=u}nas, Gediminas and Kasper, Valentin and Piga, Angelo and Ran, Shi-Ju and Rizzi, Matteo and Sierra, Germ{\'a}n and Tagliacozzo, Luca and Tirrito, Emanuele and Zache, Torsten V. and Zakrzewski, Jakub and Zohar, Erez and Lewenstein, Maciej},
  year = 2021,
  month = dec,
  journal = {Philosophical Transactions of the Royal Society A: Mathematical, Physical and Engineering Sciences},
  volume = {380},
  number = {2216},
  pages = {20210064},
  issn = {1364-503X},
  doi = {10.1098/rsta.2021.0064},
  urldate = {2026-07-01}
}

@article{aoki_flag_2026,
  title = {{{FLAG}} Review 2024},
  author = {Aoki, Y. and Blum, T. and Collins, S. and Del Debbio, L. and Della Morte, M. and Dimopoulos, P. and Feng, X. and Golterman, M. and Gottlieb, Steven and Gupta, R. and Herdo{\'\i}za, G. and Hernandez, P. and J{\"u}ttner, A. and Kaneko, T. and Lunghi, E. and Meinel, S. and Monahan, C. and Nicholson, A. and Onogi, T. and Petreczky, P. and Portelli, A. and Ramos, A. and Sharpe, S. R. and Simone, J. N. and Sint, S. and Sommer, R. and Tantalo, N. and {Van de Water}, R. and Vaquero, A. and Wenger, U. and Wittig, H.},
  year = 2026,
  month = jan,
  journal = {Physical Review D: Particles and Fields},
  volume = {113},
  number = {1},
  pages = {014508},
  publisher = {American Physical Society},
  doi = {10.1103/nfzp-p5dn},
  collaboration = {Flavour Lattice Averaging Group (FLAG)}
}

@article{balents_spin_2010,
  title = {Spin Liquids in Frustrated Magnets},
  author = {Balents, Leon},
  year = 2010,
  journal = {Nature},
  volume = {464},
  number = {7286},
  pages = {199--208},
  doi = {10.1038/nature08917}
}

@article{banuls_efficient_2017,
  title = {Efficient {{Basis Formulation}} for (1+1)-{{Dimensional SU}}(2) {{Lattice Gauge Theory}}: {{Spectral Calculations}} with {{Matrix Product States}}},
  shorttitle = {Efficient {{Basis Formulation}} for (\$1+1\$)-{{Dimensional SU}}(2) {{Lattice Gauge Theory}}},
  author = {Ba{\~n}uls, Mari Carmen and Cichy, Krzysztof and Cirac, J. Ignacio and Jansen, Karl and K{\"u}hn, Stefan},
  year = 2017,
  month = nov,
  journal = {Physical Review X},
  volume = {7},
  number = {4},
  pages = {041046},
  publisher = {American Physical Society},
  doi = {10.1103/PhysRevX.7.041046},
  urldate = {2026-07-01}
}

@article{banuls_simulating_2020,
  title = {Simulating Lattice Gauge Theories within Quantum Technologies},
  author = {Ba{\~n}uls, Mari Carmen and Blatt, Rainer and Catani, Jacopo and Celi, Alessio and Cirac, Juan Ignacio and Dalmonte, Marcello and Fallani, Leonardo and Jansen, Karl and Lewenstein, Maciej and Montangero, Simone and Muschik, Christine A. and Reznik, Benni and Rico, Enrique and Tagliacozzo, Luca and Van Acoleyen, Karel and Verstraete, Frank and Wiese, Uwe-Jens and Wingate, Matthew and Zakrzewski, Jakub and Zoller, Peter},
  year = 2020,
  month = aug,
  journal = {The European Physical Journal D},
  volume = {74},
  number = {8},
  pages = {165},
  issn = {1434-6079},
  doi = {10.1140/epjd/e2020-100571-8},
  urldate = {2026-07-01},
  langid = {english}
}

@article{banuls_tensor_2018,
  ids = {banuls_tensor_2018-3},
  title = {Tensor {{Network}} Study of the (1+1)-Dimensional {{Thirring Model}}},
  author = {Ba{\~n}uls, Mari Carmen and Cichy, Krzysztof and Kao, Ying-Jer and David Lin, C.-J. and Lin, Yu-Ping and {Tao-Lin Tan}, David},
  editor = {Della Morte, M. and Fritzsch, P. and G{\'a}miz S{\'a}nchez, E. and Pena Ruano, C.},
  year = 2018,
  journal = {EPJ Web of Conferences},
  volume = {175},
  pages = {11017},
  issn = {2100-014X},
  doi = {10.1051/epjconf/201817511017},
  urldate = {2019-06-26}
}

@article{barata_quantum_2023,
  title = {Quantum Simulation of In-Medium {{QCD}} Jets: {{Momentum}} Broadening, Gluon Production, and Entropy Growth},
  shorttitle = {Quantum Simulation of In-Medium {{QCD}} Jets},
  author = {Barata, Jo{\~a}o and Du, Xiaojian and Li, Meijian and Qian, Wenyang and Salgado, Carlos A.},
  year = 2023,
  month = sep,
  journal = {Physical Review D},
  volume = {108},
  number = {5},
  pages = {056023},
  publisher = {American Physical Society},
  doi = {10.1103/PhysRevD.108.056023},
  urldate = {2026-07-01}
}

@article{barthel_solving_2012,
  title = {Solving {{Condensed-Matter Ground-State Problems}} by {{Semidefinite Relaxations}}},
  author = {Barthel, Thomas and H{\"u}bener, Robert},
  year = 2012,
  month = may,
  journal = {Physical Review Letters},
  volume = {108},
  number = {20},
  eprint = {1106.4966},
  pages = {200404},
  doi = {10.1103/PhysRevLett.108.200404},
  urldate = {2026-09-03}
}

@article{bauer_quantum_2023,
  title = {Quantum {{Simulation}} for {{High-Energy Physics}}},
  author = {Bauer, Christian W. and Davoudi, Zohreh and Balantekin, A. Baha and Bhattacharya, Tanmoy and Carena, Marcela and {de Jong}, Wibe A. and Draper, Patrick and {El-Khadra}, Aida and Gemelke, Nate and Hanada, Masanori and Kharzeev, Dmitri and Lamm, Henry and Li, Ying-Ying and Liu, Junyu and Lukin, Mikhail and Meurice, Yannick and Monroe, Christopher and Nachman, Benjamin and Pagano, Guido and Preskill, John and Rinaldi, Enrico and Roggero, Alessandro and Santiago, David I. and Savage, Martin J. and Siddiqi, Irfan and Siopsis, George and Van Zanten, David and Wiebe, Nathan and Yamauchi, Yukari and {Yeter-Aydeniz}, K{\"u}bra and Zorzetti, Silvia},
  year = 2023,
  month = may,
  journal = {PRX Quantum},
  volume = {4},
  number = {2},
  pages = {027001},
  publisher = {American Physical Society},
  doi = {10.1103/PRXQuantum.4.027001},
  urldate = {2026-06-30}
}

@article{baumgratz_lower_2012,
  title = {Lower {{Bounds}} for {{Ground States}} of {{Condensed Matter Systems}}},
  author = {Baumgratz, Tillmann and Plenio, Martin B.},
  year = 2012,
  month = feb,
  journal = {New Journal of Physics},
  volume = {14},
  number = {2},
  pages = {023027},
  issn = {1367-2630},
  doi = {10.1088/1367-2630/14/2/023027},
  urldate = {2023-10-20}
}

@article{bender_variational_2023,
  title = {Variational {{Monte Carlo}} Algorithm for Lattice Gauge Theories with Continuous Gauge Groups: {{A}} Study of (2+1)-Dimensional Compact {{QED}} with Dynamical Fermions at Finite Density},
  shorttitle = {Variational {{Monte Carlo}} Algorithm for Lattice Gauge Theories with Continuous Gauge Groups},
  author = {Bender, Julian and Emonts, Patrick and Cirac, J. Ignacio},
  year = 2023,
  month = nov,
  journal = {Physical Review Research},
  volume = {5},
  number = {4},
  pages = {043128},
  publisher = {American Physical Society},
  doi = {10.1103/PhysRevResearch.5.043128}
}

@article{berenstein_anomalous_2022,
  title = {Anomalous Bootstrap on the Half-Line},
  author = {Berenstein, David and Hulsey, George},
  year = 2022,
  month = aug,
  journal = {Physical Review D},
  volume = {106},
  number = {4},
  pages = {045029},
  publisher = {American Physical Society},
  doi = {10.1103/PhysRevD.106.045029},
  urldate = {2024-07-12}
}

@misc{berenstein_bootstrapping_2021,
  title = {Bootstrapping {{Simple QM Systems}}},
  author = {Berenstein, David and Hulsey, George},
  year = 2021,
  publisher = {arXiv},
  doi = {10.48550/ARXIV.2108.08757},
  urldate = {2026-08-11},
  copyright = {arXiv.org perpetual, non-exclusive license}
}

@article{berenstein_bootstrapping_2022,
  title = {Bootstrapping More {{QM}} Systems},
  author = {Berenstein, David and Hulsey, George},
  year = 2022,
  month = jun,
  journal = {Journal of Physics A: Mathematical and Theoretical},
  volume = {55},
  number = {27},
  pages = {275304},
  publisher = {IOP Publishing},
  issn = {1751-8121},
  doi = {10.1088/1751-8121/ac7118},
  urldate = {2024-05-03},
  langid = {english}
}

@article{berenstein_semidefinite_2023,
  title = {A {{Semidefinite Programming}} Algorithm for the {{Quantum Mechanical Bootstrap}}},
  author = {Berenstein, David and Hulsey, George},
  year = 2023,
  month = may,
  journal = {Physical Review E},
  volume = {107},
  number = {5},
  pages = {L053301},
  issn = {2470-0045, 2470-0053},
  doi = {10.1103/PhysRevE.107.L053301},
  urldate = {2024-07-12}
}

@book{blekherman_semidefinite_2012,
  title = {Semidefinite {{Optimization}} and {{Convex Algebraic Geometry}}},
  editor = {Blekherman, Grigoriy and Parrilo, Pablo A. and Thomas, Rekha R.},
  year = 2012,
  month = dec,
  publisher = {{Society for Industrial and Applied Mathematics}},
  address = {Philadelphia, PA},
  doi = {10.1137/1.9781611972290},
  urldate = {2022-11-23},
  isbn = {978-1-61197-228-3 978-1-61197-229-0},
  langid = {english}
}

@article{borla_confined_2020,
  title = {Confined {{Phases}} of {{One-Dimensional Spinless Fermions Coupled}} to \$\textbraceleft{{Z}}\textbraceright\_\textbraceleft 2\textbraceright\$ {{Gauge Theory}}},
  author = {Borla, Umberto and Verresen, Ruben and Grusdt, Fabian and Moroz, Sergej},
  year = 2020,
  month = mar,
  journal = {Physical Review Letters},
  volume = {124},
  number = {12},
  pages = {120503},
  publisher = {American Physical Society},
  doi = {10.1103/PhysRevLett.124.120503},
  urldate = {2026-09-03}
}

@book{boyd_convex_2004,
  title = {Convex Optimization},
  author = {Boyd, Stephen P. and Vandenberghe, Lieven},
  year = 2004,
  publisher = {Cambridge University Press},
  address = {Cambridge, UK ; New York},
  isbn = {978-0-521-83378-3},
  langid = {english},
  lccn = {QA402.5 .B69 2004}
}

@book{bratteli_operator_1987,
  title = {Operator {{Algebras}} and {{Quantum Statistical Mechanics}} 1},
  author = {Bratteli, Ola and Robinson, Derek W.},
  year = 1987,
  publisher = {Springer},
  address = {Berlin, Heidelberg},
  doi = {10.1007/978-3-662-02520-8},
  urldate = {2026-09-03},
  copyright = {http://www.springer.com/tdm},
  isbn = {978-3-642-05736-6 978-3-662-02520-8}
}

@book{burgdorf_optimization_2016,
  title = {Optimization of {{Polynomials}} in {{Non-Commuting Variables}}},
  author = {Burgdorf, Sabine and Klep, Igor and Povh, Janez},
  year = 2016,
  series = {{{SpringerBriefs}} in {{Mathematics}}},
  publisher = {Springer International Publishing},
  address = {Cham},
  doi = {10.1007/978-3-319-33338-0},
  urldate = {2026-06-30},
  copyright = {http://www.springer.com/tdm},
  isbn = {978-3-319-33336-6 978-3-319-33338-0},
  langid = {english}
}

@article{buyens_matrix_2014,
  title = {Matrix {{Product States}} for {{Gauge Field Theories}}},
  author = {Buyens, Boye and Haegeman, Jutho and Van Acoleyen, Karel and Verschelde, Henri and Verstraete, Frank},
  year = 2014,
  month = aug,
  journal = {Physical Review Letters},
  volume = {113},
  number = {9},
  issn = {0031-9007, 1079-7114},
  doi = {10.1103/PhysRevLett.113.091601},
  urldate = {2018-03-16}
}

@article{buyens_tensor_2015,
  ids = {buyens_tensor_2015-1},
  title = {Tensor Networks for Gauge Field Theories},
  author = {Buyens, Boye and Haegeman, Jutho and Verstraete, Frank and Van Acoleyen, Karel},
  year = 2015,
  journal = {arXiv:1511.04288},
  eprint = {1511.04288},
  primaryclass = {hep-lat},
  adsurl = {https://ui.adsabs.harvard.edu/\#abs/2015arXiv151104288B},
  archiveprefix = {arXiv},
  eid = {arXiv:1511.04288}
}

@article{byrnes_simulating_2006,
  title = {Simulating Lattice Gauge Theories on a Quantum Computer},
  author = {Byrnes, Tim and Yamamoto, Yoshihisa},
  year = 2006,
  month = feb,
  journal = {Physical Review A},
  volume = {73},
  number = {2},
  pages = {022328},
  publisher = {American Physical Society},
  doi = {10.1103/PhysRevA.73.022328},
  urldate = {2026-07-01}
}

@misc{cataldi_hamiltonian_2025,
  title = {Hamiltonian {{Lattice Gauge Theories}}: Emergent Properties from {{Tensor Network}} Methods},
  shorttitle = {Hamiltonian {{Lattice Gauge Theories}}},
  author = {Cataldi, Giovanni},
  year = 2025,
  month = may,
  publisher = {arXiv},
  doi = {10.48550/arXiv.2501.11115},
  urldate = {2026-07-01}
}

@article{chai_scalable_2026,
  title = {Scalable Quantum Algorithm for Meson Scattering in a Lattice Gauge Theory},
  author = {Chai, Yahui and Guo, Yibin and K{\"u}hn, Stefan},
  year = 2026,
  month = sep,
  journal = {Physical Review Research},
  volume = {8},
  number = {3},
  pages = {033298},
  issn = {2643-1564},
  doi = {10.1103/f93h-rlzy},
  urldate = {2026-10-01},
  langid = {english}
}

@article{ciavarella_quantum_2024,
  title = {Quantum {{Simulation}} of {{SU}}(3) {{Lattice Yang-Mills Theory}} at {{Leading Order}} in {{Large}}-${N}_{c}$ {{Expansion}}},
  author = {Ciavarella, Anthony N. and Bauer, Christian W.},
  year = 2024,
  month = sep,
  journal = {Physical Review Letters},
  volume = {133},
  number = {11},
  pages = {111901},
  publisher = {American Physical Society},
  doi = {10.1103/PhysRevLett.133.111901},
  urldate = {2026-08-07}
}

@article{cirac_matrix_2021,
  title = {Matrix Product States and Projected Entangled Pair States: {{Concepts}}, Symmetries, Theorems},
  shorttitle = {Matrix Product States and Projected Entangled Pair States},
  author = {Cirac, J. Ignacio and {P{\'e}rez-Garc{\'i}a}, David and Schuch, Norbert and Verstraete, Frank},
  year = 2021,
  month = dec,
  journal = {Reviews of Modern Physics},
  volume = {93},
  number = {4},
  pages = {045003},
  issn = {0034-6861, 1539-0756},
  doi = {10.1103/RevModPhys.93.045003},
  urldate = {2022-09-13},
  langid = {english}
}

@misc{cobos_real-time_2025,
  title = {Real-Time Dynamics in a (2+1)-{{D}} Gauge Theory: {{The}} Stringy Nature on a Superconducting Quantum Simulator},
  shorttitle = {Real-Time Dynamics in a (2+1)-{{D}} Gauge Theory},
  author = {Cobos, Jes{\'u}s and Fraxanet, Joana and Benito, C{\'e}sar and {di Marcantonio}, Francesco and Rivero, Pedro and Kap{\'a}s, Korn{\'e}l and Werner, Mikl{\'o}s Antal and Legeza, {\"O}rs and Bermudez, Alejandro and Rico, Enrique},
  year = 2025,
  month = jul,
  number = {arXiv:2507.08088},
  eprint = {2507.08088},
  primaryclass = {quant-ph},
  publisher = {arXiv},
  doi = {10.48550/arXiv.2507.08088},
  urldate = {2025-09-08},
  archiveprefix = {arXiv}
}

@article{cochran_visualizing_2025,
  ids = {Cochran2025VisualizingDynamicsCharges-2},
  title = {Visualizing Dynamics of Charges and Strings in (2 + 1){{D}} Lattice Gauge Theories},
  author = {Cochran, T. A. and Jobst, B. and Rosenberg, E. and Lensky, Y. D. and Gyawali, G. and Eassa, N. and Will, M. and Szasz, A. and Abanin, D. and Acharya, R. and Aghababaie Beni, L. and Andersen, T. I. and Ansmann, M. and Arute, F. and Arya, K. and Asfaw, A. and Atalaya, J. and Babbush, R. and Ballard, B. and Bardin, J. C. and Bengtsson, A. and Bilmes, A. and Bourassa, A. and Bovaird, J. and Broughton, M. and Browne, D. A. and Buchea, B. and Buckley, B. B. and Burger, T. and Burkett, B. and Bushnell, N. and Cabrera, A. and Campero, J. and Chang, H.-S. and Chen, Z. and Chiaro, B. and Claes, J. and Cleland, A. Y. and Cogan, J. and Collins, R. and Conner, P. and Courtney, W. and Crook, A. L. and Curtin, B. and Das, S. and Demura, S. and De Lorenzo, L. and Di Paolo, A. and Donohoe, P. and Drozdov, I. and Dunsworth, A. and Eickbusch, A. and Elbag, A. Moshe and Elzouka, M. and Erickson, C. and Ferreira, V. S. and Burgos, L. Flores and Forati, E. and Fowler, A. G. and Foxen, B. and Ganjam, S. and Gasca, R. and Genois, {\'E} and Giang, W. and Gilboa, D. and Gosula, R. and Grajales Dau, A. and Graumann, D. and Greene, A. and Gross, J. A. and Habegger, S. and Hansen, M. and Harrigan, M. P. and Harrington, S. D. and Heu, P. and Higgott, O. and Hilton, J. and Huang, H.-Y. and Huff, A. and Huggins, W. and Jeffrey, E. and Jiang, Z. and Jones, C. and Joshi, C. and Juhas, P. and Kafri, D. and Kang, H. and Karamlou, A. H. and Kechedzhi, K. and Khaire, T. and Khattar, T. and Khezri, M. and Kim, S. and Klimov, P. and Kobrin, B. and Korotkov, A. and Kostritsa, F. and Kreikebaum, J. and Kurilovich, V. and Landhuis, D. and {Lange-Dei}, T. and Langley, B. and Lau, K.-M. and Ledford, J. and Lee, K. and Lester, B. and Le Guevel, L. and Li, W. and Lill, A. T. and Livingston, W. and Locharla, A. and Lundahl, D. and Lunt, A. and Madhuk, S. and Maloney, A. and Mandr{\`a}, S. and Martin, L. and Martin, O. and Maxfield, C. and McClean, J. and McEwen, M. and Meeks, S. and Megrant, A. and Miao, K. and Molavi, R. and Molina, S. and Montazeri, S. and Movassagh, R. and Neill, C. and Newman, M. and Nguyen, A. and Nguyen, M. and Ni, C.-H. and Ottosson, K. and Pizzuto, A. and Potter, R. and Pritchard, O. and Quintana, C. and Ramachandran, G. and Reagor, M. and Rhodes, D. and Roberts, G. and Sankaragomathi, K. and Satzinger, K. and Schurkus, H. and Shearn, M. and Shorter, A. and Shutty, N. and Shvarts, V. and Sivak, V. and Small, S. and Smith, W. C. and Springer, S. and Sterling, G. and Suchard, J. and Sztein, A. and Thor, D. and Torunbalci, M. and Vaishnav, A. and Vargas, J. and Vdovichev, S. and Vidal, G. and Vollgraff Heidweiller, C. and Waltman, S. and Wang, S. X. and Ware, B. and White, T. and Wong, K. and Woo, B. W. K. and Xing, C. and Yao, Z. Jamie and Yeh, P. and Ying, B. and Yoo, J. and Yosri, N. and Young, G. and Zalcman, A. and Zhang, Y. and Zhu, N. and Zobrist, N. and Boixo, S. and Kelly, J. and Lucero, E. and Chen, Y. and Smelyanskiy, V. and Neven, H. and {Gammon-Smith}, A. and Pollmann, F. and Knap, M. and Roushan, P.},
  year = 2025,
  month = jun,
  journal = {Nature},
  volume = {642},
  number = {8067},
  pages = {315--320},
  publisher = {Nature Publishing Group},
  issn = {1476-4687},
  doi = {10.1038/s41586-025-08999-9},
  urldate = {2025-07-25},
  copyright = {2025 The Author(s)},
  langid = {english}
}

@article{creutz_monte_1983,
  title = {Monte {{Carlo}} Computations in Lattice Gauge Theories},
  author = {Creutz, Michael and Jacobs, Laurence and Rebbi, Claudio},
  year = 1983,
  month = apr,
  journal = {Physics Reports},
  volume = {95},
  number = {4},
  pages = {201--282},
  issn = {0370-1573},
  doi = {10.1016/0370-1573(83)90016-9},
  urldate = {2026-05-08}
}

@article{dalmonte_lattice_2016,
  title = {Lattice Gauge Theory Simulations in the Quantum Information Era},
  author = {Dalmonte, M. and Montangero, S.},
  year = 2016,
  month = jul,
  journal = {Contemporary Physics},
  volume = {57},
  number = {3},
  pages = {388--412},
  publisher = {Taylor \& Francis},
  issn = {0010-7514},
  doi = {10.1080/00107514.2016.1151199},
  urldate = {2026-07-01}
}

@article{davoudi_scattering_2024,
  title = {Scattering Wave Packets of Hadrons in Gauge Theories: {{Preparation}} on a Quantum Computer},
  shorttitle = {Scattering Wave Packets of Hadrons in Gauge Theories},
  author = {Davoudi, Zohreh and Hsieh, Chung-Chun and Kadam, Saurabh V.},
  year = 2024,
  month = nov,
  journal = {Quantum},
  volume = {8},
  pages = {1520},
  publisher = {Verein zur F\"orderung des Open Access Publizierens in den Quantenwissenschaften},
  doi = {10.22331/q-2024-11-11-1520},
  urldate = {2026-09-30},
  langid = {british}
}

@article{di_meglio_quantum_2024,
  title = {Quantum {{Computing}} for {{High-Energy Physics}}: {{State}} of the {{Art}} and {{Challenges}}},
  shorttitle = {Quantum {{Computing}} for {{High-Energy Physics}}},
  author = {Di Meglio, Alberto and Jansen, Karl and Tavernelli, Ivano and Alexandrou, Constantia and Arunachalam, Srinivasan and Bauer, Christian W. and Borras, Kerstin and Carrazza, Stefano and Crippa, Arianna and Croft, Vincent and De Putter, Roland and Delgado, Andrea and Dunjko, Vedran and Egger, Daniel J. and {Fern{\'a}ndez-Combarro}, Elias and Fuchs, Elina and Funcke, Lena and {Gonz{\'a}lez-Cuadra}, Daniel and Grossi, Michele and Halimeh, Jad C. and Holmes, Zo{\"e} and K{\"u}hn, Stefan and Lacroix, Denis and Lewis, Randy and Lucchesi, Donatella and Martinez, Miriam Lucio and Meloni, Federico and Mezzacapo, Antonio and Montangero, Simone and Nagano, Lento and Pascuzzi, Vincent R. and Radescu, Voica and Ortega, Enrique Rico and Roggero, Alessandro and Schuhmacher, Julian and Seixas, Joao and Silvi, Pietro and Spentzouris, Panagiotis and Tacchino, Francesco and Temme, Kristan and Terashi, Koji and Tura, Jordi and T{\"u}ys{\"u}z, Cenk and Vallecorsa, Sofia and Wiese, Uwe-Jens and Yoo, Shinjae and Zhang, Jinglei},
  year = 2024,
  month = aug,
  journal = {PRX Quantum},
  volume = {5},
  number = {3},
  pages = {037001},
  issn = {2691-3399},
  doi = {10.1103/PRXQuantum.5.037001},
  urldate = {2026-08-11},
  langid = {english}
}

@article{elitzur_impossibility_1975,
  title = {Impossibility of Spontaneously Breaking Local Symmetries},
  author = {Elitzur, S.},
  year = 1975,
  month = dec,
  journal = {Physical Review D},
  volume = {12},
  number = {12},
  pages = {3978--3982},
  publisher = {American Physical Society},
  doi = {10.1103/PhysRevD.12.3978},
  urldate = {2026-07-20}
}

@article{emonts_variational_2020,
  ids = {emonts_variational_2020-1},
  title = {Variational {{Monte Carlo}} Simulation with Tensor Networks of a Pure $\mathds{Z}_3$ Gauge Theory in $( 2 + 1 )$ {{D}}},
  author = {Emonts, Patrick and Ba{\~n}uls, Mari Carmen and Cirac, Ignacio and Zohar, Erez},
  year = 2020,
  month = oct,
  journal = {Physical Review D},
  volume = {102},
  number = {7},
  pages = {074501},
  publisher = {American Physical Society (APS)},
  issn = {2470-0010, 2470-0029},
  doi = {10.1103/PhysRevD.102.074501},
  urldate = {2020-12-03},
  langid = {english}
}

@article{farrell_scalable_2024,
  title = {Scalable {{Circuits}} for {{Preparing Ground States}} on {{Digital Quantum Computers}}: {{The Schwinger Model Vacuum}} on 100 {{Qubits}}},
  shorttitle = {Scalable {{Circuits}} for {{Preparing Ground States}} on {{Digital Quantum Computers}}},
  author = {Farrell, Roland C. and Illa, Marc and Ciavarella, Anthony N. and Savage, Martin J.},
  year = 2024,
  month = apr,
  journal = {PRX Quantum},
  volume = {5},
  number = {2},
  pages = {020315},
  publisher = {American Physical Society},
  doi = {10.1103/PRXQuantum.5.020315},
  urldate = {2026-07-01}
}

@article{fawzi_certified_2024,
  title = {Certified Algorithms for Equilibrium States of Local Quantum {{Hamiltonians}}},
  author = {Fawzi, Hamza and Fawzi, Omar and Scalet, Samuel O.},
  year = 2024,
  month = aug,
  journal = {Nature Communications},
  volume = {15},
  number = {1},
  eprint = {2311.18706},
  pages = {7394},
  doi = {10.1038/s41467-024-51592-3},
  urldate = {2026-09-03}
}

@article{felser_efficient_2021,
  title = {Efficient {{Tensor Network Ansatz}} for {{High-Dimensional Quantum Many-Body Problems}}},
  author = {Felser, Timo and Notarnicola, Simone and Montangero, Simone},
  year = 2021,
  month = apr,
  journal = {Physical Review Letters},
  volume = {126},
  number = {17},
  pages = {170603},
  publisher = {American Physical Society},
  doi = {10.1103/PhysRevLett.126.170603},
  urldate = {2025-05-24}
}

@book{fradkin_field_2013,
  title = {Field Theories of Condensed Matter Physics},
  author = {Fradkin, Eduardo},
  year = 2013,
  publisher = {Cambridge University Press}
}

@article{fradkin_order_1978,
  title = {Order and Disorder in Gauge Systems and Magnets},
  author = {Fradkin, Eduardo and Susskind, Leonard},
  year = 1978,
  month = may,
  journal = {Physical Review D},
  volume = {17},
  number = {10},
  pages = {2637--2658},
  publisher = {American Physical Society},
  doi = {10.1103/PhysRevD.17.2637},
  urldate = {2026-09-12}
}

@article{gatermann_symmetry_2004,
  title = {Symmetry Groups, Semidefinite Programs, and Sums of Squares},
  author = {Gatermann, Karin and Parrilo, Pablo A.},
  year = 2004,
  month = sep,
  journal = {Journal of Pure and Applied Algebra},
  volume = {192},
  number = {1},
  pages = {95--128},
  issn = {0022-4049},
  doi = {10.1016/j.jpaa.2003.12.011},
  urldate = {2026-05-08}
}

@book{gattringer_quantum_2009,
  title = {Quantum Chromodynamics on the Lattice: An Introductory Presentation},
  author = {Gattringer, Christof and Lang, Christian},
  year = 2009,
  volume = {788},
  publisher = {Springer Science \& Business Media}
}

@article{gonzalez-cuadra_observation_2025,
  title = {Observation of String Breaking on a (2 + 1){{D Rydberg}} Quantum Simulator},
  author = {{Gonz{\'a}lez-Cuadra}, Daniel and Hamdan, Majd and Zache, Torsten V. and Braverman, Boris and Kornja{\v c}a, Milan and Lukin, Alexander and Cant{\'u}, Sergio H. and Liu, Fangli and Wang, Sheng-Tao and Keesling, Alexander and Lukin, Mikhail D. and Zoller, Peter and Bylinskii, Alexei},
  year = 2025,
  month = jun,
  journal = {Nature},
  volume = {642},
  number = {8067},
  pages = {321--326},
  publisher = {Nature Publishing Group},
  issn = {1476-4687},
  doi = {10.1038/s41586-025-09051-6},
  urldate = {2026-07-02},
  copyright = {2025 The Author(s), under exclusive licence to Springer Nature Limited},
  langid = {english}
}

@article{gross_50_2023,
  title = {50 {{Years}} of Quantum Chromodynamics},
  author = {Gross, Franz and Klempt, Eberhard and Brodsky, Stanley J. and Buras, Andrzej J. and Burkert, Volker D. and Heinrich, Gudrun and Jakobs, Karl and Meyer, Curtis A. and Orginos, Kostas and Strickland, Michael and Stachel, Johanna and Zanderighi, Giulia and Brambilla, Nora and {Braun-Munzinger}, Peter and Britzger, Daniel and Capstick, Simon and Cohen, Tom and Crede, Volker and Constantinou, Martha and Davies, Christine and Del Debbio, Luigi and Denig, Achim and DeTar, Carleton and Deur, Alexandre and Dokshitzer, Yuri and Dosch, Hans G{\"u}nter and Dudek, Jozef and Dunford, Monica and Epelbaum, Evgeny and Escobedo, Miguel A. and Fritzsch, Harald and Fukushima, Kenji and Gambino, Paolo and Gillberg, Dag and Gottlieb, Steven and Grafstrom, Per and Grazzini, Massimiliano and Grube, Boris and Guskov, Alexey and Iijima, Toru and Ji, Xiangdong and Karsch, Frithjof and Kluth, Stefan and Kogut, John B. and Krauss, Frank and Kumano, Shunzo and Leinweber, Derek and Leutwyler, Heinrich and Li, Hai-Bo and Li, Yang and Malaescu, Bogdan and Mariotti, Chiara and Maris, Pieter and Marzani, Simone and Melnitchouk, Wally and Messchendorp, Johan and Meyer, Harvey and Mitchell, Ryan Edward and Mondal, Chandan and Nerling, Frank and Neubert, Sebastian and Pappagallo, Marco and Pastore, Saori and Pel{\'a}ez, Jos{\'e} R. and Puckett, Andrew and Qiu, Jianwei and Rabbertz, Klaus and Ramos, Alberto and Rossi, Patrizia and Rustamov, Anar and Sch{\"a}fer, Andreas and Scherer, Stefan and Schindler, Matthias and Schramm, Steven and Shifman, Mikhail and Shuryak, Edward and Sj{\"o}strand, Torbj{\"o}rn and Sterman, George and Stewart, Iain W. and Stroth, Joachim and Swanson, Eric and {de T{\'e}ramond}, Guy F. and Thoma, Ulrike and Vairo, Antonio and {van Dyk}, Danny and Vary, James and Virto, Javier and Vos, Marcel and Weiss, Christian and Wobisch, Markus and Wu, Sau Lan and Young, Christopher and Yuan, Feng and Zhao, Xingbo and Zhou, Xiaorong},
  year = 2023,
  month = dec,
  journal = {The European Physical Journal C},
  volume = {83},
  number = {12},
  pages = {1125},
  issn = {1434-6052},
  doi = {10.1140/epjc/s10052-023-11949-2},
  urldate = {2026-07-31},
  langid = {english}
}

@article{hauschild_tensor_2024,
  title = {Tensor Network {{Python}} ({{TeNPy}}) Version 1},
  author = {Hauschild, Johannes and Unfried, Jakob and Anand, Sajant and Andrews, Bartholomew and Bintz, Marcus and Borla, Umberto and Divic, Stefan and Drescher, Markus and Geiger, Jan and Hefel, Martin and H{\'e}mery, K{\'e}vin and Kadow, Wilhelm and Kemp, Jack and Kirchner, Nico and Liu, Vincent S. and M{\"o}ller, Gunnar and Parker, Daniel and Rader, Michael and Romen, Anton and Scalet, Samuel and Schoonderwoerd, Leon and Schulz, Maximilian and Soejima, Tomohiro and Thoma, Philipp and Wu, Yantao and Zechmann, Philip and Zweng, Ludwig and Mong, Roger S. K. and Zaletel, Michael P. and Pollmann, Frank},
  year = 2024,
  journal = {SciPost Phys. Codebases},
  pages = {41},
  publisher = {SciPost},
  doi = {10.21468/SciPostPhysCodeb.41}
}

@article{hooft_magnetic_1974,
  title = {Magnetic Monopoles in Unified Gauge Theories},
  author = {Hooft, G.'t},
  year = 1974,
  journal = {Nuclear Physics B},
  volume = {79},
  number = {2},
  pages = {276--284},
  issn = {0550-3213},
  doi = {10.1016/0550-3213(74)90486-6}
}

@article{horn_hamiltonian_1979,
  title = {Hamiltonian Approach to \${{Z}}({{N}})\$ Lattice Gauge Theories},
  author = {Horn, D. and Weinstein, M. and Yankielowicz, S.},
  year = 1979,
  month = jun,
  journal = {Physical Review D},
  volume = {19},
  number = {12},
  pages = {3715--3731},
  publisher = {American Physical Society},
  doi = {10.1103/PhysRevD.19.3715},
  urldate = {2026-09-12}
}

@article{hymabaccus_repndecomp_2020,
  title = {{{RepnDecomp}}: {{A GAP}} Package for Decomposing Linear Representations of Finite Groups},
  shorttitle = {{{RepnDecomp}}},
  author = {Hymabaccus, Kaashif and Pasechnik, Dmitrii},
  year = 2020,
  month = jun,
  journal = {Journal of Open Source Software},
  volume = {5},
  number = {50},
  pages = {1835},
  issn = {2475-9066},
  doi = {10.21105/joss.01835},
  urldate = {2026-08-10},
  langid = {english}
}

@misc{inacio_unconstrained_2026,
  title = {Unconstrained Compact Lattice {{QED}}\$\_\textbraceleft 2+1\textbraceright\$ Coupled to Phonons: {{Gauss}} Sectors, Orthogonal Semimetal, and Deconfined Criticality},
  shorttitle = {Unconstrained Compact Lattice {{QED}}\$\_\textbraceleft 2+1\textbraceright\$ Coupled to Phonons},
  author = {In{\'a}cio, Jo{\~a}o C. and Assaad, Fakher F.},
  year = 2026,
  month = sep,
  number = {arXiv:2609.20502},
  eprint = {2609.20502},
  primaryclass = {cond-mat.str-el},
  publisher = {arXiv},
  doi = {10.48550/arXiv.2609.20502},
  urldate = {2026-10-05},
  archiveprefix = {arXiv}
}

@article{jansen_mapping_2026,
  title = {Mapping {{Phase Diagrams}} of {{Quantum Spin Systems}} through {{Semidefinite-Programming Relaxations}}},
  author = {Jansen, David and Farina, Donato and Mortimer, Luke and Heightman, Timothy and Leitherer, Andreas and Mujal, Pere and Wang, Jie and Ac{\'i}n, Antonio},
  year = 2026,
  month = feb,
  journal = {Physical Review Letters},
  volume = {136},
  number = {5},
  pages = {050401},
  issn = {0031-9007, 1079-7114},
  doi = {10.1103/j9rb-tnj4},
  urldate = {2026-04-21},
  langid = {english}
}

@article{jordan_uber_1928,
  title = {{\"Uber das Paulische \"Aquivalenzverbot}},
  author = {Jordan, P. and Wigner, E.},
  year = 1928,
  month = sep,
  journal = {Zeitschrift f\"ur Physik},
  volume = {47},
  number = {9},
  pages = {631--651},
  issn = {0044-3328},
  doi = {10.1007/BF01331938},
  urldate = {2026-08-04},
  langid = {ngerman}
}

@article{kebric_confinement_2021,
  title = {Confinement and {{Mott Transitions}} of {{Dynamical Charges}} in {{One-Dimensional Lattice Gauge Theories}}},
  author = {Kebri{\v c}, Matja{\v z} and Barbiero, Luca and Reinmoser, Christian and Schollw{\"o}ck, Ulrich and Grusdt, Fabian},
  year = 2021,
  month = oct,
  journal = {Physical Review Letters},
  volume = {127},
  number = {16},
  pages = {167203},
  publisher = {American Physical Society},
  doi = {10.1103/PhysRevLett.127.167203},
  urldate = {2026-09-03}
}

@article{kitaev_fault-tolerant_2003,
  title = {Fault-Tolerant Quantum Computation by Anyons},
  author = {Kitaev, {\relax A.Yu}.},
  year = 2003,
  month = jan,
  journal = {Annals of Physics},
  volume = {303},
  number = {1},
  pages = {2--30},
  issn = {00034916},
  doi = {10.1016/S0003-4916(02)00018-0},
  urldate = {2018-05-02}
}

@article{klco_standard_2022,
  title = {Standard Model Physics and the Digital Quantum Revolution: Thoughts about the Interface},
  shorttitle = {Standard Model Physics and the Digital Quantum Revolution},
  author = {Klco, Natalie and Roggero, Alessandro and Savage, Martin J},
  year = 2022,
  month = may,
  journal = {Reports on Progress in Physics},
  volume = {85},
  number = {6},
  pages = {064301},
  publisher = {IOP Publishing},
  issn = {0034-4885},
  doi = {10.1088/1361-6633/ac58a4},
  urldate = {2026-07-01},
  langid = {english}
}

@article{kogut_hamiltonian_1975,
  title = {Hamiltonian Formulation of {{Wilson}}'s Lattice Gauge Theories},
  author = {Kogut, John and Susskind, Leonard},
  year = 1975,
  month = jan,
  journal = {Physical Review D},
  volume = {11},
  number = {2},
  pages = {395--408},
  publisher = {American Physical Society},
  doi = {10.1103/PhysRevD.11.395},
  urldate = {2026-05-08}
}

@article{kogut_introduction_1979,
  title = {An Introduction to Lattice Gauge Theory and Spin Systems},
  author = {Kogut, John B.},
  year = 1979,
  month = oct,
  journal = {Reviews of Modern Physics},
  volume = {51},
  number = {4},
  pages = {659--713},
  publisher = {American Physical Society},
  doi = {10.1103/RevModPhys.51.659},
  urldate = {2026-09-12}
}

@article{lasserre_convergent_2006,
  title = {Convergent {{SDP}}-{{Relaxations}} in {{Polynomial Optimization}} with {{Sparsity}}},
  author = {Lasserre, Jean B.},
  year = 2006,
  month = jan,
  journal = {SIAM Journal on Optimization},
  volume = {17},
  number = {3},
  pages = {822--843},
  publisher = {{Society for Industrial and Applied Mathematics}},
  issn = {1052-6234},
  doi = {10.1137/05064504X},
  urldate = {2023-07-14}
}

@article{lasserre_global_2001,
  title = {Global {{Optimization}} with {{Polynomials}} and the {{Problem}} of {{Moments}}},
  author = {Lasserre, Jean B.},
  year = 2001,
  month = jan,
  journal = {SIAM Journal on Optimization},
  volume = {11},
  number = {3},
  pages = {796--817},
  publisher = {{Society for Industrial and Applied Mathematics}},
  issn = {1052-6234},
  doi = {10.1137/S1052623400366802},
  urldate = {2023-05-04}
}

@article{lawrence_semidefinite_2023,
  title = {Semidefinite Programs at Finite Fermion Density},
  author = {Lawrence, Scott},
  year = 2023,
  month = may,
  journal = {Physical Review D},
  volume = {107},
  number = {9},
  pages = {094511},
  publisher = {American Physical Society},
  doi = {10.1103/PhysRevD.107.094511},
  urldate = {2024-11-07}
}

@article{lee_doping_2006,
  title = {Doping a {{Mott}} Insulator: {{Physics}} of High-Temperature Superconductivity},
  shorttitle = {Doping a {{Mott}} Insulator},
  author = {Lee, Patrick A. and Nagaosa, Naoto and Wen, Xiao-Gang},
  year = 2006,
  month = jan,
  journal = {Reviews of Modern Physics},
  volume = {78},
  number = {1},
  pages = {17--85},
  publisher = {American Physical Society},
  doi = {10.1103/RevModPhys.78.17},
  urldate = {2026-06-30}
}

@misc{li_bootstrapping_2024,
  title = {Bootstrapping the {{Abelian Lattice Gauge Theories}}},
  author = {Li, Zhijin and Zhou, Shutong},
  year = 2024,
  month = apr,
  urldate = {2026-07-29},
  langid = {english}
}

@article{lumia_two-dimensional_2022,
  title = {Two-{{Dimensional}} \$\textbraceleft\textbackslash mathbb\textbraceleft{{Z}}\textbraceright\textbraceright\_\textbraceleft 2\textbraceright\$ {{Lattice Gauge Theory}} on a {{Near-Term Quantum Simulator}}: {{Variational Quantum Optimization}}, {{Confinement}}, and {{Topological Order}}},
  shorttitle = {Two-{{Dimensional}} \$\textbraceleft\textbackslash mathbb\textbraceleft{{Z}}\textbraceright\textbraceright\_\textbraceleft 2\textbraceright\$ {{Lattice Gauge Theory}} on a {{Near-Term Quantum Simulator}}},
  author = {Lumia, Luca and Torta, Pietro and Mbeng, Glen B. and Santoro, Giuseppe E. and Ercolessi, Elisa and Burrello, Michele and Wauters, Matteo M.},
  year = 2022,
  month = apr,
  journal = {PRX Quantum},
  volume = {3},
  number = {2},
  pages = {020320},
  publisher = {American Physical Society},
  doi = {10.1103/PRXQuantum.3.020320},
  urldate = {2026-09-12}
}

@article{magnifico_lattice_2021,
  title = {Lattice Quantum Electrodynamics in (3+1)-Dimensions at Finite Density with Tensor Networks},
  author = {Magnifico, Giuseppe and Felser, Timo and Silvi, Pietro and Montangero, Simone},
  year = 2021,
  month = jun,
  journal = {Nature Communications},
  volume = {12},
  number = {1},
  pages = {3600},
  publisher = {Nature Publishing Group},
  issn = {2041-1723},
  doi = {10.1038/s41467-021-23646-3},
  urldate = {2022-09-09},
  copyright = {2021 The Author(s)},
  langid = {english}
}

@article{magnifico_tensor_2025,
  title = {Tensor Networks for Lattice Gauge Theories beyond One Dimension},
  author = {Magnifico, Giuseppe and Cataldi, Giovanni and Rigobello, Marco and Majcen, Peter and Jaschke, Daniel and Silvi, Pietro and Montangero, Simone},
  year = 2025,
  month = aug,
  journal = {Communications Physics},
  volume = {8},
  number = {1},
  pages = {322},
  publisher = {Nature Publishing Group},
  issn = {2399-3650},
  doi = {10.1038/s42005-025-02125-x},
  urldate = {2026-07-01},
  copyright = {2025 The Author(s)},
  langid = {english}
}

@article{martinez_real-time_2016,
  title = {Real-Time Dynamics of Lattice Gauge Theories with a Few-Qubit Quantum Computer},
  author = {Martinez, Esteban A. and Muschik, Christine A. and Schindler, Philipp and Nigg, Daniel and Erhard, Alexander and Heyl, Markus and Hauke, Philipp and Dalmonte, Marcello and Monz, Thomas and Zoller, Peter and Blatt, Rainer},
  year = 2016,
  month = jun,
  journal = {Nature},
  volume = {534},
  number = {7608},
  pages = {516--519},
  publisher = {Nature Publishing Group},
  issn = {1476-4687},
  doi = {10.1038/nature18318},
  urldate = {2025-01-14},
  copyright = {2016 Springer Nature Limited},
  langid = {english}
}

@article{mazziotti_first-order_2004,
  title = {First-Order Semidefinite Programming for the Direct Determination of Two-Electron Reduced Density Matrices with Application to Many-Electron Atoms and Molecules},
  author = {Mazziotti, David A.},
  year = 2004,
  month = dec,
  journal = {The Journal of Chemical Physics},
  volume = {121},
  number = {22},
  pages = {10957--10966},
  issn = {0021-9606},
  doi = {10.1063/1.1810134},
  urldate = {2026-06-24}
}

@article{mazziotti_two-electron_2012,
  title = {Two-{{Electron Reduced Density Matrix}} as the {{Basic Variable}} in {{Many-Electron Quantum Chemistry}} and {{Physics}}},
  author = {Mazziotti, David A.},
  year = 2012,
  month = jan,
  journal = {Chemical Reviews},
  volume = {112},
  number = {1},
  pages = {244--262},
  doi = {10.1021/cr2000493},
  urldate = {2026-09-03}
}

@article{mildenberger_confinement_2025,
  title = {Confinement in a  ${{\mathbb{Z}}}_{2}$ Lattice Gauge Theory on a Quantum Computer},
  author = {Mildenberger, Julius and Mruczkiewicz, Wojciech and Halimeh, Jad C. and Jiang, Zhang and Hauke, Philipp},
  year = 2025,
  month = feb,
  journal = {Nature Physics},
  volume = {21},
  number = {2},
  pages = {312--317},
  issn = {1745-2473, 1745-2481},
  doi = {10.1038/s41567-024-02723-6},
  urldate = {2025-02-20},
  langid = {english}
}

@misc{mortimer_bounding_2026,
  title = {Bounding Many-Body Properties under Partial Information and Finite Measurement Statistics},
  author = {Mortimer, Luke and Zambrano, Leonardo and Ac{\'i}n, Antonio and Farina, Donato},
  year = 2026,
  month = jan,
  number = {arXiv:2601.10408},
  eprint = {2601.10408},
  primaryclass = {quant-ph},
  publisher = {arXiv},
  doi = {10.48550/arXiv.2601.10408},
  urldate = {2026-06-17},
  archiveprefix = {arXiv}
}

@article{mortimer_certifying_2025,
  title = {Certifying Steady-State Properties of Open Quantum Systems},
  author = {Mortimer, Luke and Farina, Donato and Di Bello, Grazia and Jansen, David and Leitherer, Andreas and Mujal, Pere and Ac{\'i}n, Antonio},
  year = 2025,
  month = sep,
  journal = {Physical Review Research},
  volume = {7},
  number = {3},
  pages = {033237},
  publisher = {American Physical Society},
  doi = {10.1103/hbrt-cn8q},
  urldate = {2025-10-29}
}

@article{nakata_variational_2001,
  title = {Variational Calculations of Fermion Second-Order Reduced Density Matrices by Semidefinite Programming Algorithm},
  author = {Nakata, Maho and Nakatsuji, Hiroshi and Ehara, Masahiro and Fukuda, Mituhiro and Nakata, Kazuhide and Fujisawa, Katsuki},
  year = 2001,
  month = may,
  journal = {The Journal of Chemical Physics},
  volume = {114},
  number = {19},
  pages = {8282--8292},
  doi = {10.1063/1.1360199},
  urldate = {2026-09-03}
}

@article{nambu_strings_1974,
  title = {Strings, Monopoles, and Gauge Fields},
  author = {Nambu, Y.},
  year = 1974,
  month = dec,
  journal = {Physical Review D},
  volume = {10},
  number = {12},
  pages = {4262--4268},
  publisher = {American Physical Society},
  doi = {10.1103/PhysRevD.10.4262},
  urldate = {2026-07-31}
}

@article{navascues_bounding_2007,
  title = {Bounding the {{Set}} of {{Quantum Correlations}}},
  author = {Navascu{\'e}s, Miguel and Pironio, Stefano and Ac{\'i}n, Antonio},
  year = 2007,
  month = jan,
  journal = {Physical Review Letters},
  volume = {98},
  number = {1},
  pages = {010401},
  issn = {0031-9007, 1079-7114},
  doi = {10.1103/PhysRevLett.98.010401},
  urldate = {2026-05-08},
  copyright = {http://link.aps.org/licenses/aps-default-license},
  langid = {english}
}

@article{navascues_convergent_2008,
  title = {A Convergent Hierarchy of Semidefinite Programs Characterizing the Set of Quantum Correlations},
  author = {Navascu{\'e}s, Miguel and Pironio, Stefano and Ac{\'i}n, Antonio},
  year = 2008,
  month = jul,
  journal = {New Journal of Physics},
  volume = {10},
  number = {7},
  pages = {073013},
  issn = {1367-2630},
  doi = {10.1088/1367-2630/10/7/073013},
  urldate = {2023-05-04},
  langid = {english}
}

@book{nesterov_lectures_2018,
  title = {Lectures on {{Convex Optimization}}},
  author = {Nesterov, Yurii},
  year = 2018,
  series = {Springer {{Optimization}} and {{Its Applications}}},
  volume = {137},
  publisher = {Springer International Publishing},
  address = {Cham},
  doi = {10.1007/978-3-319-91578-4},
  urldate = {2026-05-05},
  copyright = {http://www.springer.com/tdm},
  isbn = {978-3-319-91577-7 978-3-319-91578-4},
  langid = {english}
}

@misc{noauthor_mosek_nodate,
  title = {Mosek {{ApS}}},
  urldate = {2026-05-05},
  howpublished = {https://www.mosek.com/}
}

@book{peskin_introduction_1995,
  title = {An Introduction to Quantum Field Theory},
  author = {Peskin, Michael Edward and Schroeder, Daniel V.},
  year = 1995,
  publisher = {Addison-Wesley Pub. Co},
  address = {Reading, Mass},
  isbn = {978-0-201-50397-5}
}

@article{pironio_convergent_2010,
  title = {Convergent {{Relaxations}} of {{Polynomial Optimization Problems}} with {{Noncommuting Variables}}},
  author = {Pironio, S. and Navascu{\'e}s, M. and Ac{\'i}n, A.},
  year = 2010,
  month = jan,
  journal = {SIAM Journal on Optimization},
  volume = {20},
  number = {5},
  pages = {2157--2180},
  publisher = {{Society for Industrial and Applied Mathematics}},
  issn = {1052-6234},
  doi = {10.1137/090760155},
  urldate = {2026-06-30}
}

@article{polyakov_compact_1975,
  title = {Compact Gauge Fields and the Infrared Catastrophe},
  author = {Polyakov, A. M.},
  year = 1975,
  month = oct,
  journal = {Physics Letters B},
  volume = {59},
  number = {1},
  pages = {82--84},
  issn = {0370-2693},
  doi = {10.1016/0370-2693(75)90162-8},
  urldate = {2026-07-31}
}

@article{polyakov_quark_1977,
  title = {Quark Confinement and Topology of Gauge Theories},
  author = {Polyakov, A. M.},
  year = 1977,
  month = mar,
  journal = {Nuclear Physics B},
  volume = {120},
  number = {3},
  pages = {429--458},
  issn = {0550-3213},
  doi = {10.1016/0550-3213(77)90086-4},
  urldate = {2026-07-31}
}

@article{rico_tensor_2014,
  title = {Tensor {{Networks}} for {{Lattice Gauge Theories}} and {{Atomic Quantum Simulation}}},
  author = {Rico, E. and Pichler, T. and Dalmonte, M. and Zoller, P. and Montangero, S.},
  year = 2014,
  month = may,
  journal = {Physical Review Letters},
  volume = {112},
  number = {20},
  pages = {201601},
  publisher = {American Physical Society},
  doi = {10.1103/PhysRevLett.112.201601},
  urldate = {2026-07-01}
}

@inproceedings{rosset_replab_2021,
  title = {{{RepLAB}}: {{A Computational}}/{{Numerical Approach}} to {{Representation Theory}}},
  shorttitle = {{{RepLAB}}},
  booktitle = {Quantum {{Theory}} and {{Symmetries}}},
  author = {Rosset, Denis and {Montealegre-Mora}, Felipe and Bancal, Jean-Daniel},
  editor = {Paranjape, M. B. and MacKenzie, Richard and Thomova, Zora and Winternitz, Pavel and {Witczak-Krempa}, William},
  year = 2021,
  pages = {643--653},
  publisher = {Springer International Publishing},
  address = {Cham},
  doi = {10.1007/978-3-030-55777-5_60},
  isbn = {978-3-030-55777-5},
  langid = {english}
}

@article{savary_quantum_2016,
  title = {Quantum Spin Liquids: A Review},
  author = {Savary, Lucile and Balents, Leon},
  year = 2016,
  month = nov,
  journal = {Reports on Progress in Physics},
  volume = {80},
  number = {1},
  pages = {016502},
  doi = {10.1088/0034-4885/80/1/016502}
}

@article{scheer_bootstrapping_2026,
  title = {Bootstrapping {{Quantum Hamiltonians}} with {{Symmetry}}},
  author = {Scheer, Michael G.},
  year = 2026,
  month = jun,
  journal = {Physical Review Letters},
  volume = {136},
  number = {24},
  eprint = {2410.00810},
  pages = {246501},
  doi = {10.1103/1757-fx9c},
  urldate = {2026-09-03}
}

@article{schollwock_density-matrix_2011,
  title = {The Density-Matrix Renormalization Group in the Age of Matrix Product States},
  author = {Schollw{\"o}ck, Ulrich},
  year = 2011,
  month = jan,
  journal = {Annals of Physics},
  series = {January 2011 {{Special Issue}}},
  volume = {326},
  number = {1},
  pages = {96--192},
  issn = {0003-4916},
  doi = {10.1016/j.aop.2010.09.012},
  urldate = {2026-07-01}
}

@article{schweizer_floquet_2019,
  title = {Floquet Approach to {{$\mathbb{Z}$2}} Lattice Gauge Theories with Ultracold Atoms in Optical Lattices},
  author = {Schweizer, Christian and Grusdt, Fabian and Berngruber, Moritz and Barbiero, Luca and Demler, Eugene and Goldman, Nathan and Bloch, Immanuel and Aidelsburger, Monika},
  year = 2019,
  month = nov,
  journal = {Nature Physics},
  volume = {15},
  number = {11},
  pages = {1168--1173},
  publisher = {Nature Publishing Group},
  issn = {1745-2481},
  doi = {10.1038/s41567-019-0649-7},
  urldate = {2026-07-01},
  copyright = {2019 The Author(s), under exclusive licence to Springer Nature Limited},
  langid = {english}
}

@book{serre_linear_1977,
  title = {Linear {{Representations}} of {{Finite Groups}}},
  author = {Serre, Jean-Pierre},
  year = 1977,
  series = {Graduate {{Texts}} in {{Mathematics}}},
  volume = {42},
  publisher = {Springer},
  address = {New York, NY},
  doi = {10.1007/978-1-4684-9458-7},
  urldate = {2024-12-06},
  copyright = {http://www.springer.com/tdm},
  isbn = {978-1-4684-9460-0 978-1-4684-9458-7}
}

@article{silvi_tensor_2019,
  title = {Tensor Network Simulation of an {{SU}}(3) Lattice Gauge Theory in {{1D}}},
  author = {Silvi, Pietro and Sauer, Yannick and Tschirsich, Ferdinand and Montangero, Simone},
  year = 2019,
  month = oct,
  journal = {Physical Review D},
  volume = {100},
  number = {7},
  pages = {074512},
  publisher = {American Physical Society},
  doi = {10.1103/PhysRevD.100.074512},
  urldate = {2026-07-01}
}

@article{spagnoli_fault-tolerant_2026,
  title = {Fault-Tolerant Simulation of {{Lattice Gauge Theories}} with Gauge Covariant Codes},
  author = {Spagnoli, L. and Roggero, A. and Wiebe, N.},
  year = 2026,
  month = jan,
  journal = {Quantum},
  volume = {10},
  pages = {1968},
  publisher = {Verein zur F\"orderung des Open Access Publizierens in den Quantenwissenschaften},
  doi = {10.22331/q-2026-01-16-1968},
  urldate = {2026-08-07},
  langid = {british}
}

@article{troyer_computational_2005,
  title = {Computational {{Complexity}} and {{Fundamental Limitations}} to {{Fermionic Quantum Monte Carlo Simulations}}},
  author = {Troyer, Matthias and Wiese, Uwe-Jens},
  year = 2005,
  month = may,
  journal = {Physical Review Letters},
  volume = {94},
  number = {17},
  pages = {170201},
  publisher = {American Physical Society},
  doi = {10.1103/PhysRevLett.94.170201},
  urldate = {2026-05-08}
}

@misc{vallee_data_2026,
  title = {Data Supporting: "{{Two-Sided Bounds}} on {{Ground State Properties}} in {{Lattice Gauge Theories}}"},
  author = {Vall{\'e}e, Elo{\"i}c and Kelman, Ariel and Mildenberger, Julius and Tura, Jordi and Emonts, Patrick},
  year = 2026,
  month = sep,
  publisher = {Zenodo},
  doi = {10.5281/zenodo.22872270}
}

@article{wang_certifying_2024,
  title = {Certifying {{Ground-State Properties}} of {{Many-Body Systems}}},
  author = {Wang, Jie and Surace, Jacopo and Fr{\'e}rot, Ir{\'e}n{\'e}e and Legat, Beno{\^i}t and Renou, Marc-Olivier and Magron, Victor and Ac{\'i}n, Antonio},
  year = 2024,
  month = jul,
  journal = {Physical Review X},
  volume = {14},
  number = {3},
  pages = {031006},
  issn = {2160-3308},
  doi = {10.1103/PhysRevX.14.031006},
  urldate = {2026-08-11},
  langid = {english}
}

@misc{wang_scalable_2026,
  title = {Scalable {{Ground-State Certification}} of {{Quantum Spin Systems}} via {{Structured Noncommutative Polynomial Optimization}}},
  author = {Wang, Jie and Jansen, David and Frerot, Ir{\'e}n{\'e}e and Renou, Marc-Olivier and Magron, Victor and Ac{\'i}n, Antonio},
  year = 2026,
  month = apr,
  number = {arXiv:2604.01555},
  eprint = {2604.01555},
  primaryclass = {quant-ph},
  publisher = {arXiv},
  doi = {10.48550/arXiv.2604.01555},
  urldate = {2026-04-13},
  archiveprefix = {arXiv}
}

@article{weinberg_making_2004,
  title = {The {{Making}} of the Standard Model},
  author = {Weinberg, Steven},
  editor = {{'t Hooft}, G.},
  year = 2004,
  journal = {The European Physical Journal C: Particles and Fields},
  volume = {34},
  eprint = {hep-ph/0401010},
  pages = {5--13},
  doi = {10.1140/epjc/s2004-01761-1},
  archiveprefix = {arXiv}
}

@book{weinberg_quantum_1995,
  title = {The Quantum Theory of Fields},
  author = {Weinberg, Steven},
  year = 1995,
  month = jun,
  publisher = {Cambridge University Press},
  googlebooks = {doeDB3\_WLvwC},
  isbn = {978-0-521-55001-7},
  langid = {english}
}

@article{wen_mean-field_1991,
  title = {Mean-Field Theory of Spin-Liquid States with Finite Energy Gap and Topological Orders},
  author = {Wen, X. G.},
  year = 1991,
  month = aug,
  journal = {Physical Review B},
  volume = {44},
  number = {6},
  pages = {2664--2672},
  publisher = {American Physical Society},
  doi = {10.1103/PhysRevB.44.2664},
  urldate = {2026-06-30}
}

@book{wen_quantum_2004,
  title = {Quantum Field Theory of Many-Body Systems:From the Origin of Sound to an Origin of Light and Electrons: {{From}} the Origin of Sound to an Origin of Light and Electrons},
  author = {Wen, X.G.},
  year = 2004,
  series = {Oxford Graduate Texts},
  publisher = {OUP Oxford},
  isbn = {978-0-19-853094-7},
  lccn = {2004301677}
}

@article{white_density_1992,
  ids = {white_density_1992-1},
  title = {Density Matrix Formulation for Quantum Renormalization Groups},
  author = {White, Steven R.},
  year = 1992,
  month = nov,
  journal = {Physical Review Letters},
  volume = {69},
  number = {19},
  pages = {2863--2866},
  publisher = {American Physical Society},
  issn = {0031-9007},
  doi = {10.1103/PhysRevLett.69.2863},
  urldate = {2018-03-06}
}

@article{wilson_confinement_1974,
  title = {Confinement of Quarks},
  author = {Wilson, Kenneth G.},
  year = 1974,
  month = oct,
  journal = {Physical Review D},
  volume = {10},
  number = {8},
  pages = {2445--2459},
  publisher = {American Physical Society},
  doi = {10.1103/PhysRevD.10.2445},
  urldate = {2026-05-08}
}

@misc{wu_accurate_2025,
  title = {Accurate {{Gauge-Invariant Tensor Network Simulations}} for {{Abelian Lattice Gauge Theory}} in (2+1){{D}}: Ground State and Real-Time Dynamics},
  shorttitle = {Accurate {{Gauge-Invariant Tensor Network Simulations}} for {{Abelian Lattice Gauge Theory}} in (2+1){{D}}},
  author = {Wu, Yantao and Liu, Wen-Yuan},
  year = 2025,
  month = aug,
  number = {arXiv:2503.20566},
  eprint = {2503.20566},
  primaryclass = {cond-mat.str-el},
  publisher = {arXiv},
  doi = {10.48550/arXiv.2503.20566},
  urldate = {2026-09-03},
  archiveprefix = {arXiv}
}

@misc{xu_string_2025,
  title = {String {{Breaking Dynamics}} and {{Glueball Formation}} in a 2+{{1D Lattice Gauge Theory}}},
  author = {Xu, Kaidi and Borla, Umberto and Moroz, Sergej and Halimeh, Jad C.},
  year = 2025,
  month = jul,
  publisher = {arXiv},
  doi = {10.48550/arXiv.2507.01950},
  urldate = {2026-07-01}
}

@article{zhao_reduced_2004,
  title = {The Reduced Density Matrix Method for Electronic Structure Calculations and the Role of Three-Index Representability Conditions},
  author = {Zhao, Zhengji and Braams, Bastiaan J. and Fukuda, Mituhiro and Overton, Michael L. and Percus, Jerome K.},
  year = 2004,
  month = feb,
  journal = {The Journal of Chemical Physics},
  volume = {120},
  number = {5},
  pages = {2095--2104},
  doi = {10.1063/1.1636721},
  urldate = {2026-09-03}
}

@article{zohar_local_2020,
  title = {Local Manipulation and Measurement of Nonlocal Many-Body Operators in Lattice Gauge Theory Quantum Simulators},
  author = {Zohar, Erez},
  year = 2020,
  month = feb,
  journal = {Physical Review D},
  volume = {101},
  number = {3},
  pages = {034518},
  publisher = {American Physical Society},
  doi = {10.1103/PhysRevD.101.034518},
  urldate = {2026-09-03}
}

@article{zohar_quantum_2015,
  title = {Quantum Simulations of Lattice Gauge Theories Using Ultracold Atoms in Optical Lattices},
  author = {Zohar, Erez and Cirac, J Ignacio and Reznik, Benni},
  year = 2015,
  month = dec,
  journal = {Reports on Progress in Physics},
  volume = {79},
  number = {1},
  pages = {014401},
  publisher = {IOP Publishing},
  issn = {0034-4885},
  doi = {10.1088/0034-4885/79/1/014401},
  urldate = {2026-07-01},
  langid = {english}
}

@article{zohar_quantum_2021,
  title = {Quantum Simulation of Lattice Gauge Theories in More than One Space Dimension---Requirements, Challenges and Methods},
  author = {Zohar, Erez},
  year = 2021,
  month = dec,
  journal = {Philosophical Transactions of the Royal Society A: Mathematical, Physical and Engineering Sciences},
  volume = {380},
  number = {2216},
  pages = {20210069},
  issn = {1364-503X},
  doi = {10.1098/rsta.2021.0069},
  urldate = {2026-07-01}
}

@article{zohar_wilson_2021,
  title = {Wilson Loops and Area Laws in Lattice Gauge Theory Tensor Networks},
  author = {Zohar, Erez},
  year = 2021,
  month = aug,
  journal = {Physical Review Research},
  volume = {3},
  number = {3},
  pages = {033179},
  publisher = {American Physical Society},
  doi = {10.1103/PhysRevResearch.3.033179},
  urldate = {2026-09-12}
}

\end{document}